\documentclass[letterpaper,11pt]{article}
\usepackage[margin=1in]{geometry}
\usepackage[T1]{fontenc}

\usepackage{amsmath,amsthm,amssymb}
\usepackage{mathrsfs}
\usepackage{xspace}
\usepackage{url}
\usepackage{hyperref}
\usepackage{appendix}
\usepackage[capitalize,noabbrev]{cleveref}
\usepackage{multirow}
\usepackage{algorithm}
\usepackage{algpseudocode}
\usepackage{lipsum}
\usepackage{booktabs}
\usepackage{thm-restate}
\usepackage{graphicx}
\usepackage{subcaption}
\usepackage{tabularx}
\usepackage{comment}

\hypersetup{colorlinks=true,allcolors=blue}

\theoremstyle{plain}

\newtheorem{theorem}{Theorem}[section]
\newtheorem{lemma}[theorem]{Lemma}
\newtheorem{claim}[theorem]{Claim}
\newtheorem{corollary}[theorem]{Corollary}

\newtheorem{fact}[theorem]{Fact}

\theoremstyle{definition}
\newtheorem{definition}[theorem]{Definition}

\theoremstyle{remark}
\newtheorem{remark}[theorem]{Remark}
\newtheorem*{remark*}{Remark}

\DeclareMathOperator{\R}{\mathbb{R}}

\DeclareMathOperator{\poly}{poly}
\DeclareMathOperator{\polylog}{polylog}
\DeclareMathOperator{\cost}{cost}

\DeclareMathOperator{\OPT}{OPT}
\DeclareMathOperator{\ddim}{ddim}
\DeclareMathOperator{\tw}{tw}

\DeclareMathOperator{\dist}{dist}

\newcommand{\diam}{\mathrm{diam}}

\newcommand{\dup}{\mathrm{dup}}
\newcommand{\multi}{\mathrm{dup}}

\newcommand{\inl}{\mathrm{in}}
\newcommand{\far}{\mathrm{far}}
\newcommand{\close}{\mathrm{close}}
\newcommand{\out}{\mathrm{out}}

\newcommand{\GOPT}{\widehat{\OPT}}

\newcommand{\eps}{\varepsilon}
\renewcommand{\epsilon}{\varepsilon}
\renewcommand{\phi}{\varphi}

\newcommand{\calA}{\mathcal{A}}
\newcommand{\calB}{\mathcal{B}}
\newcommand{\calC}{\mathcal{C}}
\newcommand{\calD}{\mathcal{D}}
\newcommand{\calE}{\mathcal{E}}

\newcommand{\calL}{\mathcal{L}}
\newcommand{\calM}{\mathcal{M}}

\newcommand{\calP}{\mathcal{P}}
\newcommand{\calQ}{\mathcal{Q}}

\newcommand{\calS}{\mathcal{S}}
\newcommand{\calT}{\mathcal{T}}

\newcommand{\calV}{\mathcal{V}}

\newcommand{\supp}{\mathrm{supp}}
\newcommand{\aux}{\mathrm{aux}}
\newcommand{\ball}{\mathrm{Ball}}

\newcommand{\ProblemName}[1]{\textsc{#1}}

\newcommand{\kMedian}{\ProblemName{$k$-Median}\xspace}
\newcommand{\kMeans}{\ProblemName{$k$-Means}\xspace}

\newcommand{\kzC}{\ProblemName{$(k,z)$-Clustering}\xspace}
\newcommand{\tzC}[1]{\ProblemName{$(#1,z)$-Clustering}\xspace}
\newcommand{\tzmC}[1]{\ProblemName{$(#1,m,z)$-Clustering}\xspace}

\newcommand{\kzmC}{\ProblemName{$(k,z,m)$-Clustering}\xspace}

\newcommand{\kRMedian}{\ProblemName{$(k,m)$-Median}\xspace}

\usepackage[obeyFinal,textsize=footnotesize]{todonotes}

\title{Coresets for Robust Clustering via Black-box Reductions to Vanilla Case}
\author{Anonymous Authors}
\date{}

\author{
     Shaofeng H.-C. Jiang\thanks{
     Email: \texttt{shaofeng.jiang@pku.edu.cn}
     }\\
     Peking University
     \and
     Jianing Lou
     \thanks{Email: \texttt{loujn@pku.edu.cn}}\\
     Peking University
 }

\begin{document}

\maketitle

\begin{abstract}

We devise $\varepsilon$-coresets for robust \kzC with $m$ outliers through black-box reductions to vanilla clustering. 
Given an $\varepsilon$-coreset construction for vanilla clustering with size $N$, 
we construct coresets of size 
$N\cdot \poly\log(km\eps^{-1}) + O_z\left(\min\{ km\varepsilon^{-1}, m \varepsilon^{-2z}\log^z(km\eps^{-1}) \}\right)$
for various metric spaces, where $O_z$ hides $2^{O(z\log z)}$ factors.
This increases the size of the vanilla coreset by a small multiplicative factor of $\poly\log(km\eps^{-1})$, and the additive term is up to a $(\varepsilon^{-1}\log (km))^{O(z)}$ factor to the size of the optimal robust coreset.
Plugging in recent vanilla coreset results of
[Cohen-Addad, Saulpic and Schwiegelshohn, STOC'21; 
Cohen-Addad, Draganov, Russo, Saulpic and Schwiegelshohn, SODA'25], 
we obtain the first coresets for \kzC with $m$ outliers with size near-linear in $k$ while previous results have size at least $\Omega(k^2)$
[Huang, Jiang, Lou and Wu, ICLR'23; 
Huang, Li, Lu and Wu, SODA'25].

Technically, we establish two conditions under which a vanilla coreset is as well a robust coreset.
The first condition requires the dataset to satisfy special structures -- it can be broken into ``dense'' parts with bounded diameter.
We combine this with a new bounded-diameter decomposition that has only $O_z(km \varepsilon^{-1})$ non-dense points to obtain the $O_z(km \varepsilon^{-1})$ additive bound. 
Another sufficient condition requires the vanilla coreset to possess an extra size-preserving property.
To utilize this condition, we further give a black-box reduction that turns a vanilla coreset to the one that satisfies the said size-preserving property,
and this leads to the alternative $O_z(m\varepsilon^{-2z}\log^{z}(km\eps^{-1}))$ additive size bound.

We also give low-space implementations of our reductions in the \emph{dynamic} streaming setting.
Combined with known streaming constructions for vanilla coresets
[Braverman, Frahling, Lang, Sohler and Yang, ICML'17; 
Hu, Song, Yang and Zhong, arXiv'1802.00459],
we obtain the first dynamic streaming algorithms for coresets for \kMedian (and \kMeans) with $m$ outliers,
using space $\tilde{O}(k + m) \cdot \mathrm{poly}(d\varepsilon^{-1}\log \Delta)$
for inputs on a discrete grid $[\Delta]^d$.

\end{abstract}

\section{Introduction}

\kzC is a fundamental problem that is well studied in both computer science and related areas such as operations research.
Given a metric space $(V, \dist)$ and a dataset $X \subseteq V$,
\kzC aims to find a center set $C \subseteq V$ of $k$ points, such that the clustering objective $\cost_z(X, C)$ is minimized, i.e.,
\[
    \cost_z(X, C) := \sum_{x \in X} (\dist(x, C))^z,
\]
where $\dist(x, C) := \min_{c \in C} \dist(x, c)$.
This formulation generally captures several well known variants of $k$-clustering,
including \kMedian (when $z = 1$) and \kMeans (when $z = 2$).

Unfortunately, datasets can often be noisy, and the objective of \kzC is very sensitive to the noise.
In fact, even adding a single \emph{outlier} point that is distant to every other point could significantly bias the clustering centers towards the outlier point. 
Hence, robust variants of clustering are much desired in order to combat this issue.
We consider a natural formulation suggested by~\cite{DBLP:conf/soda/CharikarKMN01},
called \kzC clustering with $m$ outliers, \kzmC for short,
where a parameter $m$ is introduced to denote the number of outliers.
The new cost function is denoted as $\cost_z^{(m)}(X, C)$, and it is evaluated
by first removing the $m$ points that are furthest to $C$ from $X$,
denoting the resultant points as $X'$, and compute $\cost_z(X', C)$ (i.e., using the vanilla \kzC objective).
Equivalently, 
\[
    \cost_z^{(m)}(X,C) := \min_{Y \in \binom{X}{m}} \cost_z(X \setminus Y,C).
\]

Coreset~\cite{DBLP:conf/stoc/Har-PeledM04} is a powerful technique for obtaining efficient algorithms for clustering.
Roughly speaking, an $\eps$-coreset is a tiny proxy of the dataset that approximates the clustering objective within $(1 \pm \eps)$ factor for every potential center set $C$.
Coresets are not only useful for obtaining fast algorithms,
but can also be converted into streaming~\cite{DBLP:conf/stoc/Har-PeledM04},
distributed~\cite{DBLP:conf/nips/BalcanEL13} and fully-dynamic algorithms~\cite{DBLP:conf/esa/HenzingerK20} via merge-and-reduce framework~\cite{DBLP:conf/stoc/Har-PeledM04}.
Tremendous progress has been made on finding small coresets for (vanilla) \kzC.
Take \kMedian ($z = 1$) in Euclidean $\mathbb{R}^d$ (i.e., $V = \mathbb{R}^d, \dist = \ell_2$) for example,
the initial coreset size $O(k\eps^{-d}\log n)$~\cite{DBLP:conf/stoc/Har-PeledM04} has been improved in a series of works~\cite{DBLP:journals/dcg/Har-PeledK07,Feldman11unified,Sohler18Strong,FeldmanSS20Turning,DBLP:conf/stoc/HuangV20,BJKW21,Cohen-addad2021New,Cohen-Addad22Towards},
all the way to $\tilde{O}(k \epsilon^{-3})$ via a novel sampling framework~\cite{Cohen-addad2021New,Cohen-Addad22Towards},
and this nearly matches a lower bound of $\Omega(k \epsilon^{-2})$~\cite{Cohen-Addad22Towards}. (Alternatively, in the regime of $\epsilon^{-1} \gg k$, a better bound of $\tilde{O}(k^{4/3} \epsilon^{-2})$ may be obtained~\cite{DBLP:conf/nips/Cohen-AddadLSSS22,Huang22OnOptimal}.)

The sampling framework proposed by~\cite{Cohen-addad2021New}
is modified to handle several variants of clustering~\cite{Braverman22Power},
and this modified framework is recently adapted to obtain an $O(m) + \tilde{O}(k^3 \eps^{-5})$ 
size coreset for \kMedian with $m$ outliers~\cite{Huang2022Near-optimal},
which exponentially improves a decade-old $(k+m)^{O(k + m)} \poly(d \eps^{-1} \log n)$ bound~\cite{feldman2012data}.
This bound is further improved to $O(m) + \tilde{O}(k^2 \eps^{-4})$ in a recent work~\cite{Huang2023General}.
However, even this improved bound is unlikely to be tight;
Specifically, there is a sharp transition 
from $m = 0$ (i.e., the vanilla case) to $m = 1$ (i.e., only one outlier is considered) that increases the size by at least a $k \epsilon^{-1}$ factor (compared with~\cite{Cohen-Addad22Towards}).
It is unclear if this transition is fundamental and whether or not it can be avoided.
Technically, existing approaches for robust coresets are built on frameworks designed for the vanilla case and is adapted in an ad-hoc way.
In case improved bounds are obtained for the vanilla case,
it is still technically nontrivial to adapt it to the robust case.

In this paper, we systematically address these issues
by proposing a new framework via black-box reductions to the vanilla case.
Specifically, we wish to start from any vanilla coreset algorithm (which may be the optimal one),
and figure out how to add only a few points to convert it into a robust coreset,
ideally to obtain a size bound that is very close to the vanilla case.
Indeed, this study helps to fundamentally understand ``the price of robustness'' for coresets, i.e., the exact gap in size bounds between vanilla clustering and robust clustering.
Since the coreset construction is via reduction,
it also opens the door to more applications,
such as coreset algorithms for robust clustering in sublinear models
(as long as the vanilla version can be constructed in the corresponding model).

\subsection{Our Results}

Our main result, stated in \Cref{thm:informal}, is a novel coreset construction for \kzmC via a reduction to the vanilla case.
Crucially, this reduction is \emph{black-box} style which works with any vanilla coreset construction without knowing its implementation detail.
As mentioned, this type of bound helps to fundamentally understand the ``price of robustness'' for coresets.
Our result works for general metric spaces, and we choose to state the result for the Euclidean case which is arguably the most natural setting for \kzC.
Overall, this size bound has linear dependence in $N$ (albeit $N$ is evaluated on slightly larger parameters), and the main increase in size is additive.

\begin{theorem}[Euclidean case]
    \label{thm:informal}
    Assume there is an algorithm that constructs an $\eps$-coreset for \kzC of size $N(d,k,\eps^{-1})$ for any dataset from $\R^d$.
Then, there is an algorithm that 
constructs an $\eps$-coreset for \kzmC of size 
    \begin{equation}
        \label{eq:intro_size}
        \min\left\{N(d,k,O(\eps^{-1})) + A_1, N\left(O(d),O(k\log^2(km\eps^{-1})),O(\eps^{-1})\right) + A_2\right\}
    \end{equation}
    for any dataset from $\R^d$, where $A_1 = O_z\left(km\eps^{-1}\right)$ and $A_2 = O_z\left( m\eps^{-2z} \log^z(km\eps^{-1})\right)$.
The two bounds in the $\min$ follow from~\Cref{thm:reduction I} and~\Cref{thm:reduction2}, respectively.
\end{theorem}

As mentioned, our result in \Cref{thm:informal} is applicable to various other metric spaces besides Euclidean,
including doubling metrics, general finite metrics and shortest-path metrics for graphs that exclude a fixed minor.
The size bounds only needs to change the parameters of $N$ according to the metric and does not change the additive terms $A_1$ and $A_2$.
For instance, for doubling metrics, our bound simply replaces the $d$ in~\eqref{eq:intro_size} by the doubling dimension of the metric space (up to constant factor).
Detailed size bounds can be found in \Cref{sec:application}.

We start with justifying the tightness of our size bounds without plugging in any concrete vanilla coreset bound $N$. Observe that~\Cref{thm:informal} actually provides two size bounds, each corresponding to one of the two terms in the $\min$ of~\eqref{eq:intro_size}.
The first bound of $A_1 + N(d,k,O(\eps^{-1}))$ is more useful when $m$ is small.
In particular, it makes our result the first to achieve a smooth transition from $m = 0$ (vanilla case) to $m > 0$ (robust case) in an asymptotic sense, even when plugging in the optimal vanilla coreset.
Indeed, the $O_z(km \epsilon^{-1})$ bound becomes $O_z(k \epsilon^{-1})$ when $m = O(1)$,
and this is asymptotically dominated by a lower bound of $\Omega(k \epsilon^{-2})$ for vanilla coreset as shown in~\cite{Cohen-Addad22Towards}.
For the second bound, it needs to use $k$ that is $\poly\log(km\epsilon^{-1})$ larger in $N$, but under the typical case of $N(d,k,\eps^{-1}) = \poly(dk\eps^{-1})$, this only increases $N$ by a factor of $\poly\log(km\eps^{-1})$ which is minor.
Regarding its additive term $A_2$, it is actually up to only a $(\epsilon^{-1}\log (km))^{O(z)}$ factor to the optimal robust coresets,
which is nearly tight with respect to $m$
due to a lower bound of $\Omega(m)$ for robust coreset~\cite{Huang2022Near-optimal}.

\paragraph{Specific Coreset Size Bounds} 
Since our reduction is black box, the large body of results on coresets for vanilla clustering can be readily applied.
We start with listing in \Cref{tab:result} the concrete coreset bounds
obtained by plugging in recent results for \kzC~\cite{Cohen-Addad22Towards,DBLP:conf/nips/Cohen-AddadLSSS22,Huang22OnOptimal,Cohen-AddadD0SS25}.
These lead to the first coresets of near-linear size in $k$ for \kzmC in all the metric spaces listed,
thereby improving the previous $k^2$ dependency~\cite{feldman2012data,Huang2022Near-optimal,Huang2023General}.
We can also obtain improved coresets when plugging in results designed for specific parameters.
For instance, in low-dimensional Euclidean space, a coreset for $\tzC{1}$ of size $2^{O(z \log z)} \cdot \tilde{O}(\sqrt{d}\eps^{-1})$ was given by~\cite{DBLP:conf/icml/HuangHH023}.
Applying this result to \Cref{thm:informal}, we obtain a coreset for $\tzmC{1}$ of size $2^{O(z \log z)} \cdot \tilde{O}((m + \sqrt{d})\eps^{-1})$.

\begin{table}[ht]
    \centering
    \caption{\small List of relevant coreset bounds under various metric spaces.
    In all the bounds reported in the table we omit a minor multiplicative factor of $(z\log(km\eps^{-1}))^{O(z)}$.
    For our result listed in the last row, the exact form (without omitting the abovementioned minor factor) is
    $\min\{ \mathbf{N} + O_z(km\epsilon^{-1}), \mathbf{N} \cdot \polylog(km\epsilon^{-1}) + O_z(m\epsilon^{-2z}\log^z(km\eps^{-1}))  \}$,
    noting that the $\mathbf{N}$ in the second term of $\min\{\cdot,\cdot\}$ has a $\polylog(km\epsilon^{-1})$ factor.
    }
    \label{tab:result}
    \small
    \begin{tabularx}{\textwidth}{llll}
        \toprule
        metric space $M$ & problem & size & reference \\
        \midrule
        Euclidean $\mathbb{R}^d$ & vanilla ($\mathbf{N}$) &  $
        kd\eps^{-\max\{2,z\}}$ & \cite{Cohen-addad2021New} \\
        &vanilla ($\mathbf{N}$)& $k\eps^{-2-z}$ & \cite{Cohen-Addad22Towards}\\
        &vanilla ($\mathbf{N}$)& $k^{\frac{2z+2}{z+2}}\eps^{-2}$ & \cite{Huang22OnOptimal}\\
& robust, $z = 1$ & $(k+m)^{k+m}(\eps^{-1}d\log n)^2$ & \cite{feldman2012data} \\
         & robust & $m + k^3\eps^{-3z-2}$ & \cite{Huang2022Near-optimal} \\
         & robust        & $m+k^2\eps^{-2z-2}$ & \cite{Huang2023General}  \\
doubling metrics 
        & vanilla ($\mathbf{N}$) & $k\epsilon^{-\max\{2, z\}} \cdot \ddim(M)$ & \cite{Cohen-addad2021New} \\
& robust        & $m + k^2\eps^{-2z}\cdot(\ddim(M)+\eps^{-1})$ & \cite{Huang2023General}      \\
$n$-point metric space & vanilla ($\mathbf{N}$) & $k\eps^{-\max\{2,z\}}\cdot\log n$ & \cite{Cohen-addad2021New} \\
& robust, $z = 1$        & $(k+m)^{k+m}\eps^{-2}\log^4 n$ & \cite{feldman2012data} \\
         & robust        & $m + k^2\eps^{-2z}(\log n+\eps^{-1})$ & \cite{Huang2023General} \\
bounded treewidth graphs & vanilla ($\mathbf{N}$) & $k\eps^{-\max\{2,z\}}\cdot\tw$ & \cite{Cohen-addad2021New} \\
& robust        & $m + k^2\eps^{-2z-2}\cdot \tw$ & \cite{Huang2023General} \\
excluded-minor graphs & vanilla ($\mathbf{N}$), $z = 1$ & $k\eps^{-2}$ & \cite{Cohen-AddadD0SS25} \\
         & vanilla ($\mathbf{N}$) & $k\eps^{-2}\cdot \min\{\eps^{-z-1}, k\}$ & \cite{Cohen-AddadD0SS25} \\
         & robust        & $m + k^2\eps^{-2z-2}$ & \cite{Huang2023General}      \\
all the above & robust & $\mathbf{N} + \min\{km\epsilon^{-1}, m\epsilon^{-2z}\}$ & ours \\
        \bottomrule
    \end{tabularx}
\end{table}

It is worth mentioning that our entire algorithm is deterministic, provided that the given vanilla coreset algorithm is deterministic.
As highlighted in recent surveys~\cite{DBLP:journals/ki/MunteanuS18, DBLP:journals/widm/Feldman20},
deterministically constructing coresets is an important and less understood open question, even for vanilla clustering without outliers.
The best-known deterministic construction works only in low-dimensional Euclidean space~\cite{DBLP:journals/dcg/Har-PeledK07}, yielding a coreset of size $\poly(k) \cdot \eps^{-O(d)}$.
For Euclidean \kMeans, a coreset of size $k^{\poly(\eps^{-1})}$ can be constructed deterministically~\cite{FeldmanSS20Turning}, and the size 
can be improved to $\poly(k\eps^{-1})$ if we consider a slightly more general definition of coresets, known as coresets with offset\footnote{In the definition of coresets with offset, the original clustering objective is preserved by adding a universal constant (the offset) to the clustering objective on the coreset.}~\cite{Cohen-addad23Deterministic}. 
All these results can be plugged into our reduction to directly achieve a deterministic coreset construction for \kzmC, with the coreset size increasing as shown in \eqref{eq:intro_size}.
\footnote{For coresets with offset, our reductions and analysis can be easily adapted and still achieve asymptotically the same coreset size.}

\paragraph{Streaming Implementation of~\Cref{thm:informal}}

We also obtain reduction style coreset construction algorithms for \kzmC over a \emph{dynamic} stream of points in $\mathbb{R}^d$ using small space, which follows from a streaming implementation of~\Cref{thm:informal}.
We consider a standard streaming model proposed by~\cite{DBLP:conf/stoc/Indyk04}
which has been widely employed in the literature.
In this model, the input points come from a \emph{discrete} grid $[\Delta]^d$ for some integer parameter $\Delta$,
and they arrive as a dynamic stream with point insertions and deletions.
At the end of the stream, the algorithm should return a weighted set as the coreset.
We present in~\Cref{thm:intro_streaming} the result
for the case of $z = 1$ which is \kMedian with $m$ outliers.

\begin{theorem}[Informal version of \Cref{thm:dynamic coresets}]
    \label{thm:intro_streaming}
    Assume there is a streaming algorithm that constructs an $\eps$-coreset for \kMedian for every dataset from $[\Delta]^d$ presented as a dynamic stream, using space $W(d, \Delta, k,\epsilon^{-1})$ and with a failure probability of at most $1/\poly(d\log\Delta)$. 
Then, there is a streaming algorithm that constructs an $\eps$-coreset for \kRMedian with constant probability for every dataset from $[\Delta]^d$ presented as a dynamic stream,
    using space $\min\{W_1, W_2\}\cdot \poly(d\log\Delta)$, where 
    \begin{align*}
        W_1&= \tilde O(km\eps^{-1}) + W\left(d,\Delta, k, O(\eps^{-1})\right),\\
        W_2&= \tilde O(k+m\eps^{-2}) + W\left(d+ 1, \eps^{-1} \Delta^{\poly(d)}, k\poly(d), O(\eps^{-1})\right).
    \end{align*}
\end{theorem}

Observe that we did not explicitly state the size of the coreset, and only focus on the space complexity,
since the coreset size is upper bounded by the space complexity of the streaming algorithm.
Moreover, once one obtains an $\eps$-coreset at the end of the stream, one can always run another existing coreset construction again on top of it to obtain a possibly smaller coreset.
Hence, it is only the space complexity that matters in streaming coreset construction.
We also note that once the coreset is obtained from our streaming algorithm,
one can continue to find a $(1 + \epsilon)$-approximation to \kRMedian without using higher order of space.

By combining \Cref{thm:intro_streaming} with an existing streaming coreset construction for \kMedian~\cite{Braverman2017Clustering}, we obtain a space complexity of $\tilde O\left(\min\{km\eps^{-1}, m\eps^{-2}\} + k\eps^{-2}\right)\cdot \poly(d\log\Delta)$.
Remarkably, the dependence on $k$ and $\eps$ matches that of~\cite{Braverman2017Clustering}, which is $k\eps^{-2}$.
A similar bound for $z = 2$, i.e., \kMeans with $m$ outliers, may be obtained by combining with a dynamic streaming algorithm for \kMeans (e.g.,~\cite{Song2018Nearly}).
These results are the first nontrivial dynamic streaming algorithms for robust $k$-clustering,
where the dependence on $\epsilon, k, d, \log \Delta$ is comparable to existing bounds for vanilla clustering (up to degree of polynomial),
and the linear dependence in $m$ is also shown to be necessary (see \Cref{claim:lb_m}).

In fact, our result is also the first dynamic streaming coreset construction for \emph{any} variant of \kzC.
Indeed, recent works manage to devise (offline) coresets for variants of clustering,
and they can even immediately imply insertion-only streaming algorithm via standard techniques such as merge-and-reduce~\cite{DBLP:conf/stoc/Har-PeledM04},
but they are not readily applicable to dynamic point streams.
Our new reduction style framework is a new approach that opens a door to obtaining dynamic streaming algorithms for other variants of clustering,
which may be of independent interest.

 \subsection{Technical Overview}
\label{sec:tech_overview}

Our main technical contributions are the establishment of two conditions for a vanilla coreset to become a robust coreset,
and black-box reductions that turn vanilla coresets to satisfy these conditions,
which yields the bounds mentioned in \Cref{thm:informal}.
The two conditions reveal new fundamental structures of the robust clustering problem, especially the structural relation to vanilla clustering,
which may be of independent interest.
Moreover, to make vanilla coresets satisfy the two conditions in a black-box way,
it requires us to devise a bounded-diameter decomposition called almost-dense decomposition,
a new separated duplication transform for metric spaces,
as well as adapting and refining a recently developed sparse partition~\cite{Jia05Universal,DBLP:conf/icalp/Filtser20} (which was also studied under the notion of consistent hashing~\cite{arxiv.2204.02095}).
The sparse partition/consistent hashing technique is also crucially used in our streaming implementation,
which suggests a framework fundamentally different from those based on quadtrees as in previous dynamic streaming algorithms for (vanilla) clustering~\cite{Braverman2017Clustering,Song2018Nearly}.

We focus on presenting these ideas for the \kRMedian problem (which is the $z = 1$ case) in Euclidean $\mathbb{R}^d$ (although most discussion already works for general metrics).
We also discuss how to implement the reductions in streaming.

\subsubsection{First Reduction: $\left(O(km\eps^{-1}) + N\right)$ Size Bound}
\label{sec:first_overview}

\paragraph{Condition I: Vanilla Coresets on Dense Datasets Are Robust}
Our first condition comes from a simple intuition:
we wish $\cost(X, C) \approx \cost^{(m)}(X, C)$ for every $C \subseteq \R^d$, 
which is perhaps the most natural way to ensure a vanilla coreset is automatically a robust coreset.
Now, a sufficient condition is to ensure the dataset is somewhat uniform, in the sense that for every center set $C \subseteq \R^d$, the contribution from every $m$ data points
is at most $\epsilon \OPT$ (or equivalently, each data point contributes $\epsilon \OPT / m$).
This way, whatever the $m$ outliers are, the change of objective from $\cost$ to $\cost^{(m)}$ is always within $\epsilon$ factor.

We thus consider datasets that can be broken into parts
such that each part has diameter $\lambda := O(\epsilon \OPT / m)$ and contains at least $\Omega(\epsilon^{-1} m)$ points.
In such datasets, every data point can find $\Omega(\epsilon^{-1}m)$ points within distance $\lambda = O(\epsilon \OPT / m) $,
hence its contribution is naturally charged to these nearby points, which is averaged to $O(\epsilon \OPT / m)$.
We call such a dataset \emph{dense},
and the above intuition leads to a formal reduction that for dense datasets, a vanilla coreset is as well a robust coreset (\Cref{thm:main1}).

\paragraph{Reduction I: Almost-dense Decomposition}
Of course, a general dataset may not be dense.
Nevertheless, we manage to show a weaker almost-dense decomposition exists for every dataset (\Cref{lem:decomposition}).
This decomposition breaks the dataset into two parts:
a subset $A$ that is guaranteed to be dense
and a remaining subset $B$ that only consists of $O(km\eps^{-1})$ points.
The subset $B$ actually has more refined property:
it can be broken into parts such that each part has diameter at most $\lambda$ (which is similar to the dense part)
and each contains at most $O(\epsilon^{-1}m)$ points (which is the ``complement'' property of the dense part).
We call subset $B$ the \emph{sparse} subset.

To construct such as decomposition,
we start with an optimal solution $C^*$ to the robust clustering
(in the algorithm it suffices to use a tri-criteria solution, see \Cref{def:tri_criteria}).
Let $\lambda = O(\epsilon \OPT / m)$ be the target diameter bound of the dense and sparse parts.
We first identify those isolated points $F$ consisting of both the $m$ outliers and those ``far away'' points with distance more than $\lambda$ from $C^*$.
The number of far away points is bounded by $\OPT/\lambda = O(\eps^{-1} m)$ using an averaging argument.
Now since every remaining point is within $\lambda$ to $C^*$,
we cluster all these remaining points, namely $X \setminus F$,
with respect to $C^*$.
Each cluster automatically has a diameter bound by $O(\lambda)$.
Finally, we take clusters with more than $\Omega(\epsilon^{-1} m)$ points as the dense subset, and the remainder, along with $F$, forms the sparse subset, which contains at most $O(|C^*| \epsilon^{-1}m) + O(\eps^{-1}m) + m = O(km\eps^{-1})$ points in total.

With such a decomposition, one can simply put the sparse subset who has $O(km\eps^{-1})$ points into the coreset,
and use the vanilla coreset on the dense subset. This yields our first $O(km\eps^{-1}) + N$ size bound (recall that $N$ is the size of the vanilla coreset) as in \Cref{thm:informal}.

\subsubsection{Second Reduction: $(m\eps^{-2} + N)\poly\log(km\eps^{-1})$ Size Bound}
\label{sec:second_overview}

\paragraph{Overview of Reduction II}
Indeed, in Reduction I, the step of simply adding sparse subset into the coreset may sound suboptimal.
Hence, our second algorithm is built on top of Reduction I, such that we first follow the steps in Reduction I to build a vanilla coreset on the dense subset,
and then run more refined new steps to the sparse subset.
In fact, the new steps can be viewed as a separate reduction that takes the sparse subset as input and generates a robust coreset.
We focus on these new steps and call it Reduction II (without the steps already in Reduction I). 
Crucially, Reduction II itself is standalone and can run on any dataset (and thus can be independent of Reduction I), 
albeit it would result in a coreset of size $O(\log n)$ (where $n$ is the size of its input dataset).
Luckily, we would eventually run it on the sparse subset which has $O(km\eps^{-1})$ points,
so the $\log n$ factor becomes $\log(km\epsilon^{-1})$,
and the overall size bound is competitive (which yields the other size bound in \Cref{thm:informal}).

Next, we start with describing our Condition II (which is also standalone), and discuss how our Reduction II make a vanilla coreset to satisfy this condition.

\paragraph{Condition II: Size-preserving Vanilla Coresets Are Robust} 
The second condition requires a stronger property in the definition of the coreset (instead of the dense property in the dataset as in Condition I).
We consider a \emph{size-preserving} property, which is defined with respect to some $\lambda$-bounded diameter decomposition -- a partition $\mathcal{P}$ of point set $X$ such that each part $P \in \mathcal{P}$ has diameter at most $\lambda$.
For a given $\lambda$, a (vanilla) coreset $S$ is size-preserving, if there \emph{exists} a $\lambda$-bounded diameter decomposition $\mathcal{P}$, such that for every $P \in \mathcal{P}$, $|P| = |S \cap P|$
(here we interpret the coreset, which is a weighted set, as a multiset).
This ``size-preserving property'' of coresets intuitively means the coreset could be interpreted as
moving each data point by at most $\lambda$ distance.
More formally, there is a one-to-one correspondence $g$ between dataset $X$ and coreset $S$ 
such that $\dist(x, g(x)) \leq \lambda$.

Our Condition II states that if a vanilla coreset $S$ satisfies the \emph{size-preserving property} then it is automatically a robust coreset (\Cref{cor:second}).
Next, we provide an overview of the difficulties we encountered in the proof of Condition II and the high-level ideas for resolving them.

\subparagraph{Naive Approach: Bounding the Gap between $\cost$ and $\cost^{(m)}$}

We explore the condition to make the gap between $\cost(X, C)$ and $\cost^{(m)}(X, C)$ is small.
Unlike Condition I, here we are now making no assumptions about the structure of $X$.
Observe that the precise difference between the two cost functions is the outliers:
Let $X_{\out}$ be the set of outliers (which is defined with respect to $C$),
then the vanilla $\cost(X, C)$ is larger and has an extra term $\sum_{x \in X_{\out}} \dist(x, C)$ compared with the robust $\cost^{(m)}(X, C)$.
An immediate idea to ``avoid'' this term in vanilla $\cost(X, C)$,
is to additionally put a center at each point in $X_{\out}$
so that points from $X_{\out}$ contribute $0$.
This works perfectly when each outlier is separated from every other point since each point in $X_{\out}$ forms a singleton cluster,
and this suggests a necessary condition:
the vanilla coreset $S$ needs to hold for $k' = k + m$ centers, i.e., it is a coreset for $k'$-\textsc{Median}.

Unfortunately, centers on outliers may not form singleton clusters, in which case $\cost(X, C \cup X_{\out})$ can underestimate $\cost^{(m)}(X, C)$. 
Even if it is possible to pick another $C'$ to make this said gap small,
it is still challenging to also prove that $\cost(S,C\cup C')$ is always close to $\cost^{(m)}(S,C)$ for an arbitrary $S$ and $C$
(which is needed to imply $\cost^{(m)}(X,C)\approx \cost^{(m)}(S,C)$).

\subparagraph{Key Observation: Vanilla Coresets Also Preserve the Gap}
Hence, instead of seeking for a specific center set $C'$ that makes the gap $\cost^{(m)}(X, C) - \cost(X, C \cup C')$ small for every $C$,
we show a weaker guarantee, such that the \emph{gap is preserved} between $X$ and the coreset $S$:
we allow a large gap, just that if the gap is large in $X$, then it is also large in $S$, and the quantity is comparable.
Specifically, we show that a coreset $S$ for $k'$-\textsc{Median} with size-preserving property can preserve this gap.
Namely,
\begin{equation}
    \label{eqn:gap_preserve}
    \cost^{(m)}(S, C) - \cost(S, C \cup C') = \cost^{(m)}(X, C) - \cost(X, C \cup C') \pm \epsilon \cost^{(m)}(X, C)
\end{equation}
for some $C'\subseteq X_\out$.
Indeed, this combining with the guarantee $\cost(S, C \cup C') \in (1 \pm \eps) \cost(X, C \cup C')$ (which follows from the $k'$-\textsc{Median} coreset guarantee of $S$) 
implies $\cost^{(m)}(S, C) \approx \cost^{(m)}(X, C)$ (which is the robust coreset guarantee).

To see why the gap can be preserved,
the intuition is that there are only a few points that contribute non-zero values to the gap $ \cost^{(m)}(X,C) - \cost(X,C \cup C') $.
To put it simply, let us first consider $C' = X_{\out}$. Then only those inliers $x$ of $X$ that are closer to outliers $X_{\out}$ than to $C$ can make $\dist(x, C) \neq \dist(x, C \cup C')$ hold. Let $D$ denote such inliers, and let us loosely assume that only $g(D)$ contributes non-zero values to the gap $\cost^{(m)}(S,C) - \cost(S,C \cup C')$ (recalling that $g$ is a one-to-one correspondence between the dataset $X$ and the coreset $S$, such that $\dist(x, g(x)) \leq \lambda$ for every $x \in X$).
Then the difference of these two gaps can be upper bounded by $|\cost(D,C) - \cost(g(D),C)| + |\cost(D,C\cup C') - \cost(g(D),C\cup C')|\le O(|D|\cdot \lambda)$ using the triangle inequality, which is sufficient if $|D|\le O(m)$ and we pick $\lambda = \eps\OPT/m$.
Unfortunately, the naive setup of $C' = X_\out$, as we do in the analysis above, may \emph{not} lead to a small $D$ (noticing that $D$ is defined with respect to $C'$). 
Moreover, $C'$ cannot be made small as well, as this could result in a large error of $\eps \cdot \cost(X, C \cup C')$ (recalling that \eqref{eqn:gap_preserve} combining the coreset guarantee $\cost(S,C\cup C')\in (1\pm \eps)\cdot \cost(X,C\cup C')$ forms the complete cost preservation).
Hence, it remains to find an appropriate $C'$ that satisfies both requirements, and we provide a simple algorithm (\Cref{alg:additional centers}) to precisely achieve this.

\paragraph{Reduction II: Making Vanilla Coresets Size-preserving}
To utilize Condition II, we need to show that, for a given $\lambda > 0$, we can build a vanilla coreset $S$ that is size-preserving using a generic vanilla coreset algorithm in a black-box way. Specifically the coreset should satisfy that there exists a $\lambda$-bounded parition $\calP=(P_1,\dots,P_t)$ of $X$
such that 
$|P_i| = |S\cap P_i|$ for $i\in [t]$ (recalling that we treat a coreset as a multiset).
We first assume a generic partition $\calP$ is given, and assume that $t=|\calP| = k$ for simplicity.

\subparagraph{Key Observation: Separated Instance is Easy} 

We make use of a simple observation: 
if the $k$ parts are well-separated, meaning they are sufficiently distant from each other, then a vanilla coreset for $k$-\textsc{Median} is already (nearly) size-preserving.
To see this, if we put one center at each of the $P_j$'s for $j \neq i$ but without any center at $P_i$,
the clustering cost is roughly $|P_i| \cdot \dist(P_i, X \setminus P_i)$ and this is up to $(1\pm\eps)$ factor
to $|S \cap P_i| \cdot \dist(P_i, X \setminus P_i)$ (by the coreset guarantee of $S$).
This implies a slightly relaxed $|S \cap P_i| \in (1 \pm \epsilon) |P_i|$, and it can be further adjusted to become strictly size-preserving.

\subparagraph{Key Idea: Ensuring Separation Artificially} 
A straightforward idea to deal with a generic $\calP$ is to artificially separate the parts. 
More precisely, for Euclidean $\R^d$, we define a mapping $h : P \to \mathbb{R}^{d + 1}$, such that for $x \in P_i$,
$h(x) := (x, i \cdot w)$ for some sufficiently large $w$ (which can be loosely considered as infinite).
We define $P_i' := h(P_i)$. 
Clearly, $\calP' := (P_1', \ldots, P_k')$ is a $\lambda$-bounded partition for $X' := h(X)$, and now the partition $\calP'$ is well-separated. Thus, one can apply the above argument for $\calP'$, obtain a size-preserving coreset $S'$ for $X'$,
and find the corresponding $S$ for $X$.

\subparagraph{New Issue: Ensuring Coreset Guarantee} However, a crucial issue arises: how can we ensure that $S$ is a coreset for $X$? In other words, we need to translate the coreset guarantee of $S'$ (on $X'$) to that of $S$ (on $X$). 
It seems that the only way is to find a ``bridge'' center set $C'$ for any fixed $C$ such that $\cost(X,C) \approx \cost(X',C')$ and $\cost(S,C) \approx \cost(S',C')$. If such $C'$ exists and $S'$ can handle at least $|C'|$ centers, then the coreset guarantee of $S$ follows. 
Now, let us focus on the requirement $\cost(X,C)\approx\cost(X',C')$. The tricky situation arises when some points in $X$ may have the same nearest center in $C$ but belong to different parts; for instance, assume that $x_i \in P_i$ and $x_j \in P_j$ both have the nearest center $c \in C$. After mapping through $h$, the images $x_i' = h(x_i)$ and $x_j' = h(x_j)$ can be ``infinitely'' far apart. Therefore, to maintain the clustering cost, we need to include two copies, $(c, i \cdot w)$ and $(c, j \cdot w)$, in $C'$. In the worst case, we would need to include $|\calP| = k$ copies of every center in $C$, resulting in a size of $k^2$ for $C'$, and thus $S'$ should be a coreset for \ProblemName{$k^2$-Median}. Unfortunately, this worst-case scenario is unavoidable if $\calP$ is generic, and the reduction to \ProblemName{$k^2$-Median}\xspace falls far from satisfaction, as our ultimate goal is near-linear size in $k$.

\subparagraph{Solution: Grouping Parts using Sparse Partition~\cite{Jia05Universal}} 

To address the above problem, we start by using a more careful way to ``prune'' centers in $C'$.
For a part $P_i \in \calP$, if it is far from $C$, specifically $\dist(P_i, C) > \eps^{-1}\lambda \ge \eps^{-1}\diam(P_i)$, then it suffices to include only one center in $C'$ for $P_i$, which is $c_i = \arg\min_{c \in C} \dist(c, P_i)$, and we can verify that $\cost(P_i, \{c_i\}) \in (1 \pm \eps) \cdot \cost(P_i, C)$. Through this pruning, these far parts can increase $C'$ by at most $k$, which is acceptable.

The challenging part is then handling the close parts. For a close part $P_i$, we denote by $C_i$ the centers that is relevant to it, i.e., $C_i$ contains the nearest centers to points in $P_i$. This implies that $\cost(P_i,C) = \cost(P_i,C_i)$, meaning we only need to include $\{(c,i\cdot w) : c \in C_i\}$ in $C'$ for every $i$.
As discussed earlier, $\sum_i |C_i|$ could be $\Omega(k^2)$ in the worst case.

To further reduce the total number of relevant centers, we propose a new idea: group these parts together to create a new partition $(G_1, \dots, G_l)$ of $X$. The relevant centers of $G_j$ are $\bigcup_{i: P_i \subseteq G_j} C_i$. Our goal then becomes to find a grouping strategy such that $\sum_j |\bigcup_{i: P_i \subseteq G_j} C_i|$ is small.

We achieve this using a more structured partition called \emph{sparse partition}~\cite{Jia05Universal}. Specifically, we can compute a partition $\calQ = (G_1, \dots, G_l)$ of $X$ such that each part $G_j$ has a diameter at most $O(\eps^{-1}\lambda \log n)$, and any subset with a diameter $O(\eps^{-1}\lambda)$ intersects at most $O(\log n)$ parts. For simplicity, assume each part in $\calP$ is entirely covered by exactly one part in $\calQ$ (though this is not generally true, the proof follows similarly). 
Then, $\calQ=(G_1,\dots,G_l)$ fits our needs perfectly. To see this, consider any center $c$. Notice that $c$ can appear in at most $O(\log n)$ of the sets $\bigcup_{i: P_i \subseteq G_j} C_i$, since otherwise the subset $\bigcup_{i: c \in C_i} P_i$, whose diameter is $O(\eps^{-1}\lambda)$, would intersect more than $O(\log n)$ parts, contradicting the guarantee provided by the sparse partition.
As a result, $\sum_j |\bigcup_{i: P_i \subseteq G_j} C_i| \le O(k \log n)$. This ensures that it suffices to construct a coreset for \ProblemName{$O(k \log n)$-Median}\xspace that is size-preserving with respect to $\calQ$.

An issue is that $\calQ$ is $O(\eps^{-1}\lambda\log n)$-bounded, but this can be adjusted through rescaling, resulting in a size increase by a factor of $\poly(\eps^{-1}\log n)$. Given that this coreset applies to a subset of size $n = \poly(km\epsilon^{-1})$ in the full algorithm, the $\log n$ term is acceptable.

\subsubsection{Streaming Implementations}

We provide an overview of how these steps can be adapted to the streaming setting using foundational streaming algorithmic tools.

\paragraph{$\lambda$-Bounded Partition}
Recall that both of our reductions rely on a $\lambda$-bounded partition $\calP$ (where $\lambda = \epsilon \OPT / m$, and we can guess this $\OPT$). Therefore, the key is to design a streaming algorithm that finds a $\lambda$-bounded partition. However, in the streaming setting, it is not possible to maintain the $\lambda$-bounded partition explicitly, as it requires $\Omega(n)$ space, which is not affordable.

Our strategy is to employ a hashing version of sparse partition defined in~\cite{arxiv.2204.02095} that is space efficient and data oblivious and thus suitable for streaming. This hashing partitions the entire $\mathbb{R}^d$ into buckets of diameter of $O(\lambda)$. Hence, after the stream ends, those non-empty buckets form a $\lambda$-bounded partition of $X$. Moreover, this hashing also guarantees that any subset of points with a diameter of $\lambda / \poly(d)$ intersects at most $\poly(d)$ buckets (similar to the sparse partition guarantee).
By combining with the existence of the small-sized partition in \Cref{lem:decomposition},
we can conclude there are very few, i.e., $\poly(d\epsilon^{-1}) \cdot (k  + m) $, distinct buckets after this hashing.

In the following, we fix such a hashing, and we outline how our reductions can be adapted to the streaming setting using foundational streaming algorithmic tools.

\paragraph{The First $O(km)\poly(\eps^{-1}d\log\Delta)$ Space Bound} 
We first show how to implement our Reduction I in the streaming setting. Our main task is to identify all points that lie in small buckets containing fewer than $O(\epsilon^{-1}m)$ points, which correspond to the sparse subset $B$, and construct a vanilla coreset for the points lying in other buckets, which correspond to the dense subset. The latter is easy; once the former is achieved, we can simply run a vanilla dynamic streaming algorithm on a stream representing $X \setminus B$ (which can be achieved using some tricks).

The former task can be translated into a streaming algorithm:
identify at most $a := \poly(d\epsilon^{-1}) \cdot (k  + m)$ buckets,
such that each contain at most $b := O(\eps^{-1}m)$ points.
This task is similar to sparse recover in streaming, but it is two-level which is different from the standard one, i.e., it should recover only from those buckets that contain small number of points.
Hence, we introduce a modified sparse recovery structure (\Cref{lem:sparse recovery}) for this specific task, and it takes about $O(ab)\poly(d\log\Delta) = O(km+m^2) \poly(\epsilon^{-1}d\log\Delta)$ space. 

However, this $m^2$ dependence is still suboptimal. 
To achieve a near-linear space bound in $m$, we show that there exists a subset $F \subseteq X$ of size $m \cdot \poly(d \eps^{-1})$ such that, after removing $F$, the remaining points span very few buckets, allowing us to set $a:=k\poly(d)$ (see \Cref{lem:bounded size}). This subset $F$ can be efficiently identified using a two-level $\ell_0$-sampler~\cite[Lemma 3.3]{arxiv.2204.02095} in the dynamic stream (see \Cref{lem:isolated}). We then use (modified) sparse recovery to identify the sparse subset after removing $F$, reducing the total space to $O(km)\poly(\eps^{-1}d\log\Delta)$.

\paragraph{The Second $\tilde{O}(k+m)\cdot \poly(\eps^{-1}d\log\Delta)$ Space Bound}

Our second algorithm first uses the same streaming algorithm to find $F$ as in the first bound, and recall that the non-empty buckets after removing $F$ is small, i.e., $|\phi(X\setminus F)|\le O(k\poly(d))$. This implies that $\phi(X\setminus F)$ gives rise to a small-sized $\lambda$-bounded partition, denoted by $\calP$. 
Our Condition II states that if we construct a vanilla coreset $S$ that is size-preserving, i.e., $|S \cap P| = |P|$ for every $P$, it is directly a robust coreset.
Hence, in summary, our next step is to implement the construction of a size-preserving vanilla coreset in the dynamic streaming setting (see Reduction II in \Cref{sec:tech_overview}). Specifically, we map every point $x$ in the stream into $(x,\phi(x)\cdot w)$, where $w$ is sufficiently large, and that $\phi(x)$ can be encoded in $\poly(d \log \Delta)$ bits and thus can be interpreted as an integer.
This transformed data is then fed to an instance of a streaming algorithm designed for vanilla coresets. 
At the end of the stream, we obtain a coreset from this instance and consider its pre-image under $\phi$ (which deletes the last coordinate) as the coreset for $X$, which is approximately size-preserving.
We run a sparse recovery algorithm to recover the sizes of each bucket to calibrate the weights (to make it strictly size-preserving).

 \subsection{Related Work}
\label{sec:related}
Besides clustering with outliers, researchers have extensively explored coresets for various other clustering variants. 
Two particularly interesting variants are clustering with capacity constraints
and clustering with fairness constraints, both of which have been investigated in~\cite{DBLP:conf/nips/Chierichetti0LV17, DBLP:conf/waoa/0001SS19, DBLP:conf/nips/HuangJV19,Cohen-Addad2019Fixed-Parameter,DBLP:conf/icalp/BandyapadhyayFS21, Braverman22Power} for coreset constructions.
Notably, the coresets for both these variants can be reduced to an assignment-preserving coreset~\cite{DBLP:conf/waoa/0001SS19, DBLP:conf/nips/HuangJV19}, for which several studies apply a hierarchical sampling framework~\cite{Chen09OnCoresets} to construct small-sized coresets with assignment-preserving properties~\cite{DBLP:conf/icalp/BandyapadhyayFS21, Braverman22Power}. Recently, Huang et al.~\cite{Huang2023General} introduced a general unified model to capture a wide range of clustering problems with general assignment constraints, including clustering with fairness/capacity constraints, clustering with outliers (which we consider), and a fault-tolerant variant of clustering~\cite{DBLP:journals/tcs/KhullerPS00}.

Coresets for other variants of clustering have also been considered, including projective clustering~\cite{Feldman11unified,FeldmanSS20Turning,DBLP:conf/aistats/Tukan0ZBF22}, clustering with missing values~\cite{DBLP:conf/nips/BravermanJKW21}, ordered-weighted clustering~\cite{DBLP:conf/icml/BravermanJKW19}, and line clustering~\cite{DBLP:conf/nips/MaromF19, DBLP:conf/nips/LotanSF22}. \section{Preliminaries}
\label{sec:prelim}

For integer $n \geq 1$, let $[n] := \{1,2,\dots,n\}$.
Given a mapping $\phi: X \to Y$, for a subset $Z \subseteq X$, denote $\phi(Z) := \{\phi(x) : x \in Z\}$, and for an image $y \in Y$, denote its inverse by $\phi^{-1}(y) := \{x \in X : \phi(x) = y\}$.
Let $M=(V,\dist)$ be an underlying metric space,
and we assume oracle access to $\dist$.
For a point set $X\subseteq V$, let $\diam(X) := \max_{x,y\in X}\dist(x,y)$ be the diameter of $X$.
For a set $C=\{c_1,\dots,c_t\} \subseteq V$ ($t>0$),
a clustering of $X$ with respect to $C$ is a partition $\{P_1,\ldots,P_t\}$ of $X$ such that for every $i\in [t]$ and every point $x\in P_i$, $\dist(x,c_i) \le \dist(x,C)$.

\paragraph{Weighted Set}
A set $S$ with an associated weight function $w_S : S \to \R_{\ge 0}$ a \emph{weighted set}.
When we talk about a weighted set $A$, we denote the weight function of $A$ by $w_A$ unless otherwise specified. 
For an unweighted set $B$, we interpret it as weighted set with unit weight function $w_{B}\equiv 1$.
We say a weighted set $Y$ is a weighted subset of $X$, denoted by $(Y,w_Y)\subseteq (X,w_X)$,
if $Y\subseteq X$ and for every $x\in Y$, $w_Y(x)\le w_X(x)$. 
For two weighted sets $X$ and $Y$, the weight of their union $Z := X \cup Y$ is defined as $w_Z := w_X + w_Y$. 
We denote by $X - Y$ the weighted set $Z$ with the weight function $w_Z$, where $w_Z = w_X - w_Y$, and $Z$ is the support of $w_Z$.
Here, for two weight functions $f:X\to \R$ and $g:Y\to \R$, we use $f + g$ to denote the weight function $h : X\cup Y\to \R$ such that for every $x\in X\cup Y$, $h(x) := f(x) + g(x)$. 
Likewise, we denote by $f-g$ the subtraction of weight functions.\footnote{If $x \not\in X$, we define $f(x) := 0$, and similarly, if $x \not\in Y$, we define $g(x) := 0$.}

\paragraph{Clustering with Outliers (Weighted)}
In our proof, we need to consider the weighted version of \kzmC, and we extend the definition as follows.
For a weighted set $X\subseteq V$, and a center set $C\subseteq V$ with $|C|=k$, we define the clustering objective for \kzC as $\cost_z(X,C):=\sum_{x\in X}w_X(x)\cdot \dist^z(x,C)$. 
Moreover, given real number $h \ge 0$, we denote by $\calL_X^{(h)}:=\{(Y,w_Y)\subseteq (X,w_X): w_Y(Y) = h\}$ the collection of all possible weighted subset of $X$ with a total weight equal to $h$. 
We define the objective for \kzmC as 
\begin{equation}
    \label{eq:costm}
    \cost_z^{(m)}(X,C):=\min_{(Y,w_Y)\in\calL_X^{(m)}} \cost_z(X - Y,C).
\end{equation}
One can verify that when $m=0$, the definition of $\cost_z^{(m)}$ is equivalent to $\cost_z$.
We call a weighted set $Y$ an $m$-outlier set of $X$ with respect to $C$, if $(Y,w_Y)\in\calL_X^{(m)}$ and $\cost_z^{(m)}(X,C) = \cost_z(X-Y,C)$. 
We use $\OPT_z^{(m)}(X)$ to denote the optimal objective value for \kzmC on $X$, i.e., $\OPT_z^{(m)}(X):=\min_{C\in V^k} \cost_z^{(m)}(X,C)$.

\paragraph{Coresets} 
Our definition of coreset is almost identical to the one in~\cite[Definition 2.1]{Huang2022Near-optimal}, except that the coreset preserves the objective for all potential number of outliers $0 \leq h \leq m$ up to $m$ simultaneously.\footnote{Although the definition of coresets in~\cite{Huang2022Near-optimal} does not require preserving objective for all $0\le h\le m$ simultaneously, their algorithm actually constructs such a coreset (see~\cite[Remark 3.2]{Huang2022Near-optimal}).}

\begin{definition}[Coresets]
    Given $0<\eps<1$ and a dataset $X\subseteq V$, an $\eps$-coreset of $X$ for \kzmC is a weighted set $S\subseteq X$ such that
    \begin{equation*}
        \forall C\in V^k, 0\le h\le m, \quad \cost_z^{(h)}(S,C) \in(1\pm\eps)\cdot \cost_z^{(h)}(X,C).
    \end{equation*} 
    A weighted set is called an $\eps$-coreset for the special case of \kzC
    if it is an $\eps$-coreset for \kzmC with $m = 0$.
\end{definition}

For analysis purposes, we also need a weaker notion of coreset, in \Cref{def:additive_coreset}, that allows for additive errors. 
Such notions were also widely used in previous coreset constructions (see e.g., ~\cite{Cohen-addad2021New,Cohen-Addad22Towards,Braverman22Power,Huang2022Near-optimal}).

\begin{definition}[Additive-error coresets]
    \label{def:additive_coreset}
    Given $0<\eps<1, \eta\ge 0$ and a dataset $X\subseteq V$, an $(\eps,\eta)$-coreset of $X$ for \kzmC is a weighted set $S\subseteq X$ such that
    \begin{equation*}
        \forall C\in V^k, 0\le h\le m, \quad \left|\cost_z^{(h)}(S,C) - \cost_z^{(h)}(X,C)  \right|\le \eps\cdot \cost_z^{(h)}(X,C) + \eta.
    \end{equation*} 
\end{definition}

An important property of coresets is that it is composable.
Roughly speaking, if $S_A$ is a coreset of $A$ and $S_B$ is a coreset of $B$, then $S_A\cup S_B$ is a coreset of $A\cup B$. 
It is well known that composability holds directly from the definition of vanilla coreset,
and here we verify that it also holds for our definition of coresets for \kzmC.
We present this fact with respect to the more general additive-error coreset,
and we provide a proof in the Appendix~\ref{sec:composability} for completeness.

\begin{restatable}[Composability of coresets]{fact}{ComposabilityCoresets}
    \label{fact:composability}
    For $0<\eps<1, \eta_1,\eta_2\ge 0$, and two datasets $X,Y\subseteq V$, if $S_{X}$ is an $(\eps,\eta_1)$-coreset of $X$ for \kzmC, and $S_{Y}$ is an $(\eps,\eta_2)$-coreset of $Y$ for \kzmC, then $S_X\cup S_Y$ is an $(\eps,\eta_1+\eta_2)$-coreset of $X\cup Y$ for \kzmC.
\end{restatable}

\paragraph{Tri-criteria Approximation}

Similar to many previous works~\cite{Braverman22Power,Huang2022Near-optimal,Huang2023General}, we need a tri-criteria approximation algorithm for the \kzmC problem.
Here, an $(\alpha, \beta, \gamma)$-approximation is a set of $\beta k$ centers, whose cost is at most $\alpha$ times the optimal, allowing a violation $\gamma$ to the number of outliers.
\begin{definition}[$(\alpha,\beta,\gamma)$-Approximation]
    \label{def:tri_criteria}
    Given $\alpha,\beta,\gamma \ge 1$ and a dataset $X\subseteq V$, we say a center set $C\subseteq  V$ is an $(\alpha,\beta,\gamma)$-approximation solution to \kzmC on $X$ if it holds that $\cost_z^{(\gamma m)}(X,C)\le \alpha\cdot\OPT_z^{(m)}(X)$ and $|C|\le \beta k$.
\end{definition}
There exists a randomized algorithm that computes a $(2^{O(z)}, O(1), O(1))$-approximation in near-linear time with high probability~\cite{Bhaskara19Greedy} (which is also applicable to general metrics). For a deterministic approach,
a $(2^{O(z)}, O(1), 1)$-approximation solution can be obtained in polynomial time (e.g.,~\cite{DBLP:journals/talg/FriggstadKRS19}).
A more detailed discussion on other possible algorithms for such an approximation can be found in~\cite[Appendix A]{Huang2022Near-optimal}.

\begin{lemma}[Generalized triangle inequalities]
    \label{lem:triangle}
    Let $a,b\ge 0$ and $\epsilon\in (0,1)$, then for $z\ge 1$, the following holds.
    \begin{itemize}
        \item [1.] (Lemma A.1 of~\cite{Makarychev19Performance})$(a+b)^z\le (1+\epsilon)^{z-1}\cdot a^z + (1+\frac{1}{\epsilon})^{z-1}\cdot b^z$;
        \item [2.] (Claim 5 of~\cite{Sohler18Strong}) $(a+b)^z\le (1+\epsilon)\cdot a^z + (\frac{3z}{\epsilon})^{z-1}\cdot b^z$.
    \end{itemize}
\end{lemma}

\section{Reduction I: Density of Datasets}
\label{sec:Reduction-I}

\begin{theorem}
    \label{thm:reduction I}
    Suppose an underlying metric space $M=(V,\dist)$ is given.
    Assume there exists an algorithm $\calA$ that, given $0<\eps<1$, integers $k,z\ge 1$ and an $n$-point dataset from $M$ as input, constructs an $\eps$-coreset for \kzC of size $N(n,k,\eps^{-1})$ in time $T(n,k,\eps^{-1})$.
    Then there is an algorithm that, given $0<\eps<1$, integers $k,z,m\ge 1$, an $n$-point dataset $X$ from $M$ and a $(2^{O(z)},O(1),O(1))$-approximation solution $C^*\subseteq V$ to \kzmC on $X$ as input, constructs an $\eps$-coreset $S$ of $X$ for \kzmC.
    This algorithm invokes $\calA$ a constant number of times and runs in time $\tilde O(nk) + T(n,k,O(\eps^{-1}))$. The constructed coreset $S$ has size 
    $
        2^{O(z\log z)}\cdot O\left(km\eps^{-1}\right) +  N\left(n,k,O(\eps^{-1})\right).
    $
\end{theorem}

\paragraph{Construction Based on Bounded Partition}
As we mention, the reduction relies on a space partition that breaks the input into bounded partition. Here, by ``bounded'', we mean the diameter of each part in the partition is upper bounded by some parameter $\lambda > 0$. 
Formally, we introduce the notion of \emph{$\lambda$-bounded partitions} in \Cref{def:bounded partition}. 
\begin{definition}[$\lambda$-Bounded partition]
    \label{def:bounded partition}
For a dataset $Y \subseteq V$ and  $\lambda\ge 0$, we say a partition $\calP = \{P_i \}_i$ of $Y$ is \emph{$\lambda$-bounded} if for every $P_i \in \calP$, $\diam(P_i)\le \lambda$.
\end{definition}

We describe our algorithm in~\Cref{alg:construction based on bounded partition}, which additionally takes a $\lambda$-bounded partition of $X$ as input for some $\lambda > 0$. We note that the value of $\lambda$ would affect the accuracy of the final coreset, as shown in~\Cref{thm:main1}. Later on, we select an appropriate $\lambda$ to ensure that the output of~\Cref{alg:construction based on bounded partition} is an $\eps$-coreset.
Essentially, our algorithm computes a decomposition of the dataset, based on the $\lambda$-bounded partition $\calP$, into $X_{\calD}$ and $X_{\calS}$, where $X_{\calD}$ contains the points within dense parts, while $X_\calS$ contains those within sparse parts.
Here, ``dense'' and ``sparse'' refer to whether a part contains many or few points, respectively.
Then, it simply constructs a coreset $S_{\calD}$ of $X_{\calD}$ for \kzC and returns $S_{\calD} \cup X_{\calS}$ as the coreset for \kzmC of the dataset.

\begin{algorithm}
    \caption{Coreset construction for \kzmC of $X$ based on bounded partition}
    \label{alg:construction based on bounded partition}
    \begin{algorithmic}[1]
        \Require{a $\lambda$-bounded partition $\calP$ of $X$}
        \State let $\calD\gets\{P\in\calP: |P|\ge (1+\eps^{-1})m\}$ and $\calS\gets \calP\setminus\calD$ \label{alg 1 line:decomposition}
        \State construct an $\eps$-coreset $S_{\calD}$ for \kzC of $X_{\calD}:=\bigcup_{P\in\calD} P$ \label{alg 1 line:coreset for dense subset}
        \\
        \Return $S_{\calD}\cup X_{\calS}$, where $X_{\calS}:=\bigcup_{P\in\calS} P$ 
    \end{algorithmic}
\end{algorithm}

\begin{lemma}[Vanilla coresets on dense datasets are robust]
    \label{thm:main1}
For a dataset $X\subseteq V$, $0<\eps<1$, and integers $k,z,m\ge 1$, if there exists a $\lambda$-bounded partition $\calP$ of $X$ for some $\lambda\ge 0$ such that for every $P\in\calP$, $|P|\ge (1+\eps^{-1})m$, then an $\eps$-coreset of $X$ for \kzC
    is also an $\left(O(\eps), 2^{O(z\log(z+1))}\cdot m\lambda^z \right)$-coreset for \kzmC.
\end{lemma}
\begin{proof}
    The proof can be found in~\Cref{sec:first reduction}.
\end{proof}
The intuition behind~\Cref{thm:main1} stems from the observation that,
for every point $x$, there are at least $\Omega(\epsilon^{-1}m)$ points
lie in the same part with $x$ and that these points are ``close'' to $x$ (because of the $\lambda$-bounded property). This implies that the contribution for every $x$ (for any center set $C$) is only roughly $\epsilon / m$ times the total cost of all points.
Hence, removing up to $m$ points as outliers only change the objective cost 
by $O(\epsilon)$ times.
This roughly means that the cost of outlier version and the vanilla version are very close. Since the vanilla coreset preserves the cost of the vanilla version, it naturally preserves the outlier version according to this observation.

Hence, if we are given a $\lambda$-bounded partition $\calP$ with sufficiently small $\lambda$, then $S_{\calD}$ becomes an $O(\eps)$-coreset of $X_{\calD}$ for \kzmC (with negligible additive error). By the composability of coresets and treating $X_{\calS}$ as a naive coreset of itself, we conclude that $S = S_{\calD}\cup X_{\calS}$ is an adequately accurate coreset of $X = X_{\calD}\cup X_{\calS}$ for \kzmC.

\paragraph{Size Bound: Construction of $\lambda$-Bounded Partition}
As shown in~\Cref{alg:construction based on bounded partition}, the resulting coreset $S$ contains $|S_{\calD}| + |X_{\calS}| = N(n,k,O(\eps^{-1})) + |X_{\calS}|$ points, where $|X_{\calS}|$ is the number of points contained in sparse parts. 
To obtain an upper bound for the size of the coreset, we show in \Cref{lem:decomposition} an algorithm that finds an \emph{almost-dense} $\lambda$-bounded partition (by ``almost-dense'' we mean the number of points in the sparse parts is small).

\begin{lemma}[Almost-dense bounded partition]
    \label{lem:decomposition}
Given $\lambda>0$, a dataset $X \subseteq V$ and an $(\alpha,\beta,\gamma)$-approximation solution $C^*$ to \kzmC on $X$, one can construct in $\tilde O(\beta n k)$ time a $2\lambda$-bounded partition $\calP$ of $X$. This partition can be decomposed into $\calP = \calD \cup \calS$ with the following properties:
    \begin{enumerate}
        \item  (Dense parts) Every part in $\calD$ is dense, i.e., for every $P\in\calD$, $|P|\ge (1+\eps^{-1})m$.
        \item  (Sparse parts) The number of sparse points is small, i.e., $|\bigcup_{P\in\calS} P|\le O(\beta k m\eps^{-1}) + \gamma m + \cost_z^{(\gamma m)}(X,C^*)\cdot \lambda^{-z}$. 
    \end{enumerate}

\end{lemma}

One may find the bound of the number of sparse points somewhat unnatural since it relates the number of points to the seemingly incomparable clustering objective. As mentioned earlier, in our applications we would set $\lambda$ sufficiently small, say, $\lambda \le \left((m^{-1}\cost_z^{(\gamma m)}(X,C^*)\right)^{1/z}$, so that the $\cost_z^{(\gamma m)}$ term is ``canceled out''.

\begin{proof}[Proof of \Cref{lem:decomposition}]
    We start with the construction.
    Let $L\subseteq X$ with $|L| = \gamma m$ denote the $\gamma m$-outliers set of $X$ with respect to $C^*$.
    Define $F := L\cup \{x\in X: \dist(x,C^*) > \lambda \}$,
    which includes all $\gamma m$ outliers and those inliers, i.e., $X \setminus L$, with distances more than $\lambda$ from $C^*$.
    Let $P_1,\dots, P_{\beta k}$ denote the clustering of $X\setminus F$ with respect to $C^*$.
    Finally, let $\calD:=\{P_i:i\in [\beta k], |P_i| \ge (1+\eps^{-1})m\}$ and $\calS:=\{P_i:i\in[\beta k], |P_i|<(1+\eps^{-1})m \}\cup \{\{x\}: x\in F \}$.

    From the construction, it is immediate that $\calD \cup \calS$ is a partition of $X$, and $\calD$ contains only dense parts. 
Since $F$ contains all points that have distance more than $\lambda$ from $C^*$,
    every point in any part of $\calD \cup \calS$ has distance at most $\lambda$ to $C^*$. This implies that $\calD\cup\calS$ is $2\lambda$-bounded by the triangle inequality.
It remains to prove the second property, i.e., the bounded number of sparse points.

    Note that $\calS$ contains two types of subsets: sparse subsets and singletons. For the sparse subsets $\{P_i: i\in [\beta k], |P_i|< (1+\eps^{-1})m\}$, the sum of their sizes is at most $\beta k \times (1+\eps^{-1})m = O(\beta k m \eps^{-1})$. As for the singletons, each corresponds to a point in $F$. Hence, we need to bound the size of $F$.

    Notice that $X \setminus L$ is the set of inliers.
    Let $\kappa>0$ denote the number of far inliers,
    i.e., inliers that are with distance more than $\lambda$ from $C^*$.
    Then we have $|F| \le |L| + \kappa = \gamma m + \kappa$. For the value of $\kappa$, we have 
    \begin{equation*}
        \kappa\cdot \lambda^z \le \sum_{x\in F\setminus L}\dist^z(x,C^*)\le  \cost_z^{(\gamma m)}(X,C^*),
    \end{equation*}
    which implies that $\kappa\le \cost_z^{(\gamma m)}(X,C^*)\cdot \lambda^{-z}$.
    Therefore, $|F|\le \gamma m + \cost_z^{(\gamma m)}(X,C^*)\cdot \lambda^{-z}$.
    This completes the proof.
\end{proof}

\paragraph{Proof of~\Cref{thm:reduction I}} 
Recall that $C^*$ is a given $(\alpha,\beta,\gamma)$-approximation solution to \kzmC on $X$, where $\alpha = 2^{O(z)}$ and $\beta = \gamma = O(1)$.
Let $\lambda:= (z+1)^{-\xi}\cdot \left(\frac{\eps\cdot \cost_z^{(\gamma m)}(X,C^*)}{\alpha m} \right)^{1/z}$ for sufficiently large constant $\xi > 0$. 
We first apply \Cref{lem:decomposition} to compute a $\lambda$-bounded partition $\calP=\calD\cup\calS$ of the dataset $X$, and then apply~\Cref{alg:construction based on bounded partition} to this partition $\calP$, which outputs a weighted set $S = S_{\calD}\cup X_{\calS}$.
Guaranteed by~\Cref{thm:main1}, $S_{\calD}$ is an $(O(\eps),\eta)$-coreset of $X_{\calD}$ for \kzmC with additive error
\begin{equation*}
    \eta\le 2^{O(z\log(z+1))}\cdot m\lambda^z\le \eps\cdot\alpha^{-1}\cdot\cost_z^{(\gamma m)}(X,C^*)\le \eps\OPT_z^{(m)}(X).
\end{equation*}
Since $X_{\calS}$ is an $\eps$-coreset of itself for \kzmC, by composability (\Cref{fact:composability}), we have that $S = S_{\calD}\cup X_{\calS}$ is an $(O(\eps),\eta)$-coreset of $X$ for \kzmC. Then, consider any $k$-point center set $C\subseteq V$ and any $0\le h\le m$, we have 
\begin{equation*}
    \left|\cost_z^{(h)}(X,C) - \cost_z^{(h)}(S,C) \right|\le O(\eps)\cdot \cost_z^{(h)}(X,C) + \eps\cdot\OPT_z^{(m)}(X)\le O(\eps)\cdot \cost_z^{(h)}(X,C),
\end{equation*}
hence certifying that $S$ is an $O(\eps)$-coreset of $X$ for \kzmC.
According to~\Cref{lem:decomposition}, the coreset size is $N(n,k,\eps^{-1}) + |X_{\calS}|\le N(n,k,\eps^{-1}) + O(\gamma + \beta k\eps^{-1} + 2^{O(z\log z)}\alpha\eps^{-1})\cdot m = N(n,k,\eps^{-1}) + 2^{O(z\log z)}\cdot O(km\eps^{-1})$. 
Regarding the runtime, our construction is merely the combination of 
the computation of a bounded partition (e.g., algorithm of~\Cref{lem:decomposition}) and~\Cref{alg:construction based on bounded partition} (with the bottleneck of runtime being the vanilla coreset construction).
Hence, the total running time is $\tilde O(nk) + T(n,k,\eps^{-1})$. To finish the proof, it remains to scale $\eps$.
\qed

      \subsection{Vanilla Coresets on Dense Datasets Are Robust: Proof of~\Cref{thm:main1}}
\label{sec:first reduction}

In this section, we prove~\Cref{thm:main1}, which demonstrates that a vanilla coreset on a dense dataset is also a robust coreset. As mentioned earlier, the proof is based on the observation that the contribution of every point is negligible. We formulate this in the following technical lemma.

\begin{lemma}
    \label{lem:small dist}
For a dataset $X\subseteq V$, if there exists a $\lambda$-partition $\calP$ of $X$ such that for every $P\in\calP$, $|P|\ge (1+\eps^{-1})m$, then
    for any subset $C\subset V$ and for every $x\in X$, it holds that 
    \begin{equation*}
        \dist^z(x,C)\le \frac{2\eps}{m}\cdot \cost_z^{(m)}(X,C) + (3z\lambda)^z.
    \end{equation*}
\end{lemma}
\begin{proof}
    Fix a center set $C$ and a point $x\in X$, let $L$ denote $m$-outlier set of $X$ with respect to $C$ and let $P\in \calP$ denote the part that contains $x$, i.e., $x\in P$. Since $|P|\ge (1+\eps^{-1})m$ and $\diam(P)\le \lambda$, we have
    \begin{equation*}
        \left|\ball_X(x,\lambda)\setminus L \right|\ge \left|P\setminus L \right|\ge \eps^{-1}m,
    \end{equation*} 
    which implies that there exists a point $y\in P\setminus L$ such that $\dist^z(y,C)\le \frac{\eps}{m}\cdot\cost_z^{(m)}(X,C)$.
    Then by the generalized triangle inequality (see~\Cref{lem:triangle}), we have 
    \begin{align*}
        \dist^z(x,C)\quad&\le\quad 2\cdot \dist^z(y,C) + \left(3z \right)^{z-1} \dist^z(x,y)\\
        &\le\quad \frac{2\eps}{m}\cdot \cost_z^{(m)}(X,C) +  \left(3z\lambda \right)^z,
    \end{align*}
    which completes the proof of \Cref{lem:small dist}.
\end{proof}

\begin{proof}[Proof of~\Cref{thm:main1}]
    It suffices to prove for every $k$-point center set $C\subseteq V$ and every $0\le h\le m$,
    \begin{equation}
        \label{eq:temp7}
        \left|\cost_z^{(h)}(X,C) - \cost_z^{(h)}(S,C) \right|\le O(\eps)\cdot\cost_z^{(h)}(X,C) + 2^{O(z\log(z+1))}\cdot m\lambda^z.
    \end{equation}
    Since $S$ is an $\eps$-coreset of $X$ for \kzC, it holds that 
    \begin{equation}
        \label{eq:dense coreset}
        \left|\cost_z(X,C) - \cost_z(S,C)\right|\le \eps \cdot\cost_z(X,C).
    \end{equation}
    We then prove the following two inequalities, which, combining with \eqref{eq:dense coreset}, directly imply~\eqref{eq:temp7}.
    \begin{equation}
        \label{eq:coreset 1}
        \left|\cost_z^{(h)}(X,C) - \cost_z(X,C) \right|\le O(\eps)\cdot \cost_z^{(m)}(X,C) + 2^{O(z\log(z+1))}\cdot m\lambda^z, 
    \end{equation}
    \begin{equation}
        \label{eq:coreset 2}
        \left|\cost_z^{(h)}(S,C) - \cost_z(S,C) \right|\le O(\eps)\cdot \cost_z^{(m)}(X,C) + 2^{O(z\log(z+1))}\cdot m\lambda^z.
    \end{equation}
    Here, we only provide the proof of~\eqref{eq:coreset 1}, and the proof of~\eqref{eq:coreset 2} is similar.
    Let the set $L$ denote an $h$-outlier set of $X$ with respect to $C$, then it holds that 
    \begin{align*}
        &\quad \cost_z(X,C)\\
        \ge&\quad \cost_z^{(h)}(X,C)\\
        =&\quad \cost_z(X,C) - \cost_z(L,C)\\
        =&\quad\cost_z(X,C) - \sum_{y\in L}w_L(y)\cdot\dist^z(y,C)\\
        \ge&\quad \cost_z(X,C) - \sum_{y\in L}w_L(y)\cdot \left(\frac{2\eps}{m}\cost_z^{(m)}(X,C) + (3z\lambda)^z \right) & \text{(by~\Cref{lem:small dist})}\\
        \ge&\quad \cost_z(X,C) - O(\eps)\cdot \cost_z^{(m)}(X,C) - m\cdot (3z\lambda)^z,
    \end{align*}
    which proves~\eqref{eq:coreset 1}. Once we obtain~\eqref{eq:coreset 1} and~\eqref{eq:coreset 2}, we have 
    \begin{align*}
        &\quad\left|\cost_z^{(h)}(X,C) - \cost_z^{(h)}(S,C) \right|\\
        \le&\quad \left|\cost_z(X,C) - \cost_z(S,C) \right| + \left|\cost_z^{(h)}(X,C) - \cost_z(X,C) \right| + \left|\cost_z^{(h)}(S,C) - \cost_z(S,C) \right|\\
        \le&\quad \eps\cdot \cost_z(X,C) + O(\eps)\cdot \cost_z^{(m)}(X,C) + 2^{O(z\log (z+1))}\cdot m\lambda^z\\
        \le&\quad O(\eps)\cdot \cost_z^{(m)}(X,C) + 2^{O(z\log(z+1))}\cdot m\lambda^z\\
        \le&\quad O(\eps)\cdot \cost_z^{(h)}(X,C) + 2^{O(z\log(z+1))}\cdot m\lambda^z,
    \end{align*}
    which completes the proof.
\end{proof}

\section{Reduction II: Size-preserving Property}
\label{sec:reduction2}

We present an alternative reduction in \Cref{thm:reduction2}.
The main difference to \Cref{thm:reduction I} is to avoid the multiplication of $k$ and $m$ in the coreset size.
This reduction requires a slightly stronger technical guarantee from the vanilla coreset algorithm,
such that the algorithm not only needs to construct a small coreset for metric $M$ that contains the dataset,
but also for a family of \emph{separated duplications} of $M$.
Roughly speaking, a separated duplication of $M$ ``copies'' $M$ into several disjoint sub-metrics, and makes the distance between two copies large.  We provide the formal definition below.

\begin{definition}[Seperated duplication of a metric space]
    \label{def:separated duplication}
    Given a metric space $M=(V,\dist)$, for real number $w\ge 0$ and integer $h\ge 1$, a metric space $M'=(V\times [h],\dist')$ is called an $w$-separated $h$-duplication of $M$, if 
    \begin{enumerate}
        \item $\forall x,y\in X, i\in [h]$, it holds that $\dist'((x,i),(y,i)) = \dist(x,y)$; and
        \item $\forall x,y\in X, i,j\in [h]$ such that $i\neq j$, it holds that $\dist'((x,i),(y,j))\ge \max\{w,\dist(x,y)\}$.
    \end{enumerate}
    For the special case of $h = 1$ and $w = 0$, we say $M'$ is a $w$-separated $h$-duplication of $M$ if and only if $M'=M$.
\end{definition}

\begin{theorem}
    \label{thm:reduction2}
    Suppose an underlying metric space $M=(V,\dist)$ and a family $\calM^\dup=\left\{M_{h,w}^\dup\right\}_{h\ge 1,w\ge 0}$ of $w$-separated $h$-duplication of $M$ are given. 
Assume there exists an algorithm $\calA$ such that, for every $0<\eps<1$, integers $k,z\ge 1$, metric $M^\dup\in\calM^\dup$ and  $n$-point dataset from $M^\dup$, it runs in time $T(M^\dup, n, k, \eps^{-1})$ to construct an $\eps$-coreset of size $N(M^\dup, n, k, \eps^{-1})$ for \kzC on $M^\dup$.
Then, there is an algorithm that, given $0<\eps<1$, integers $k,z,m\ge 1$, an $n$-point dataset $X\subseteq V$
    and a $(2^{O(z)},O(1),O(1))$-approximation $C^*\subseteq V$ to \kzmC on $X$ 
    as input, 
    constructs $\eps$-coreset for \kzmC of size 
    \begin{equation*}
        A + N(M,n,k,O(\eps^{-1})) + N\left(M^\dup_{h',w'}, O(km\eps^{-1}), O\left(k\log^2(km\eps^{-1})\right), O(\eps^{-1})\right),
    \end{equation*}
    where 
        $A = 2^{O(z\log z)}\cdot O\left(m \eps^{-2z} \log^z(km\eps^{-1})\right)$,
    $h' = \min\{O(k\log(km\eps^{-1})),n\}$ and $w' = O(z\eps^{-1}\cdot \diam(X)\cdot n^{1/z})$.
    This algorithm invokes $\calA$ a constant number of times and runs in time 
    \begin{equation*}
        \tilde O(nk) + \poly(km\eps^{-1}) + T(M,n,k,\eps^{-1}) + T\left(M^\dup_{h',w'}, h', O\left(k\log^2(k m\eps^{-1})\right),O(\eps^{-1})\right).
    \end{equation*}

\end{theorem}

Notice that we assume a family $\calM^\dup$ is \emph{given} alongside $M$ in \Cref{thm:reduction2},
since this family $\calM^\dup$ may depend on the specific type of $M$ (e.g., $M$ is Euclidean or shortest-path metric of a graph)
and may require some careful design.
Luckily, in most cases the separated duplication family $\calM^\dup$ needs not be much more ``complex'' than $M$,
and this particularly means as long as a vanilla coreset works on $M$, it can also deal with metrics in $\calM^\dup$.

Specifically, our main claim is that for metrics including Euclidean spaces, general finite metrics, doubling metrics and shortest-path metrics of minor-excluded graphs,
the complexity/dimension parameter (such as Euclidean dimension) of the separated duplication
only increases by a constant (factor).
For instance, for the $d$-dimension Euclidean metric $M = (\mathbb{R}^d, \ell_2)$,
we show that for every $w$ and $h$ there exists a $w$-separated $h$-duplication of $M$ that can be realized by a
$(d + 1)$-dimensional Euclidean space $(\mathbb{R}^{d + 1}, \ell_2)$ (\Cref{lem:duplication for Euclidean case}).
Hence, a vanilla coreset construction for those metrics automatically works for the separated duplication (with only a constant factor increase in coreset size).
We summarize the results for the separated duplication in \Cref{tab:separated_duplication}.

\begin{table}[ht]
    \centering
    \caption{Summary of the results for the separated duplication. For all cases of the original metric space $M$ that we list, we demonstrate the existence of a family $\calM^\dup$, such that each $M^\dup \in \calM^\dup$ possesses the property as shown in the table.}
    \label{tab:separated_duplication}
    \begin{tabular}{@{}cccc@{}}
    \toprule
    \multicolumn{2}{c}{space $M$}                        &  property of $\calM^\dup$ & reference \\ \midrule
    \multicolumn{2}{c}{$(\R^d, \ell_p)$ for $p\ge 1$}                                 & can be realized by $(\R^{d+1}, \ell_p)$     &   \Cref{lem:duplication for Euclidean case}        \\ \cmidrule(r){1-2}
    \multirow{2}{*}{general metric} &  doubling dimension $\ddim(M)$ & $\ddim(M^\dup)\le 2\ddim(M)+2$              &     \Cref{lem:duplication doubling}      \\
                                    &  ambient size $n$              &  ambient size $\le n^2$                 & \Cref{remark:finite}          \\ \cmidrule(r){1-2}
    \multirow{2}{*}{graph metric}   &  treewidth $\tw$               &  treewidth $\tw$                        &      \Cref{lem:duplication graph}     \\
                                    & excluding a fixed minor $H$        & excluding the same minor $H$                & \Cref{lem:duplication graph}          \\ \bottomrule
    \end{tabular}
    \end{table}

\subsection{Proof of~\Cref{thm:reduction2}}
\label{sec:reduction2_proof}

Recall that in~\Cref{thm:reduction I}, the additive term $O_z(km\eps^{-1})$ in the coreset size corresponds to the number of sparse points, i.e. $X_{\calS}$, since we directly add these points into the coreset.
In this proof, we still follow the steps of \Cref{thm:reduction I}, except that we use a more refined method to construct a coreset for $X_{\calS}$.

We do not attempt to leverage any structural property of $X_{\calS}$ since it may be quite arbitrary.
Instead, we suggest a new framework:
we show in \Cref{cor:second} that as long as a vanilla coreset construction algorithm additionally satisfies the ``size-preserving'' property,
then it is as well a coreset for clustering with outliers.
Of course, this size-preserving property is nontrivial, and we cannot simply assume a vanilla coreset algorithm to satisfy this property.
Hence, another important step (described in \Cref{alg:size-preserving}) is a general black-box reduction (albeit it requires the vanilla coreset algorithm also works for separated duplication of the underlying metric),
that turns a generic vanilla coreset algorithm into the one that is size-preserving.

Actually, the abovementioned reduction steps work for a general dataset (not only for $X_{\calS}$).
However, the caveat is that this reduction may lead to a $\poly\log(n)$ factor in the final coreset size where $n$ is the size of the dataset.
This is in general not acceptable, but luckily, since we only need to apply this reduction on $X_{\calS}$ which has only $O(km\epsilon^{-1})$ points,
this $\poly\log$ factor becomes negligible.
Indeed, the only special property we use from $X_{\calS}$ is that it has a small number of points.

\paragraph{Size-preserving (Vanilla) Coresets Are Robust}
We define a coreset (or more generally, a weighted set) $S \subseteq X$ to be size-preserving with a diameter bound $\lambda$ if there exists a $\lambda$-bounded partition $\calP$ of $X$ such that, for every $P \in \calP$, $w_S(P \cap S) = |P|$. Similar to~\Cref{thm:main1}, we establish a relationship between the error of $S$ for being a robust coreset and the diameter bound of the partition $\calP$.

\begin{restatable}[Size-preserving vanilla coresets are robust]{lemma}{SizePreservingAreRobust}
    \label{cor:second}
    Suppose a metric $M=(V,\dist)$ is given.
For $0<\eps<1$, $\lambda>0$, integers $k,z,m\ge 1$ and a dataset $X\subseteq V$, if a weighted set $S\subseteq X$ satisfies that there exists a $\lambda$-bounded partition $\calP$ of $X$ such that
    \begin{enumerate}
        \item $S$ is an $\eps$-coreset of $X$ for \tzC{k+|\calP| }, and 
        \item  $\forall P\in\calP$, $|P| = w_S(S\cap P)$,
    \end{enumerate}
    then $S$ is an $\left(O(\eps), 2^{O(z\log (z+1))}\cdot m\lambda^z\eps^{1-z}\right)$-coreset of $X$ for \kzmC.
\end{restatable}

\begin{proof}
    The proof can be found in \Cref{sec:new proof}.
\end{proof}
\paragraph{Reduction from Vanilla Coresets to Size-preserving Coresets}
We next show how to turn a vanilla coreset to a stronger vanilla coreset that also satisfies the size-preserving property.
Our algorithm needs a metric decomposition called \emph{sparse partition}, which is introduced by~\cite{Jia05Universal}.
A sparse partition is a partition of the metric/dataset,
such that each part has bounded diameter and the partition additionally satisfies a sparsity property.
This sparsity property requires that any metric ball of small radius intersects only a few parts.
We notice that this sparsity property is different from the previously mentioned sparsity as in the almost-dense decomposition, which instead requires that each part contains only a few points and is not about intersection.

We provide a formal definition and a specific version that we would use~\cite{Jia05Universal}, restated using our language as follows.

\begin{definition}[Sparse partition~\cite{Jia05Universal}]
    \label{def:sparse partition}
    Given a metric space $M=(V,\dist)$ and a subset $X\subseteq V$, we say a partition $\calP$ of $X$ is a \emph{$(\mu,\Gamma,\Lambda)$-sparse partition} if it holds that 
    \begin{itemize}
        \item[a)] (Bounded diameter) for every part $P\in \calP$, $\diam(P)\le \mu$; and
        \item[b)] (Sparsity) for every $x\in X$, the ball $\ball_X(x,\mu/\Gamma)=\{y\in X:\dist(x,y)\le \mu/\Gamma\}$ intersects at most $\Lambda$ parts in $\calP$, i.e., $|\{P\in\calP: P\cap \ball_X(x,\mu/\Gamma)\neq\emptyset\}|\le \Lambda$.
    \end{itemize}
\end{definition}

\begin{theorem}[\cite{Jia05Universal}]
    \label{thm:sparse partition}
    Given a metric space $M=(V,\dist)$, an $n$-point dataset $X\subseteq V$ and $\mu > 0$, there is a $(\mu,\Gamma,\Lambda)$-sparse partition of $X$ with $\Gamma = O(\log n), \Lambda = O(\log n)$ that can be computed in time $\poly(n)$.
\end{theorem}

We describe our algorithm in~\Cref{alg:size-preserving}, which takes as input an $n$-point dataset $X\subseteq V$, a real number $\mu>0$ representing the diameter bound, and an integer $k'\geq 1$ specifying the number of centers that a coreset can handle.
The algorithm first computes a sparse partition $\calQ=(X_1,\dots,X_l)$ of $X$ using \Cref{thm:sparse partition}. Then, it maps the dataset $X$ into some $M^\dup_{l,w}\in \calM^{\dup}$ with sufficiently large $w$ to ensure that the parts of $\calQ$ are well-separated. 
This can be achieved by mapping each point $x\in X_i$ to $(x,i)$. Guaranteed by the separation property of $M^\dup_{l,w}$, the distance between any pair of parts is at most $w$. We note that this well-separation property is crucial for ensuring the size-preserving property.
Finally, the algorithm constructs a coreset for \tzC{k'} on $M^\dup_{l,w}$, and returns the pre-image of this coreset.
We have the following lemma that demonstrates the correctness of \Cref{alg:size-preserving}.

\begin{algorithm}[ht]
    \caption{Reduction from size-preserving vanilla coresets to vanilla coresets}
    \label{alg:size-preserving}
    \begin{algorithmic}[1]
        \Require{an $n$-point dataset $X\subseteq V$, real number $\mu>0$ and an integer $k'\ge 1$}
        \State compute a $(\mu,\Gamma,\Lambda)$-sparse partition $\calQ=(X_1,\dots,X_l)$ of $X$ \label{alg line:sparse partition} \Comment{use algorithm of~\Cref{thm:sparse partition}}
        \State for every $i\in [l]$, let $X_i':=\{(x,i)\in V\times [l]: x\in X_i\}$, and let $X':=\bigcup_{i=1}^l X_i'$
        \label{alg line:mapping}
        \State construct an $\eps$-coreset $S'$ of $X'$ for \tzC{k'} on metric $M^\dup_{l,200z\eps^{-1}\cdot \diam(X)\cdot n^{1/z}}\in \calM^\dup$ \label{alg line:vanilla coreset}

        \Comment{use the assumed algorithm as a black-box}

        \State let $S:=\{x\in X: \exists i\in[l], (x,i)\in S'\}$, and define $w_S:S\to\R_{\ge 0}$ such that 
        $
            \forall (x,i)\in S, w_S(x) = w_{S'}((x,i))
        $
        \label{alg line:preimage}
        \\
        \Return $(S,w_S)$ and the partition $\calQ$
    \end{algorithmic}
\end{algorithm}
\begin{restatable}[Correctness of~\Cref{alg:size-preserving}]{lemma}{ReductionSizePreserving}
    \label{lem:size-preserving}
    If there exists a $\frac{\eps\mu}{1000z\Gamma}$-bounded partition of $X$ of size $t$ ($t\ge 1$) and $k'$ satisfies that $k' \ge (k+t)\Lambda$, then~\Cref{alg:size-preserving} returns a weighted set $S\subseteq X$ and a $\mu$-bounded partition $\calQ = (X_1,\dots,X_l)$ of $X$ with $l\le t\Lambda$ such that
    \begin{enumerate}
        \item $S$ is an $\eps$-coreset of $X$ for \kzC, and
        \item $\forall i\in [l]$, $w_S(S\cap X_i)\in (1+\eps)\cdot |X_i|$.
    \end{enumerate}
\end{restatable}
\begin{proof}
    The proof can be found in \Cref{sec:reduction for size_preserving}.
\end{proof}

We note that the size-preserving property guaranteed by \Cref{lem:size-preserving} is actually weaker than what we need, as it only approximately preserves the sizes. However, we argue that this is still sufficient because we can calibrate the weights of this coreset, which only affects the accuracy by a factor of $1+\eps$.

As shown in~\Cref{lem:size-preserving}, the lower bound of $k'$ relies on two factors: the minimum size of the bounded partition of $X$ and the parameter $\Lambda$ of the sparse partition $\calQ$. Therefore, it is essential for both of these factors to be as small as possible.
According to~\Cref{lem:decomposition}, an $O(k)$-sized bounded partition exists for the majority of points. This suggests applying \Cref{alg:size-preserving} only to this majority and directly incorporating the remaining points into the coreset.
Regarding the parameter $\Lambda$, we assert that the result of~\Cref{thm:sparse partition}, which provides a sparse partition with $\Lambda = O(\log n)$, is sufficient for our needs. This is because we will apply the algorithm to the small subset, say, of size $\poly(km\eps^{-1})$, rendering such $\log n$ factor negligible.

\paragraph{Concluding~\Cref{thm:reduction2}} 

Recall that we are given an $(\alpha,\beta,\gamma)$-approximation $C^*$ to \kzmC on $X$, where $\alpha = 2^{O(z)},\beta = \gamma = O(1)$. We first apply~\Cref{lem:decomposition} to $X$ with the diameter bound 
\begin{equation*}
    \lambda:= (z+1)^{-\xi}\cdot \eps^{2}\cdot \left(\frac{ \cost_z^{(\gamma m)}(X,C^*)}{\alpha m} \right)^{1/z}\cdot \left(\log(km\eps^{-1})\right)^{-1}
\end{equation*}
for sufficiently large constant $\xi > 0$. 
Let $\calP = \calD\cup \calS$ denote the $\lambda$-bounded partition computed by the algorithm of~\Cref{lem:decomposition}, where for every $P\in\calD$, $|P|\ge (1+\eps^{-1})m$. Let $X_{\calD}:=\bigcup_{P\in\calD} P$.
We construct an $\eps$-coreset $S_{\calD}$ of $X_{\calD}$ for \kzC on the original metric space $M$ using the assumed algorithm as a black-box. Guaranteed by~\Cref{thm:main1}, $S_\calD$ is also an $\left(O(\eps), O(\eps)\cdot \OPT_z^{(m)}(X)\right)$-coreset of $X_{\calD}$ for \kzmC.

As for the sparse subsets $\calS$, we further decompose it into $\calS_1$ and $\calS_2$, where $|\calS_1|\leq \beta k$ and $\calS_2$ covers only a few points. Such a decomposition always exists. Moreover, we claim that the partition computed by \Cref{lem:decomposition} already provides this, as stated in the following claim.

\begin{claim}
    \label{claim:refined decopmosition}
    Given a metric $M=(V,\dist)$, a integer $\lambda >0$, a dataset $X\subseteq V$ and an $(\alpha,\beta,\gamma)$-approximation solution $C^*$ to \kzmC of $X$, let $\calP=\calD\cup\calS$ be a $\lambda$-bounded partition computed using~\Cref{lem:decomposition}, the sparse subsets $\calS$ can be further decomposed into $\calS_{1}$ and $\calS_{2}$ such that $|\calD| + |\calS_1|\le \beta k$ and $|\bigcup_{P\in\calS_2} P|\le \gamma m + \cost_z^{(\gamma m)}(X,C^*)\cdot \lambda^{-z}$.
\end{claim}
The claim follows directly from the construction and the size analysis of $\calS$ in the proof of \Cref{lem:decomposition}. Provided by this decomposition, let $X_{\calS,1}:=\bigcup_{P\in\calS_1} P$ and $X_{\calS,2}:=\bigcup_{P\in\calS_2} P$, we have $|X_{\calS,1}|\le O(km\eps^{-1})$ and $X_{\calS,1}$ admits a $\lambda$-bounded partition of size at most $\beta k$.

Hence, we apply~\Cref{alg:size-preserving} to $X_{\calS,1}$ with $\mu = 1000\eps^{-1}z\Gamma\lambda$ and $k' = (k + \beta k\Lambda + \beta k)\Lambda$, where $\Gamma = \Lambda = O(\log |X_{\calS,1}|) = O(\log(km\eps^{-1}))$ as stated in~\Cref{thm:sparse partition}. By~\Cref{lem:size-preserving},~\Cref{alg:size-preserving} returns a weighted set $S\subseteq X_{\calS,1}$ and a $\mu$-bounded partition $\calQ$ of $X$ with $|\calQ|\le \beta k\Lambda$ such that $S$ is an $\eps$-coreset of $X_{\calS,1}$ for \tzC{k+\beta k\Lambda}, and for all $Q\in \calQ$, $w_S(S\cap Q)\in (1+\eps)\cdot |Q|$.

We calibrate the weight of $S$ to meet the exact size-preserving requirement of~\Cref{cor:second} as follows: we define a new weight function $w_S':S\to \R_{\geq 0}$ such that, for every $Q\in\calQ$ and every $x\in S\cap Q$, $w_S'(x) = w_S(x)\cdot \frac{|Q|}{w_S(S\cap Q)}$. It is easy to verify that the weighted set $(S,w_{S}')$ is size-preserving with respect to $\calQ$. Furthermore, $w_S(S\cap Q)\in (1\pm \eps)\cdot |Q|$ implies that for every $x\in Q$, $w'_S(x)\in (1\pm O(\eps))\cdot w_S(x)$.
Therefore, for every $C\in V^k$, we have
\begin{equation*}
    \sum_{x\in S}w'_S(x)\cdot \dist^z(x,C) \in (1\pm O(\eps))\cdot \sum_{x\in S}w_S(x)\cdot \dist^z(x,C).
\end{equation*}
Since $\sum_{x\in S}w_S(x)\cdot \dist^z(x,C)\in (1\pm \eps)\cdot \cost_z(X_{\calS,1},C)$ due to the fact that $(S,w_S)$ is an $\eps$-coreset, we conclude that $(S,w_S')$ is an $O(\eps)$-coreset of $X_{\calS,1}$ for \kzC.

We then use the function $w_S'$ to weight $S$ instead of using $w_S$. Therefore, $S$ is an $\eps$-coreset for \tzC{k+|\calQ|} that is size-preserving with respect to $\calQ$.
By~\Cref{cor:second}, $S$ is an $\left(O(\eps), 2^{O(z\log (z+1))}\cdot m\mu^z\eps^{1-z}\right)$-coreset of $X_{\calS,1}$ for \kzmC, where the additive error is bounded by 
\begin{equation*}
    2^{O(z\log (z+1))}\cdot m\mu^z\eps^{1-z} = 2^{O(z\log(z+1))}\cdot m\eps^{1-2z}\cdot
    \lambda^z\cdot \log^z(km\eps^{-1})\le \eps\cdot \OPT_z^{(m)}(X),
\end{equation*}
since $\lambda \le (z+1)^{-\xi}\cdot \eps^{2}\cdot \left(\frac{ \OPT_z^{(m)}(X)}{m} \right)^{1/z}\cdot \left(\log(km\eps^{-1})\right)^{-1}$ with $\xi > 0$ being a sufficiently large constant.

Finally, by the composability of coresets (see~\Cref{fact:composability}), $S_{\calD}\cup S\cup X_{\calS,2}$ is an $O(\eps)$-coreset of $X$ for \kzmC, and the coreset size is 
\begin{equation*}
    N(M,n,k,O(\eps^{-1})) + N\left(M^\dup_{l,200z\eps^{-1}\cdot \diam(X)\cdot n^{1/z}}, |X_{\calS,1}|, k', O(\eps^{-1})\right) + |X_{\calS,2}|,
\end{equation*}
where $|X_{\calS,1}| = O(km\eps^{-1})$ and $k'=(k+\beta k\Lambda + \beta k)\Lambda = O(k\log^2(km\eps^{-1}))$.
By~\Cref{claim:refined decopmosition}, the set $X_{\calS,2}$ has a size of 
\begin{equation*}
    \gamma m + \cost_z^{(\gamma m)}(X,C^*)\cdot \lambda^z\le 2^{O(z\log z)}\cdot O\left(m \eps^{-2z}\cdot \log^z(km\eps^{-1})\right).
\end{equation*}

Regarding the running time, the construction is a continuation of that of \Cref{thm:reduction I}, and the additional runtime is primarily due to the invocation of the vanilla coreset construction on $M^\dup_{l,200z\eps^{-1}\cdot \text{diam}(X)\cdot n^{1/z}}$. Hence, the time complexity shown in \Cref{thm:reduction2} follows.

It remains to scale $\eps$ by a constant.\qed

      \subsection{Size-preserving Vanilla Coresets Are Robust: Proof of~\Cref{cor:second}}
\label{sec:new proof}

\SizePreservingAreRobust*

Let $\calP=(X_1,\dots, X_t)$ for some $1\le t\le n$ be a $\lambda$-bounded satisfying the premise of \Cref{cor:second}.
For every $i\in [t]$, let $S_i:=S\cap X_i$ be a weighted set with weight function $w_{S_i}$ such that for every $x\in S_i$, $w_{S_i}(x) = w_S(x)$.
Recall that the weighted set $S$ preserves the size of each part, so we have $w_{S_i}(S_i) = |X_i|$ for every $i\in [t]$.

The proof begins with a sufficient condition for $S$ to be a coreset for \kzmC, as stated in the following claim. This claim is standard, and the proof can be found in the literature of coresets for robust clustering (e.g., \cite{Huang2022Near-optimal,Huang2023General}). For completeness, we provide a proof in~\Cref{sec:composability}.

\begin{restatable}{lemma}{ASufficientCondition}
    \label{claim:a sufficient condition}
    For some $0 < \eps < 1$ and $\eta > 0$,
    if for all $C\in V^k$ and real numbers $h_1,\dots, h_t\ge 0$ with $\sum_{i=1}^t h_t\le m$, it holds that 
    \begin{equation}
        \label{eq:a sufficient condition}
        \left|\sum_{i=1}^t \cost_z^{(h_i)}(S_i,C) -  \sum_{i=1}^t \cost_z^{(h_i)}\cost_z(X_i,C)\right|\le \eps\cdot \sum_{i=1}^t \cost_z^{(h_i)}\cost_z(X_i,C) + \eta,
    \end{equation}
    then $S$ is an $(\eps,\eta)$-coreset (see \Cref{def:additive_coreset}) of $X$ for \kzmC.
\end{restatable}

Therefore, we consider a fixed center set $C\in V^k$ and fixed $h_1,\dots,h_t\ge 0$ with $\sum_{i=1}^th_i\le m$ in the following,
and our goal is to show that
\begin{equation}
    \label{eq:goal}
    \begin{aligned}    
    &\quad \left|\sum_{i=1}^t \cost_z^{(h_i)}(S_i,C) -  \sum_{i=1}^t \cost_z^{(h_i)}\cost_z(X_i,C)\right|\\
    \le&\quad O(\eps)\cdot \sum_{i=1}^t \cost_z^{(h_i)}\cost_z(X_i,C) + 2^{O(z\log (z+1))}\cdot m\lambda^z\eps^{1-z}.
    \end{aligned}
\end{equation}
If this is true, then the proof is finished by applying \Cref{claim:a sufficient condition}.
To prove \eqref{eq:goal},
the key idea is to find an auxiliary center set $C^\aux$ such that 
\begin{equation}
    \label{eq:bound 1}
    \left|\cost_z(X,C\cup C^\aux) - \cost_z(S,C\cup C^\aux) \right|\le \eta_1
\end{equation}
and 
\begin{equation}
    \label{eq:bound 2}
    \left|\left(\sum_{i=1}^t\cost_z^{(h_i)}(X_i,C) - \cost_z(X,C\cup C^\aux) \right) - \left(\sum_{i=1}^t\cost_z^{(h_i)}(S_i,C) - \cost_z(S,C\cup C^\aux) \right) \right|\le \eta_2
\end{equation}
with $\eta_1 + \eta_2\le O(\eps)\cdot \sum_{i=1}^t \cost_z^{(h_i)}\cost_z(X_i,C) + 2^{O(z\log (z+1))}\cdot m\lambda^z\eps^{1-z}$,
then \eqref{eq:goal} follows from triangle inequality.
Hence, in the following, we focus on showing the existence of such $C^\aux$.

\paragraph{Construction of Auxiliary Center Set $C^\aux$}
We explicitly provide the construction of $C^\aux$.
The construction consists of two steps: we first find a (weighted) set of ``significant'' outliers, denoted by $Z$, and then we define the auxiliary center set $C^\aux$ as a $\lambda$-covering of $Z$, which satisfies that $\forall x\in Z$, $\dist(x,C^\aux)\le\lambda$.

To find $Z$, we first identify the outliers and inliers of each $X_i$ with respect to $C$ and $h_i$. Specifically, let $X_\out \subseteq X$ with weight function $w_X^\out:X_\out\to\R_{\ge 0}$ be an outlier set such that for every $i \in [t]$, $X_\out\cap X_i$ is the $h_i$-outlier set of $X_i$, i.e., $w_X^{\out}(X_\out \cap X_i) = h_i$ and $\cost_z(X_i,C) - \cost_z(X_\out \cap X_i,C) = \cost_z^{(h_i)}(X_i,C)$. Moreover, let weighted set $X_\inl := X - X_\out$ with weight function $w_X^\inl = w_X-w_X^\out$ (recalling that $w_X\equiv 1$ for unweighted $X$) be the inlier set. We find the significant outliers $Z$ among $X_\out$ via the process described in~\Cref{alg:additional centers}. 

\begin{algorithm}[ht]
    \caption{Finding significant outliers $Z$}
    \label{alg:additional centers}
    \begin{algorithmic}[1] 
\State initialize $(U,w_U)\gets (X_\inl,w_X^\inl)$, $(Z,w_Z)\gets (X_\out, w_X^\out)$ 
        \State initialize $a_{q,p}\gets 0$ for every $q\in X_\out,p\in X_\inl$
        \While{$\exists q\in Z, p\in U$ s.t. $\dist(p,q)\le \dist(p,C) + 4\lambda$}
            \State let $a_{q,p}\gets \min\{w_Z(q), w_U(p) \}$
            \label{alg line:allocate}  
            \Comment{$p$ eliminates a fraction $a_{q,p}$ of $q$}
            \State let $w_U(p)\gets w_U(p) - a_{q,p}$ and $w_Z(q)\gets w_Z(q) - a_{q,p}$
            \If{$w_U(p) = 0$}
                \State let $U\gets U\setminus \{p\}$
            \EndIf
            \If{$w_Z(q) = 0$}
                \State let $Z\gets Z\setminus \{q\}$
            \EndIf
        \EndWhile
        \\
        \Return $Z$
    \end{algorithmic}
\end{algorithm}

According to the process of \Cref{alg:additional centers}, we first have the following fact regarding the weighted sets $U$, $Z$, and the values $a_{p,q}$.
\begin{fact}
    \label{fact:UZ}
    The following holds after \Cref{alg:additional centers} terminates:
    \begin{itemize}
        \item [1.] For every $q\in X_\out$, $\sum_{p\in X_\inl} a_{q,p} = w_X^\out(q) - w_Z(q)$, and therefore $\sum_{q\in X_\out}\sum_{p\in X_\inl} a_{q,p} \le w_X^\out(X_\out)\le m$.
        \item [2.] For every $p\in X_\inl$, we have $\sum_{q\in X_\out} a_{q,p} = w_X^\inl(p) - w_U(p)\le w_X^{\inl}(p)$. 
    \end{itemize}
\end{fact}

Then, we define the auxiliary center set $C^\aux\subseteq Z$ to be a $\lambda$-covering of $Z$ of size at most $t$. The existence of $C^\aux$ is presented in the following lemma.

\begin{lemma}
    \label{claim:bounded size of Caux}
    There exists a set $C^\aux\subseteq Z$ with $|C^\aux|\le t$ such that $C^\aux$ is a $\lambda$-covering of $Z$, i.e., for every $x\in Z$, $\dist(x,C^\aux)\le \lambda$.
\end{lemma}
\begin{proof}
    We construct $C^\aux$ as follows: we first initialize $C^\aux = \emptyset$, then for every $P\in \calP$, if $P\cap Z\neq\emptyset$, we pick an arbitrary point $q\in P\cap Z$ and add $q$ into $C^\aux$. Clearly, we have that $|C^\aux|\le |\calP|\le t$ and $C^\aux$ is a $\lambda$-covering of $Z$ because $\calP$ is $\lambda$-bounded.
\end{proof}

Before we proceed to prove~\eqref{eq:bound 1} and~\eqref{eq:bound 2}, we establish some lemmas regarding significant outliers $Z$ and the remaining outliers $X_\out - Z$, which we call insignificant outliers.
Notice that $X_\out-Z$ are the outliers eliminated by \Cref{alg:additional centers}. For each outlier $q\in X_\out$, we have recorded a value $a_{q,p}$, which represents the fraction of $q$ that is eliminated by $p$. We demonstrate that the contribution to $\cost_z(X_\out,C)$ from the eliminated fraction of an outlier $q$ is comparable to the contribution from the inliers eliminating it, leading to a small $\cost_z(X_\out - Z,C)$.

\begin{lemma}[Cost of non-significant outliers]
    \label{lem:for non-significant}
    It holds that
    \begin{equation*}
        \cost_z(X_\out - Z,C)\le O(1)\cdot \cost_z(X_\inl,C) + 2^{O(z\log(z+1))}\cdot m\cdot \lambda^z.
    \end{equation*}
\end{lemma}
\begin{proof}
    For every outlier $q\in X_\out$ and inlier $p\in X_\inl$ such that $a_{q,p}>0$, it holds that $\dist(q,p)\le \dist(p,C) + 4\lambda$ (see Line~\ref{alg line:allocate} of \Cref{alg:additional centers}). By the triangle inequality, we have $\dist(q,C)\le \dist(q,p) + \dist(p,C) \le 2\dist(p,C) + 4\lambda$. 
    Therefore, we have 
    \begin{align*}
        &\quad \cost_z(X_\out - Z,C)\quad\\
        =&\quad \sum_{q\in X_\out} \left(w_X^\out(q) - w_Z(q) \right)\cdot \dist^z(q,C)\\
        \le&\quad \sum_{q\in X_\out}\sum_{p\in X_\inl}a_{q,p}\cdot\left(2\dist(p,C) + 4\lambda \right)^z\\
        \le&\quad \sum_{q\in X_\out}\sum_{p\in X_\inl}a_{q,p}\cdot\left(O(1)\cdot \dist^z(p,C) + 2^{O(z\log(z+1))}\cdot \lambda^z \right)\\
        =&\quad O(1)\cdot \sum_{p\in X_\inl}\left(\sum_{q\in X_\out} a_{q,p} \right)\cdot \dist^z(p,C) + \left(\sum_{q\in X_\out}\sum_{p\in X_\inl}a_{q,p} \right)\cdot 2^{O(z\log(z+1))}\cdot \lambda^z\\
        \le&\quad O(1)\cdot \sum_{p\in X_\inl} w_X^\inl(p)\cdot \dist^z(p,C) + 2^{O(z\log(z+1))}\cdot m\cdot \lambda^z\\
        =&\quad O(1)\cdot \cost_z(X_\inl,C) + 2^{O(z\log(z+1))}\cdot m\cdot \lambda^z.
    \end{align*} 
    In the first inequality, we use $w_X^\out(q) - w_Z(q) = \sum_{p \in X_\inl} a_{q,p}$ from \Cref{fact:UZ} and $\dist(q,C)\le 2\dist(p,C)+4\lambda$ if $a_{q,p}>0$. The second inequality is due to the generalized triangle inequality (\Cref{lem:triangle}), and the third inequality follows from $\sum_{q \in X_\out} a_{q,p} \le w_X^\inl(p)$ and $\sum_{q \in X_\out} \sum_{p \in X_\inl} a_{q,p} \le m$ from \Cref{fact:UZ}.
\end{proof}

For significant outliers, it is challenging to establish an upper bound on their cost, as they might be significantly far away from inliers. However, from another perspective, this implies that if we place a center on a significant outlier (or its $\lambda$-covering), the center will ``absorb'' only a few inliers from the original center set $C$. We formalize this observation in the following lemma, where we actually upper-bound the size of a larger set that contains the inliers absorbed by $C^\aux$.

\begin{lemma}[$C^\aux$ absorbs only a few inliers]
    \label{claim:for significant}
    Let $A:=\{x\in X_\inl: \dist(x,C^\aux)\le \dist(x,C) + 4\lambda \}$, then it holds that $w_X^\inl(A)\le m$.
\end{lemma}
\begin{proof}
    Let $B:=\{x\in X_\inl:\dist(x,Z)\le \dist(x,C) + 4\lambda\}$. We have $A\subseteq B$ since $C^\aux\subseteq Z$, and thus $w_X^\inl(A)\le w_X^\inl(B)$.
    According to the process of~\Cref{alg:additional centers}, after the ``while'' loop terminates, for every $q\in Z$, there is no inlier $p\in U$ such that $\dist(p,q)\le \dist(p,C) + 4\lambda$, which implies that $B\cap U=\emptyset$ and thus $B\subseteq X_\inl\setminus U$.
Hence, we have $w_X^\inl(B)\le w_X^\inl(X_\inl\setminus U) = w_X^\inl(X_\inl) - w_X^\inl(U)$. 
By \Cref{fact:UZ}, we have $w_X^\inl(U)\ge w_U(U)$ and $w_X^\inl(X_\inl) - w_U(U) = \sum_{q\in X_\out}\sum_{p\in X_\inl} a_{q,p}\le m$. Hence, $w_X^\inl(A)\le w_X^\inl(B)\le w_X^\inl(X_\inl) - w_U(U)\le m$.
\end{proof}

\paragraph{Upper Bound for~\eqref{eq:bound 1}} Recall that $S$ is an $\eps$-coreset of $X$ for \tzC{k+t}, and \Cref{claim:bounded size of Caux} implies that $|C\cup C^\aux|\le k+t$. We have
\begin{equation}
    \label{eq:cost(S)-cost(X)}
    \begin{aligned}    
    &\quad\left|\cost_z(S,C\cup C^\aux) - \cost_z(X,C\cup C^\aux)\right| \\
    \le&\quad \eps\cdot \cost_z(X,C\cup C^\aux)\\
    =&\quad \eps\cdot\left(\cost_z(X_\inl,C\cup C^\aux) + \cost_z(X_\out,C\cup C^\aux) \right)\\
    \le&\quad \eps\cdot\left(\cost_z(X_\inl,C) + \cost_z(Z,C\cup C^\aux) + \cost_z(X_\out - Z,C\cup C^\aux) \right)
    \end{aligned}
\end{equation}
where the last inequality is due to $\cost_z(X_\inl,C\cup C^\aux)\le \cost_z(X_\inl,C)$ and $\cost_z(X_\out,C\cup C^\aux) = \cost_z(Z, C\cup C^\aux) + \cost_z(X_\out - Z,C\cup C^\aux)$ since $(Z,w_Z)\subseteq (X_\out,w_X^\out)$.

For $\cost_z(Z,C\cup C^\aux)$, we use the fact that $C^\aux$ is a $\lambda$-covering of $Z$ (see~\Cref{claim:bounded size of Caux}), and thus we have
\begin{equation}
    \label{eq:error for significant}
    \cost_z(Z,C\cup C^\aux)\le \cost_z(Z,C^\aux)\le w_Z(Z)\cdot \lambda^z\le m\cdot \lambda^z.
\end{equation}
For $\cost_z(X_\out - Z,C\cup C^\aux)$, which is the cost of insignificant outliers, we have established an upper bound for it in~\Cref{lem:for non-significant}, specifically,
\begin{equation}
    \label{eq:error for non-significant}
    \cost_z(X_\out - Z,C\cup C^\aux)\le \cost_z(X_\out - Z,C)\le O(1)\cdot \cost_z(X_\inl,C) + 2^{O(z\log(z+1))}\cdot m\cdot \lambda^z.
\end{equation}
By plugging~\eqref{eq:error for significant} and~\eqref{eq:error for non-significant} into \eqref{eq:cost(S)-cost(X)}, we obtain the following result for~\eqref{eq:bound 1}.
\begin{equation}
    \label{eq:key 1}
    \left|\cost_z(X,C\cup C^\aux) - \cost_z(S,C\cup C^\aux) \right|\le O(\eps)\cdot \left(\cost_z(X_\inl,C) + 2^{O(z\log(z+1))}\cdot m\cdot \lambda^z\right).
\end{equation}

\paragraph{Upper Bound for~\eqref{eq:bound 2}} 

Let $(S_\out,w_{S}^\out)\subseteq (S,w_S)$ be an outlier set such that for every $i\in [t]$, $w_{S}^\out(S_\out\cap X_i) = h_i$, and $\cost_z(S_i,C) - \cost_z(S_\out,C) = \cost_z^{(h_i)}(S_i,C)$. Let weighted set $S_\inl:=S - S_\out$ with weight function $w_S^\inl = w_S - w_S^\out$ be the inlier set.
By this definition and the definition of $X_\inl$, we have 
\begin{align*}
    \sum_{i=1}^t\cost_z^{(h_i)}(X_i,C)= \cost_z(X_\inl,C),\quad \sum_{i=1}^t\cost_z^{(h_i)}(S_i,C) = \cost_z(S_\inl,C).
\end{align*}

Therefore, \eqref{eq:bound 2} is equal to
\begin{equation}
    \label{eq:Xin vs Sin}
    \begin{aligned}
        &\quad \left|\left(\cost_z(X_\inl,C) - \cost_z(X,C\cup C^\aux) \right) - \left(\cost_z(S_\inl,C) - \cost_z(S,C\cup C^\aux) \right) \right|\\
    \le&\quad \left|\left(\cost_z(X_\inl,C) - \cost_z(X_\inl,C\cup C^\aux) \right) - \left(\cost_z(S_\inl,C) - \cost_z(S_\inl,C\cup C^\aux) \right)  \right|\\
    &+\left|\cost_z(X_\out,C\cup C^\aux) - \cost_z(S_\out,C\cup C^\aux) \right|
    \end{aligned}
\end{equation}
by using the fact that $\cost_z(X,C\cup C^\aux) = \cost_z(X_\inl,C\cup C^\aux) + \cost_z(X_\out,C\cup C^\aux)$ and $\cost_z(S,C\cup C^\aux) = \cost_z(S_\inl,C\cup C^\aux) + \cost_z(S_\out,C\cup C^\aux)$, along with the triangle inequality.
We separately upper-bound the two terms.

\begin{restatable}[Error Analysis for the First Term]{lemma}{ErrorFirst}
    \label{lem:error first}
    It holds that 
    \begin{align*}
        &\quad \left|\left(\cost_z(X_\inl,C) - \cost_z(X_\inl,C\cup C^\aux) \right) - \left(\cost_z(S_\inl,C) - \cost_z(S_\inl,C\cup C^\aux) \right)  \right|\\
        \le&\quad O(\eps)\cdot \cost_z(X_\inl,C) + 2^{O(z\log (z+1))}\cdot m\cdot \eps^{1-z}\cdot \lambda^z
    \end{align*}
\end{restatable}
\begin{proof}
    For a (weighted) set $Q$, let $\Delta(Q):=\cost_z(Q,C) - \cost_z(Q,C\cup C^\aux)$.
    For every $i\in [t]$, let $x_i:=\arg\min_{x\in X_i}\dist(x,C)$ be the closest point to $C$ among $X_i$. Consider for now a fixed center $c\in C^\aux$, let $I_c:=\{i\in [t]: \dist(x_i,c)\le \dist(x_i,C) + 2\lambda\}$. Observe that for any $i\neq I_c$ and any $x\in X_i$, by repeatedly applying the triangle inequality, we have 
    \begin{equation}
        \dist(x,c)\ge \dist(x_i,c) - \lambda > \dist(x_i,C) + \lambda \ge \dist(x,C).
    \end{equation}
    Let $I:=\bigcup_{c\in C^\aux}I_c$. It holds directly that for any $x\not\in\bigcup_{i\in I} X_i$ and any $c\in C^\aux$, $\dist(x,c) > \dist(x,C)$, which implies that $\dist(x,C^\aux)> \dist(x,C)$. Thus, $\dist^z(x,C) = \dist^z(x,C\cup C^\aux)$, contributing $0$ to both $\Delta(X_\inl)$ and $\Delta(S_\inl)$. 
    
    Hence, we have 
    \begin{equation}
        \label{eq:Delta}    
    \begin{aligned}
        \left|\Delta(X_\inl) - \Delta(S_\inl) \right|\quad
        =&\quad \left|\sum_{i\in I}\Delta(X_\inl\cap X_i) - \sum_{i\in I} \Delta(S_\inl\cap X_i) \right|\\
        \le&\quad \sum_{i\in I}\left|\Delta(X_\inl\cap X_i) - \Delta(S_\inl\cap X_i) \right|.
    \end{aligned}
    \end{equation}
    Then, consider a fixed index $i\in I$ and let $r_i:=w_X^\inl(X_\inl\cap X_i) = w_{S}^\inl(S_\inl\cap X_i)$. We have 
    \begin{align*}
        &\quad \left|\Delta(X_\inl\cap X_i) - r_i\cdot \left(\dist^z(x_i, C) - \dist^z(x_i,C\cup C^\aux) \right) \right|\\
        =&\quad \left|\sum_{x\in X_\inl\cap X_i} w_X^\inl(x)\cdot \big(\left(\dist^z(x,C) - \dist^z(x_i,C) \right) - \left(\dist^z(x,C\cup C^\aux) - \dist^z(x_i,C\cup C^\aux) \right) \big) \right|\\
        \le&\quad \sum_{x\in X_\inl\cap X_i}w_X^\inl(x)\cdot\big(\left|\dist^z(x,C) - \dist^z(x_i,C) \right| + \left|\dist^z(x,C\cup C^\aux) - \dist^z(x_i,C\cup C^\aux) \right| \big).
    \end{align*}
    Since $\dist(x,x_i)\le \diam(X_i)\le \lambda$, by applying the generalized triangle inequality (\Cref{lem:triangle}), we further deduce that 
    \begin{align*}
        &\quad \left|\Delta(X_\inl\cap X_i) - r_i\cdot \left(\dist^z(x_i, C) - \dist^z(x_i,C\cup C^\aux) \right) \right|\\
        \le&\quad \sum_{x\in X_\inl\cap X_i}w_X^\inl(x)\cdot\left(\eps\cdot \dist^z(x_i,C) + \eps\cdot \dist^z(x_i,C\cup C^\aux) + 2^{O(z\log(z+1))}\cdot \eps^{1-z}\cdot \lambda^z \right)\\
        \le&\quad O(\eps)\cdot r_i\cdot \dist^z(x_i,C) + 2^{O(z\log (z+1))}\cdot r_i\cdot \eps^{1-z}\cdot \lambda^z,
    \end{align*}
    where the last inequality is due to $\dist(x_i,C\cup C^\aux)\le \dist(x_i,C)$. Similar result can be obtained for $\Delta(S_\inl\cup X_i)$. Specifically,
    \begin{align*}
        &\quad \left|\Delta(S_\inl\cap X_i) - r_i\cdot \left(\dist^z(x_i, C) - \dist^z(x_i,C\cup C^\aux) \right) \right|\\
        \le&\quad O(\eps)\cdot r_i\cdot \dist^z(x_i,C) + 2^{O(z\log (z+1))}\cdot r_i\cdot \eps^{1-z}\cdot \lambda^z.
    \end{align*}
    Therefore, by triangle inequality, we have 
    \begin{align*}
        &\quad \left|\Delta(X_\inl\cap X_i) - \Delta(S_\inl\cap X_i) \right|\\
        \le&\quad O(\eps)\cdot r_i\cdot \dist^z(x_i,C) + 2^{O(z\log (z+1))}\cdot r_i\cdot \eps^{1-z}\cdot \lambda^z\\
        \le&\quad O(\eps)\cdot \cost_z(X_\inl\cap X_i,C) + 2^{O(z\log (z+1))}\cdot r_i\cdot \eps^{1-z}\cdot \lambda^z,
    \end{align*}
    where the last inequality is due to $\dist(x_i,C)\le \dist(x,C)$ for every $x\in X_\inl\cap X_i$ and $r_i = w_X^\inl(X_\inl\cap X_i)$.
    By plugging this into~\eqref{eq:Delta}, we have
    \begin{equation}    
        \label{eq:final Delta}
        \left|\Delta(X_\inl) - \Delta(S_\inl) \right|
        \le O(\eps)\cdot \cost_z(X_\inl, C) + 2^{O(z\log (z+1))}\cdot \sum_{i\in I} r_i\cdot \eps^{1-z}\cdot \lambda^z
    \end{equation}
    
    It remains to bound $\sum_{i\in I} r_i$ which is equal to $w_X^\inl\left(\bigcup_{i\in I} (X_\inl\cap X_i) \right)$ due to the disjointedness of $X_1,\dots, X_t$.
To this end, consider for now a fixed $c\in C^\aux$, we recall the definition $I_c:=\{i\in [t]: \dist(x_i,c)\le \dist(x_i,C) + 2\lambda\}$, where $x_i\in X_i$, and recall that we defined in~\Cref{claim:for significant} a set $A=\{x\in X_\inl:\dist(x,C^\aux)\le \dist(x,C) + 4\lambda\}$, for which we prove that $w_X^\inl(A)\le m$.
For every $i\in I_c$ and every $x\in X_i\cap X_\inl$, we have 
    \begin{equation}
        \dist(x,c)\le \dist(x_i,c) + \lambda \le \dist(x_i,C) + 3\lambda\le \dist(x,C) + 4\lambda.
    \end{equation}
    This implies that $x\in A$. Therefore, $\bigcup_{i\in I_c} (X_\inl\cap X_i)\subseteq A$ for every $c\in C^\aux$, and thus $\bigcup_{i\in I} (X_\inl\cap X_i)  = \bigcup_{c\in C^\aux}\bigcup_{i\in I_c}(X_\inl\cap X_i)\subseteq A$, which implies that $w_X^\inl\left(\bigcup_{i\in I}(X_\inl\cap X_i)\right)\le w_X^\inl(A)\le  m$. 
    By plugging this into~\eqref{eq:final Delta}, we complete the proof of \Cref{lem:error first}.
\end{proof}

\begin{restatable}[Error Analysis for the Second Term]{lemma}{ErrorSecond}
    \label{lem:error second}
    It holds that 
    \begin{equation*}
        \left|\cost_z(X_\out,C\cup C^\aux) - \cost_z(S_\out,C\cup C^\aux) \right|\le O(\eps)\cdot \cost_z(X_\inl,C) + 2^{O(z\log(z+1))}\cdot m\cdot \eps^{1-z}\cdot \lambda^z.
    \end{equation*}
\end{restatable}
\begin{proof}
    We first have 
    \begin{align*}
        &\quad \left|\cost_z(X_\out,C\cup C^\aux) - \cost_z(S_\out,C\cup C^\aux) \right|\\
        =&\quad \left|\sum_{i=1}^t \cost_z(X_\out\cap X_i, C\cup C^\aux) - \sum_{i=1}^t\cost_z(S_\out\cap X_i, C\cup C^\aux) \right|\\
        \le&\quad \sum_{i=1}^t\left|\cost_z(X_\out\cap X_i, C\cup C^\aux) - \cost_z(S_\out\cap X_i, C\cup C^\aux)\right|.
    \end{align*}
    Consider a fixed $i\in[t]$, we let $x_i:=\arg\min_{p\in X_i}\dist(p,C\cup C^\aux)$ be the closest point to $C\cup C^\aux$ among points in $X_i$. Then, recall that $w_X^\out(X_\out\cap X_i) = h_i = w_{S}^\out(S_\out\cap X_i)$, we have 
    \begin{align*}
        &\quad\left|\cost_z(X_\out\cap X_i,C\cup C^\aux) - h_i\cdot \dist^z(x_i,C\cup C^\aux) \right|\\
        \le&\quad \sum_{x\in X_\out\cap X_i} w_X^\out(x)\cdot \left| \dist^z(x,C\cup C^\aux) - \dist^z(x_i,C\cup C^\aux) \right|\\
        \le&\quad \sum_{x\in X_\out\cap X_i}w_X^\out(x)\cdot\left(\eps\cdot \dist^z(x_i,C\cup C^\aux) + 2^{O(\log(z+1))}\cdot \eps^{1-z}\cdot \dist^z(x,x_i) \right)\\
        \le&\quad \eps\cdot h_i\cdot \dist^z(x_i,C\cup C^\aux) + 2^{O(z\log(z+1))}\cdot h_i\cdot \eps^{1-z}\cdot \lambda^z.
    \end{align*}  
    We can derive a similar bound for $\left|\cost_z(S_\out\cap X_i,C\cup C^\aux) - h_i\cdot \dist^z(x_i,C\cup C^\aux) \right|$, which is 
    \begin{align*}
        &\quad \left|\cost_z(S_\out\cap X_i,C\cup C^\aux) - h_i\cdot \dist^z(x_i,C\cup C^\aux) \right|\\
        \le&\quad \eps\cdot h_i\cdot \dist^z(x_i,C\cup C^\aux) + 2^{O(z\log(z+1))}\cdot h_i\cdot \eps^{1-z}\cdot \lambda^z.
    \end{align*}
    Hence, by applying the triangle inequality, we have 
    \begin{align*}
        &\quad \left|\cost_z(X_\out\cap X_i,C\cup C^\aux) - \cost_z(S_\out\cap X_i,C\cup C^\aux) \right|\\
        \le&\quad 2\eps\cdot h_i\cdot \dist^z(x_i,C\cup C^\aux) + 2^{O(z\log(z+1))}\cdot h_i\cdot \eps^{1-z}\cdot \lambda^z\\
        \le&\quad 2\eps\cdot \cost_z(X_\out\cap X_i,C\cup C^\aux) + 2^{O(z\log(z+1))}\cdot h_i\cdot \eps^{1-z}\cdot \lambda^z
    \end{align*}
    Summing up over all $i\in [t]$, we have 
    \begin{equation}
        \label{eq:term 2}
        \begin{aligned}    
        &\quad \left|\cost_z(X_\out,C\cup C^\aux) - \cost_z(S_\out,C\cup C^\aux) \right|\\
        \le&\quad \sum_{i=1}^t\left|\cost_z(X_\out\cap X_i, C\cup C^\aux) - \cost_z(S_\out\cap X_i, C\cup C^\aux)\right|\\
        \le&\quad 2\eps\cdot \cost_z(X_\out, C\cup C^\aux) + 2^{O(z\log(z+1))}\cdot m\cdot \eps^{1-z}\cdot \lambda^z\\
        \le&\quad O(\eps)\cdot \cost_z(X_\inl,C) + 2^{O(z\log(z+1))}\cdot m\cdot \eps^{1-z}\cdot \lambda^z,
        \end{aligned} 
    \end{equation}
    where the last inequality is due to~\eqref{eq:error for significant} and~\eqref{eq:error for non-significant}.
    This finishes the proof of \Cref{lem:error second}.
\end{proof}

By plugging the bounds from~\Cref{lem:error first} and~\Cref{lem:error second} into~\eqref{eq:Xin vs Sin}, we obtain 
\begin{equation}
    \label{eq:key 2}
    \begin{aligned}    
    &\quad \left|\left(\cost_z(X_\inl,C) - \cost_z(X,C\cup C^\aux) \right) - \left(\cost_z(S_\inl,C) - \cost_z(S,C\cup C^\aux) \right) \right|\\
    \le&\quad O(\eps)\cdot \cost_z(X_\inl,C) + 2^{O(z\log (z+1))}\cdot m\cdot \eps^{1-z}\cdot \lambda^z.
    \end{aligned}
\end{equation}

\paragraph{Concluding~\Cref{cor:second}}
By combining~\eqref{eq:key 1},~\eqref{eq:key 2} and the triangle inequality, we have
\begin{align*}
    \quad &\left| \sum_{i=1}^t\cost_z^{(h_i)}(X_i,C) - \sum_{i=1}^t\cost_z^{(h_i)}(S_i,C) \right|\\
    =\quad &\left|\cost_z(X_\inl,C) - \cost_z(S_\inl,C) \right|\\
    \le\quad &\left|\left(\cost_z(X_\inl,C) - \cost_z(X,C\cup C^\aux) \right) - \left(\cost_z(S_\inl,C) - \cost_z(S,C\cup C^\aux) \right) \right|\\
    +&\left|\cost_z(X,C\cup C^\aux) - \cost_z(S,C\cup C^\aux) \right| \\
    \le\quad& O(\eps)\cdot \cost_z(X_\inl,C) + 2^{O(z\log(z+1))}\cdot m\cdot \eps^{1-z}\cdot \lambda^z\\
    =\quad & O(\eps)\cdot \sum_{i=1}^t\cost_z^{(h_i)}(X_i,C) + 2^{O(z\log(z+1))}\cdot m\cdot \eps^{1-z}\cdot \lambda^z.
\end{align*}
Since above inequality holds for all $C\in V^k$ and all real numbers $h_1,\dots,h_t\ge 0$ with $\sum_{i=1}^th_i\le m$, we have proved that $S$ satisfies the condition of \Cref{claim:a sufficient condition}, with $\eta = 2^{O(z\log(z+1))}\cdot m\eps^{1-z}\lambda^z$. 
Therefore, we conclude that $S$ is an $(O(\eps),2^{O(z\log(z+1))}\cdot m\eps^{1-z}\lambda^z)$-coreset of $X$ for \kzmC.
\qed

\subsection{Correctness of~\Cref{alg:size-preserving}: Proof of~\Cref{lem:size-preserving}}
\label{sec:reduction for size_preserving}

\ReductionSizePreserving*

Write $\mathcal{P} = (P_1, \ldots, P_t)$ where $\calP$ is a $\lambda$-bounded partition for $\lambda := \frac{\eps\mu}{1000z\Gamma}$, and let $k'\ge (k+t)\Lambda$ be an integer, as in the premise.
Let $M^\dup_{l,w} = (V\times [l], \dist_{l,w})$ be the $w$-separated $l$-duplication of $M$ (see \Cref{def:separated duplication}) used in Line~\ref{alg line:mapping} of \Cref{alg:size-preserving}, where $w=200z\eps^{-1}\cdot \diam(X)\cdot n^{1/z}$. 

We first have the following claim that bounds the size of $\calQ$ according to the sparsity property of $\calQ$.

\begin{claim}
    \label{claim:bounded Q}
    It holds that $|\calQ|=l\le t\Lambda$.
\end{claim}
\begin{proof}
    Recall that $\calQ$ is a $(\mu,\Gamma,\Lambda)$-sparse partition (see~\Cref{def:sparse partition}), by the sparsity property, for every part $P\in\calP$, since $\diam(P)\le \lambda \le \mu/\Gamma$, we have $P$ intersects at most $\Lambda$ parts in $\calQ$. Moreover, since $\bigcup_{Q\in \calQ} Q = X = \bigcup_{P\in\calP} P$, every part $Q$ intersect at least one part in $\calP$. Therefore, we have $|\calQ| \le \sum_{P\in\calP} |\{Q\in \calQ: Q\cap P\neq \emptyset\}|\le |\calP|\cdot \Lambda=t\Lambda$.
\end{proof}

Then, we separately prove the two properties of $S$ in the following.

\paragraph{Size-preserving Property of $S$} 
Notice that for every $i \in [l]$, $|X_i| = |X_i'|$ and $w_S(S \cap X_i) = w_{S'}(S' \cap X_i')$. Here, $S'$ is an $\eps$-coreset of $X'$ for \tzC{k'} on the metric space $M^\dup_{l,w}$, constructed in Line~\ref{alg line:vanilla coreset} of \Cref{alg:size-preserving}. Therefore, we aim to prove the size-preserving property of $S'$, i.e., $\forall i\in [l]$, $w_{S'}(S'\cap X'_i)\in (1\pm \eps)\cdot |X'_i|$, which is equivalent to $w_{S}(S\cap X_i)\in (1\pm \eps)\cdot |X_i|$.

For every $i\in [l]$, we pick an arbitrary point $x_i\in X_i$ and let $x_i':=(x_i,i)\in X_i'$. 
Recall that $w\ge \diam(X)$, the separation property of $M^\multi_{l,w} = (V\times [l], \dist_{l,w})$ (see~\Cref{def:separated duplication}) implies the following fact.

\begin{fact}
\label{claim:closest center}
For every $i\in [l]$ and every $p\in X_i$, let $p':= (p,i) \in X_i'$. Then it holds that $\min_{j\in [l]} \dist_{l,w}(p',x'_j) = \dist_{l,w}(p',x'_i) = \dist(p,x_i)$.
\end{fact}
Then, for a fixed index $i\in [l]$, we consider the center set $C:=\{x'_j: j\in [l],j\neq i \}\subseteq V\times [l]$, which contains all $x'_j$ for $j\in [l]$ except $x'_i$. 
By~\Cref{claim:closest center}, we have
\begin{align*}
\cost_z(X',C)\quad&=\quad \sum_{j\in [l]\setminus\{i\}  }\sum_{p'\in X'_j}\dist^z_{l,w}(p',x'_j) + \sum_{p'\in X_i'} \dist^z_{l,w}(p', C)\\
&=\quad \sum_{j\in [l]\setminus \{i\}}\sum_{p\in X_j}\dist^z(p,x_j) + \sum_{p'\in X_i'} \dist^z_{l,w}(p', C)
\end{align*}
For the first term, we have $\sum_{j\in [l]\setminus \{i\}}\sum_{p\in X_j}\dist^z(p,x_j)\le n\cdot (\diam(X))^z\le \eps\cdot w^z/200$.
For the term $\sum_{p'\in X_i'} \dist^z_{l,w}(p', C)$, we have $\sum_{p'\in X_i'} \dist^z_{l,w}(p', C)\ge w^z$.
Hence, 
\begin{align}
\label{eq:cost X_i'}
\cost_z(X',C)\in (1\pm \eps/100)\sum_{p'\in X_i'} \dist^z_{l,w}(p', C).
\end{align}
Moreover, for every $p'=(p,i)\in X'_i$, by the triangle inequality, we have 
\begin{equation*}
|\dist_{l,w}(p', C) - \dist_{l,w}(x'_i,C) |\le \dist_{l,w}(p', x_i')\le \diam(X)\le \frac{\eps w}{200z} \le \frac{\eps}{200z}\cdot \dist_{l,w}(x'_i,C).
\end{equation*}
As a result, we obtain $\dist_{l,w}^z(p',C)\in (1\pm \eps/50)\cdot \dist_{l,w}^z(x_i',C)$ using the fact that $1-2\alpha\le (1-\alpha/z)^z, (1+\alpha/z)^z\le 1+2\alpha$ for $\alpha\in (0,1)$. Therefore, we have
\begin{equation*}
\sum_{p'\in X_i'} \dist^z_{l,w}(p', C)\in (1\pm \eps/50)\cdot |X_i'|\cdot \dist^z_{l,w}(x_i',C).
\end{equation*}
Combined with~\eqref{eq:cost X_i'}, we have 
\begin{equation}
\label{eq:cost X'}
\cost_z(X',C)\in (1\pm \eps/10)\cdot |X_i'|\cdot \dist^z_{l,w}(x_i',C).
\end{equation}
For $\cost_z(S',C)$, let $C^\mathrm{full}:=\{x'_j:j\in [l]\}$ and we have 
\begin{align*}
    \cost_z(S',C) =& \sum_{j\in [l]\setminus\{i\}} \sum_{p'\in S'\cap X_j'} w_{S'}(p')\cdot \dist^z_{l,w}(p',x_j') + \sum_{p'\in S'\cap X_i'} w_{S'}(p')\cdot \dist_{l,w}^z(p',C).
\end{align*}
The first term can be upper bounded by 
\begin{align*}
    \sum_{j\in [l]\setminus\{i\}} \sum_{p'\in S'\cap X_j'} w_{S'}(p')\cdot \dist^z_{l,w}(p',x_j')\quad&\le\quad \sum_{j\in [l]} \sum_{p'\in S'\cap X_j'} w_{S'}(p')\cdot \dist^z_{l,w}(p',x_j')\\
    &=\quad \cost_z(S',C^\mathrm{full})\\
    &\le\quad 2\cost_z(X',C^\mathrm{full})\\
    &\le\quad \eps\cdot w^z / 100.
\end{align*}
where the second inequality is due to the coreset guarantee of $S'$, and the last inequality is due to $\cost_z(X',C^\mathrm{full}) \le \eps \cdot w^z / 200$ as discussed above.
Then, by a similar argument, we can obtain the same result for $\cost_z(S',C)$:
\begin{equation}
\label{eq:cost S'}
\cost_z(S',C)\in (1\pm \eps/10)\cdot w_{S'}(S'\cap X_i')\cdot \dist^z_{l,w}(x_i',C).
\end{equation}
Finally, since $|C| = l-1$ and $S'$ is an $\eps/10$-coreset for \tzC{k'} with $k' = (t+k)\Lambda\ge l$ by~\Cref{claim:bounded Q}, we have $\cost_z(S',C)\in (1\pm\eps/10)\cdot \cost_z(X',C)$. Hence, we obtain that $w_{S'}(S'\cap X_i')\in (1\pm \eps)\cdot |X_i'|$, which is equivalent to $w_{S}(S\cap X_i)\in (1\pm \eps)\cdot |X_i|$.

\paragraph{Coreset Guarantee of $S$} Consider any center set $C\in V^k$. Our strategy is to construct a ``bridge'' center set $C'\subseteq V\times [l]$ such that 
\begin{equation}
\label{eq:bridge}
\cost_z(X,C)\in (1\pm \eps/10)\cdot \cost_z(X',C'),\quad \cost_z(S,C)\in (1\pm \eps/10)\cdot \cost_z(S',C').
\end{equation}
If there indeed exists such a $C'$ with $|C'|\le k'= (t+k)\Lambda$, then as the coreset guarantee of $S'$ states that $\cost_z(S',C')\in (1\pm \eps/10)\cdot \cost_z(X',C')$, we conclude that $\cost_z(X,C)\in (1\pm \eps)\cdot \cost_z(S,C)$.

Hence, it remains to give the construction of the bridge center set, which utilizes the $\lambda$-bounded partition $\calP=(P_1,\dots, P_t)$ (recalling that $\lambda = \frac{\eps\mu}{1000z\Gamma})$. For every $j\in [t]$, we define $C_j^\close:=\{c\in C: \dist(c,P_j)\le \frac{\mu}{4\Gamma} \}$, $c_j^\far:= \arg\min_{c\in C\setminus C_j^\close} \dist(c, P_j)$, and let $C_j:=C_j^\close\cup \{c_j^\far\}$. We have the following lemma.

\begin{lemma}
\label{claim:C_j}
For every $x\in P_j$, it holds that 
\begin{equation*}
    \dist^z(x,C)\le \dist^z(x, C_j)\le (1+\eps/50)\cdot \dist^z(x,C).
\end{equation*}
\end{lemma}
\begin{proof}
The first inequality follows from $C_j\subseteq C$. Hence, it remains to prove the second inequality. Firstly, for every $c\in C\setminus C_j^\close$, let $y:=\arg\min_{x\in P_j} \dist(x,c)$ be the closest point to $c$ in $P_j$. Then, for every $x\in P_j$, we have 
\begin{align*}
    \left|\dist(x,c) - \dist(c,y)\right|\le \dist(x,y)\le \diam(P_j)\le\lambda = \frac{\eps\mu}{1000z\Gamma} \le \frac{\eps}{250z} \cdot \dist(c,y),
\end{align*}
where we are using the fact that $\dist(c,y) = \dist(c,P_j)> \frac{\mu}{4\Gamma}$. Thus, we have $\dist(x,c) \in (1\pm \frac{\eps}{250z})\cdot \dist(c,P_j)$. Combined with the definition of $c_j^\far$, we have $\dist(x,c_j^\far)\le (1+\frac{\eps}{100z})\cdot \dist(x,C\setminus C_j^\close)$, which implies that
\begin{align*}
    \dist(x,C_j) \quad = &\quad  \min\left\{\dist(x,C_j^\close), \dist(x,c_j^\far)\right\} \\
    \le&\quad \left(1+\frac{\eps}{100z}\right)\cdot\min\left\{\dist(x,C_j^\close), \dist(x,C\setminus C_j^\close)\right\}\\
    =&\quad \left(1+\frac{\eps}{100z}\right)\cdot \dist(x,C).
\end{align*}
By using the fact $(1+\alpha/z)^z \le 1 + 2\alpha$ for $\alpha\in (0,1)$, we complete the proof of \Cref{claim:C_j}.
\end{proof}

Then, for every $i\in [l]$, let $D_i:=\bigcup_{j\in [t]: P_j\cap X_i\neq \emptyset} C_j$.
Now consider any $x\in X_i$ and let $j\in [t]$ be an integer such that $x\in P_j$. We have $P_j\cap X_i\neq \emptyset$, which means that $C_j\subseteq D_i$. Therefore, by ~\Cref{claim:C_j}, we have
\begin{equation}
\label{eq:Di}
\dist^z(x,C)\le \dist^z(x,D_i)\le \dist^z(x,C_j)\le (1+\eps/50)\cdot \dist^z(x,C).
\end{equation}
This implies that 
\begin{equation}
\label{eq:cost P_i}
\cost_z(X,C) = \sum_{i=1}^l \cost_z(X_i,C) \in (1\pm\eps/50)\cdot \sum_{i=1}^l\cost_z(X_i,D_i),
\end{equation}
and
\begin{equation}
\label{eq:cost S_X}
\cost_z(S,C) = \sum_{i=1}^l\cost_z(S\cap X_i,C)\in (1\pm \eps/50)\cdot \sum_{i=1}^l\cost_z(S\cap X_i,D_i).
\end{equation}

Let $D'_i:=\{(c,i)\in V\times [l]:c\in D_i\}$ and $C':=\bigcup_{i=1}^l D_i'$ for every $i\in[l]$. For every $i,j\in [l]$ with $i\neq j$, the separation property of $M^\dup_{l,w}$ (see \Cref{def:separated duplication}) implies that 
\begin{equation*}
\cost_z(X_i',D_j')\ge \cost_z(X_i,D_j)\ge \cost_z(X_i,C).
\end{equation*}

By~\eqref{eq:Di}, we further deduce that 
\begin{equation*}
\cost_z(X_i',D_j')\ge \cost_z(X_i,C)\ge (1+\eps/50)^{-1}\cdot \cost_z(X_i,D_i)= (1+\eps/50)^{-1}\cdot \cost_z(X_i',D_i'),
\end{equation*}
implying that $\cost_z(X_i',D_i')\in (1\pm\eps/50)\cdot \cost_z(X_i',C')$.
Therefore, we have 
\begin{equation*}
\sum_{i=1}^l\cost_z(X_i,D_i) = \sum_{i=1}^l\cost_z(X_i',D_i') \in (1\pm \eps/50)\cdot\cost_z(X',C').
\end{equation*}
Similarly, we can obtain 
\begin{equation*}
\sum_{i=1}^l\cost_z(S\cap X_i,D_i)\in (1\pm\eps/50)\cdot \cost_z(S',C').
\end{equation*}
Combined with~\eqref{eq:cost P_i},~\eqref{eq:cost S_X}, we confirm that $C'$ satisfies the bridge property stated in~\eqref{eq:bridge}. It remains to prove the following \Cref{claim:bounded C'} regarding the size of $C'$, which would finish the proof of \Cref{lem:size-preserving}.

\begin{lemma}
\label{claim:bounded C'}
It holds that
$|C'|=\sum_{i=1}^l |D_i'|\le (t+k)\Lambda$.
\end{lemma}
\begin{proof}

For every $i\in [l]$, recall that $|D'_i|=|D_i|=|\bigcup_{j\in [t]:P_j\cap X_i\neq\emptyset} C_j|$. For simplicity, let $I_i:=\{j\in [t]: P_j\cap X_i\neq\emptyset\}$ denote the indices of sets that intersect $X_i$. Since $C_j = C_j^\close\cup\{c_j^\far\}$ for every $j\in[t]$, we have $|D_i| \le |\bigcup_{j\in I_i} C_j^\close| + |I_i|$. In the following, we separately provide upper bounds for $\sum_{i=1}^l |\bigcup_{j\in I_i} C_j^\close |$ and $\sum_{i=1}^l |I_i|$.

For $\sum_{i=1}^l |\bigcup_{j\in I_i} C_j^\close |$, we can rewrite it as 
\begin{equation}
\label{eq:rewrite1}
\sum_{i=1}^l |\bigcup_{j\in I_i} C_j^\close| = \sum_{i=1}^l\sum_{c\in C} \tau_{c,i} = \sum_{c\in C} \sum_{i=1}^l \tau_{c,i},
\end{equation}
where $\tau_{c,i}=\begin{cases}
1&\exists j\in I_i, c\in C_j^\close\\
0&\text{otherwise}
\end{cases}$.
Fix a center $c\in C$, for every $i\in [l]$ such that $\tau_{c,i} = 1$, let $j\in I_i$ be an arbitrary index satisfying that $c\in C_j^\close$ and $p\in P_j\cap X_i$ be an arbitrary point. We have 
\begin{equation}
\dist(c,p)\le \dist(c,P_j) + \diam(P_j)\le \frac{\mu}{4\Gamma} + \lambda\le \frac{\mu}{2\Gamma}.
\end{equation}
Therefore, $\tau_{c,i} = 1$ implies that $\ball_X(c,\frac{\mu}{2\Gamma})$ intersects $X_i$. Pick an arbitrary point $x\in \ball_X(c,\frac{\mu}{2\Gamma})$, we have $\ball_X(c,\frac{\mu}{2\Gamma})\subseteq \ball_X(x, \frac{\mu}{\Gamma})$.
By the sparsity of partition $\calQ=\{X_1,\dots,X_l\}$, 
$\ball_X(x,\frac{\mu}{\Gamma})$ can intersect at most $\Lambda$ parts within $\calQ$. Hence, $\ball_X(c,\frac{\mu}{2\Gamma})$ intersects at most $\Lambda$ parts, which implies that $\sum_{i=1}^l\tau_{c,i}\le \Lambda$. According to~\eqref{eq:rewrite1}, we have $\sum_{i=1}^l |\bigcup_{j\in I_i} C_j^\close |\le k\Lambda$.

For $\sum_{i=1}^l|I_i|$, we can rewrite it as 
\begin{equation}
\sum_{i=1}^l|I_i| = \sum_{j=1}^t |\{i\in [l]: P_j\cap X_i\neq\emptyset\}|.
\end{equation}
Similarly, pick an arbitrary point $x\in P_j$, we have
$P_j\subset \ball_X(x,\frac{\eps\mu}{100z\Gamma})$ and $\ball_X(x,\frac{\eps\mu}{100z\Gamma})$ can intersect at most $\Lambda$ parts, which implies that $\sum_{i=1}^l|I_i|\le t\Lambda$.

In conclusion, we have $|D'|\le \sum_{i=1}^l|\bigcup_{j\in I_i}C_j^\close | + \sum_{i=1}^l |I_i| \le (t+k)\Lambda$.
\end{proof}

\section{Streaming Implementations}
\label{sec:streaming}

In this section, we implement our reduction algorithms (both~\Cref{thm:reduction I} and~\Cref{thm:reduction2}) in the dynamic streaming setting. 
We consider the standard geometric streaming model proposed by Indyk~\cite{DBLP:conf/stoc/Indyk04},
where the input is from a discrete set $[\Delta]^d$ for some integer $\Delta \geq 1$ and is presented as a stream of point insertion and deletions.
We obtain two black-box reductions, stated together in \Cref{thm:dynamic coresets}, that turns a streaming coreset construction algorithm for \kzC to the one for \kzmC, with guarantee similar to \Cref{thm:reduction I} and \Cref{thm:reduction2}, respectively.

\begin{theorem}[Streaming Coresets for \kzmC]
    \label{thm:dynamic coresets}

    Assume there exists a streaming algorithm that, given $0< \eps,\delta < 1$, integers $k, z, d, \Delta \geq 1$, and a dataset from $[\Delta]^d$ presented as a dynamic stream, uses space $W(d,\Delta, k, \eps^{-1},\delta^{-1})$ to construct an $\eps$-coreset for \kzC with a failure probability of at most $\delta$.
Then there exists a streaming algorithm that, given $0< \eps,\delta < 1$, integers $k,z,m,d,\Delta\ge 1$, and a dataset $X\subseteq [\Delta]^d$ presented as a dynamic stream, constructs an $\eps$-coreset of $X$ for \kzmC. This algorithm has a failure probability of at most $\delta$ and uses space $\min\{W_1,W_2\}\cdot \poly(d\log(\delta^{-1}\Delta))$, where 
    \begin{equation*}   
        \begin{aligned}
        W_1 &= 2^{O(z\log z)}\cdot \tilde O\left((k+d^z)m\eps^{-1}\right)+ W\left(d,\Delta,k,O(\eps^{-1}),O(\delta^{-1}zd\log\Delta)\right),\\
        W_2 &= 2^{O(z\log z)}\cdot \tilde O\left(k + m(d/\eps)^{2z}\right) + W\left(d+1, z\eps^{-1}\Delta^{\poly(d)}, k\poly(d), O(\eps^{-1}),O(\delta^{-1}zd\log\Delta)\right).
        \end{aligned}
    \end{equation*}
\end{theorem}

      \subsection{$W_1\cdot\poly(d\log(\delta^{-1}\Delta))$-Space: Streaming Implementation of~\Cref{thm:reduction I}}
\label{sec:streaming coreset}
\paragraph{Algorithm Overview}
The first space complexity $W_1\cdot \poly(d\log\Delta)$ is derived by implementing~\Cref{thm:reduction I} in the dynamic streaming setting.
Recall that the coreset construction of~\Cref{thm:reduction I} is based on a $\lambda$-bounded partition, where the value of $\lambda$ depends on the optimal objective $\OPT_z^{(m)}(X)$ (which we can ``guess'' using small additional space via standard approaches).
In the following, we assume that $\lambda$ is fixed and proceed to present the subroutines we will use in the algorithm's design.
The main challenge lies in the streaming construction of almost-dense $\lambda$-bounded partition as shown in~\Cref{lem:decomposition}, which inherently requires storing the entire dataset using $\Omega(n)$ space. 
To address this, our strategy is to set up a partition of the space $[\Delta]^d$ into buckets before the stream begins,
which can be achieved efficiently (in space) by a geometric decomposition technique called \emph{consistent hashing} (see~\Cref{lem:consistent hashing}). 
The non-empty buckets, which contain points from the dataset, form a partition of $X$.
Although there may be numerous non-empty buckets, we show (in~\Cref{lem:bounded size}) that, similar to~\Cref{claim:refined decopmosition} (which is a refinement of~\Cref{lem:decomposition}), there exists a set $F$ of isolated points with a bounded size that can be efficiently extracted in a dynamic stream. Moreover, after removing this set from the dataset, the number of non-empty buckets significantly decreases. Hence, $X\setminus F$ admits a small-sized partition (defined by the consistent hashing).
Next, we employ a two-level sparse recovery approach to identify all buckets containing a small number of points (see~\Cref{lem:find light parts}), thereby obtaining the sparse subset $X_{\calS}$. The dense subset is then $X\setminus (X_\calS \cup F)$, and we construct an $\epsilon$-coreset $S_{\calD}$ for \kzC by independently running a streaming coreset construction on a stream that represents $X\setminus (X_{\calS} \cup F)$, which comprises the dataset stream of $X$ with the removals corresponding to $X_{\calS} \cup F$. 
Finally, the coreset for $\kzmC$ of $X$ is formed as $S_{\calD} \cup X_{\calS} \cup F$.

\paragraph{Defining $\lambda$-Bounded Partition via Consistent Hashing} 
We need to use a geometric hashing technique called \emph{consistent hashing}.
The definition of consistent hashing is actually equivalent to that of a sparse partition of the space $\mathbb{R}^d$,
except that it is data oblivious, and requires little space to store and evaluate for every point $x$ the ID of the part that $x$ belongs to. 
Consistent hashing was first applied in the dynamic streaming setting in a recent work~\cite{arxiv.2204.02095}, and we summarize their result in the following lemma.

\begin{lemma}[{\cite[Theorem 5.1]{arxiv.2204.02095}}]
    \label{lem:consistent hashing}
    For any $\lambda > 0$, there exists a (deterministic) hash $\phi :\R^d\to\R^d$ such that $\{\phi^{-1}(y):y\in \phi(\R^d)\}$ forms a $(\mu,\Gamma,\Lambda)$-sparse partition (see \Cref{def:sparse partition}) of $\R^d$ in Euclidean space, where $\Gamma= O(d)$ and $\Lambda = O(d\log d)$.
Furthermore, $\phi$ can be described using $O(d^2\log^2d)$ bits and one can evaluate $\phi(X)$ for any point $x\in\R^d$ in space $O(d^2\log^2d)$.
\end{lemma}
A consistent hash $\phi:\R^d\to\R^d$, obtained from \Cref{lem:consistent hashing} in a data-oblivious way, maps each point $x\in X$ to the \emph{bucket} $\phi(x)$, and these buckets can be used to create a partition. Formally, we can define a partition $\calP_{\phi}(X):=\{X\cap \phi^{-1}(y):y\in \phi(X) \}$ of $X$ based on $\phi$. The partition satisfies that $|\calP_{\phi}(X)|=|\phi(X) |$ and,
according to the diameter property in~\Cref{def:sparse partition}, $\calP_{\phi}(X)$ is $\lambda$-bounded.
In the following lemma, we demonstrate the existence of a small isolated set $F$ similar to~\Cref{claim:refined decopmosition}.

\begin{lemma}
    \label{lem:bounded size}
    For a consistent hash $\phi$ with a diameter bound $\lambda$ and a dataset $X\subseteq \R^d$, there exists a subset $F\subseteq X$ of $X$ with $|F|\le m + (\Gamma/\lambda)^z\cdot\OPT_z^{(m)}(X)$ such that $|\phi(X\setminus F)|\le k\Lambda$.
\end{lemma}
\begin{proof}
    Let $\calP=\calD\cup \calS$ be the decomposition of $X$ computed using~\Cref{lem:decomposition} with $C^*$ being the optimal solution to \kzmC and $\lambda':=\lambda / \Gamma$.
According to~\Cref{claim:refined decopmosition}, $\calS$ can be further decomposed into $\calS_1$ and $\calS_2$ such that $|\calD| + |\calS_1|\le k$.
    Moreover, let $F:=\bigcup_{P\in\calS_2} P$, it holds that $|F|\le m + \OPT_z^{(m)}(X)\cdot \lambda'^{-z} = m + (\Gamma/\lambda)^z\cdot\OPT_z^{(m)}(X)$. 
    
    Since for every $P\in (\calD\cup \calS_1)$, $\diam(P)\le \lambda'=\lambda/\Gamma$, which implies that $|\phi(P)|\le \Lambda$ due to~\Cref{lem:consistent hashing}. Therefore, we have 
    $|\phi(X\setminus F)|\le \sum_{P\in (\calD\cup\calS_1)}|\phi(P)|\le k\Lambda$.
\end{proof}

\paragraph{Extract Isolated Points $F$} 
After fixing a mapping $\phi$, the next step is to extract most of the isolated points, for which we have~\Cref{lem:isolated}.
Here we aim to achieve a slightly different goal: extracting a set $G$ from $X$ such that $|\phi(X\setminus G)|$ approximates $|\phi(X\setminus Y)|$ for any fixed set $Y$ with $|Y| \le T$ (where $T$ is a parameter to be determined). This result suffices for our purposes because, as guaranteed by~\Cref{lem:bounded size}, there exists a small-sized set $F$ with $|\phi(X\setminus F)| \le k\Lambda = O(kd\log d)$. Therefore, by carefully choosing $T$, we can obtain a set $G$ with nearly the same result: $|\phi(X\setminus G)| \le 2k\Lambda$.

\begin{restatable}[Extract isolated points]{lemma}{ExtractIsolatedPoint}
    \label{lem:isolated}
    There exists a randomized algorithm that, given $0<\delta<1$, an integer $T >0$, a hash $\phi:\R^d\to\R^d$ such that the value for any point in $\R^d$ can be evaluated in space $\poly(d)$ and a dataset $X\subseteq [\Delta]^d$ presented as a dynamic stream, uses $\tilde O\left(T\poly(d\log(\delta^{-1}\Delta))\right)$ space to sample a random subset $G\subseteq X$ of size $\tilde O\left(T\poly(d\log(\delta^{-1}\Delta))\right)$ such that, for any subset $Y\subseteq X$ with $|Y|\le T$, $\Pr\left[|\phi(X\setminus G)|\le 2\cdot |\phi(X\setminus Y)|\right]\ge 1-\delta$.
\end{restatable}
The proof of~\Cref{lem:isolated}, provided in Appendix~\ref{apd:isolated}, relies on a two-level sampling procedure as follows:
we first choose a non-empty bucket $y\in \phi(X)$ uniformly at random (u.a.r.), then choose u.a.r. a point $x\in \phi^{-1}(y)\cap X$. 
The set $G$ is then constructed by repeating the two-level sampling to sample points from $X$ \emph{without replacement}.
Let $Y$ be a set such that $|Y|\le  T$, for simplicity, we assume that $\phi(Y)\cap \phi(X\setminus Y)=\emptyset$. 
The high-level idea is that, if $|\phi(Y)|\ge |\phi(X\setminus Y)|$, meaning that $Y$ occurs the majority of buckets, then the two-level sampling will, with constant probability, return a point that is from $Y$. 
Thus, by repeating two-level sampling without replacement enough times to obtain $G$, it is likely that $G$ contains most of the points from $Y$, resulting in $|\phi(Y\setminus G)|\le  |\phi(X\setminus Y)|$. 
Finally, we have $|\phi(X\setminus G) |\le |\phi(X\setminus Y)| + |\phi(Y\setminus G )|\le 2|\phi(X\setminus Y)|$. 
It remains to implement the two-level sampling \emph{without replacement} in the streaming setting, for which we employ a subroutine of two-level $\ell_0$-sampler proposed in~\cite[Lemma 3.3]{arxiv.2204.02095}.
\paragraph{Extract Sparse Subset}
We then present an algorithm that, given a mapping $\phi$ and a dataset $X$, if $|\phi(X)|$ is bounded, then the algorithm recovers all the parts of $\calP_{\phi}(X)$ with small sizes from the data stream. The algorithm is an application of a two-level extension of the sparse recovery (from a frequency vector); see e.g.~\cite{Cormode06Combinatorial}, and we provide a proof in Appendix~\ref{sec:find light parts}.
\begin{restatable}[Indentify Sparse Subsets]{lemma}{IndentifySparseSubsets}
    \label{lem:find light parts}
    There exists a streaming algorithm that, given integers $0<\delta<1$, $ N, M>0$, a mapping $\phi:\R^d\to\R^d$ such that the value for any point in $\R^d$ can be evaluated in space $\poly(d)$, and a dataset $X\subseteq [\Delta]^d$ represented as a dynamic stream, returns a collection of subsets of $X$ or $\perp$. 
If $|\calP_{\phi}(X)|\le  N$, the algorithm returns
    $\{P\in\calP_{\phi}(X): |P|\le  M\}$; otherwise, it returns $\perp$. The algorithm uses space $\tilde O\left( N M\cdot \poly(d\log(\delta^{-1}\Delta))\right)$ and fails with probability at most $\delta$.
\end{restatable}
\begin{remark}
    A straightforward corollary of~\Cref{lem:bounded size} is that $|\calP_\phi(X)| = |\phi(X)|\le k\Lambda + m + (\Gamma/\lambda)^z\cdot\OPT_z^{(m)}(X)$. Hence, one might consider the naive approach of directly identifying all the sparse parts in $\calP_{\phi}(X)$.
However, we note that in the streaming setting, verifying whether a part $P\in\calP_\phi(X)$ contains more than $(1+\eps^{-1})m$ points may require $\Omega(m\eps^{-1})$ space. This would result in a factor of $m^2$ in the total space, which is not acceptable. 
\end{remark}

\paragraph{Putting Things Together} 
Now we are ready to prove~\Cref{thm:dynamic coresets} by combining all aforementioned subroutines together.
Let us assume momentarily that we have a guess $\GOPT$ of $\OPT_z^{(m)}(X)$, and we will remove this assumption later.
For simplicity, let $\calA$ and $\calB$ denote the algorithm of~\Cref{lem:isolated} and the algorithm of~\Cref{lem:find light parts}, respectively. Moreover, let $\calC$ denote a streaming algorithm that can construct an $\eps$-coreset for \kzC of a dataset from $[\Delta]^d$ presented as a dynamic stream, using space $W(d,\Delta,k,\eps^{-1},\delta^{-1})$ and failing with probability $\delta$.
Then, our algorithm is described in~\Cref{alg:streaming coreset}.

\begin{algorithm}
    \caption{Coreset construction for \kzmC of $X$ presented by a dynamic stream}
    \label{alg:streaming coreset}
    \begin{algorithmic}[1]
        \Require{a guess $\GOPT$ of $\OPT_z^{(m)}(X)$}
        \State construct a consistent hash $\phi$ with diameter bound $\lambda:=(z+1)^{-\xi}\cdot \left(\frac{\eps\cdot \GOPT}{m} \right)^{1/z}$ for sufficiently large constant $\xi>0$ using~\Cref{lem:consistent hashing}
        \State run $\calA$ on inputs $\delta, T := 2^{O(z\log(z+1))}\cdot md^z\eps^{-1}, \phi$ and the dynamic stream that represents $X$ to obtain a subset $G\subseteq X$.
\State run $\calB$ on inputs $\delta, N := O(kd\log d), M := (1+\eps^{-1})m, \phi$ and a dynamic stream that represents $X\setminus G$ \label{alg line:calB}
        \State if $\calB$ returns $\perp$, the algorithm returns $\perp$; otherwise, let $\calS\gets$ the output of $\calB$ and let $X_{\calS}\gets\bigcup_{P\in\calS}P$
\State run $\calC$ on a dynamic stream that represents $X_{\calD}:=X\setminus (G\cup X_{\calS})$ to obtain an $\eps$-coreset $S_{\calD}$ of $X_{\calD}$ \label{alg line:calC}
        \\
        \Return $S_{\calD}\cup X_{\calS}\cup G$ \Comment{algorithm fails if either $\calB$ or $\calC$ fails}
    \end{algorithmic}
\end{algorithm}

Clearly,~\Cref{alg:streaming coreset} uses space that is the cumulative space of $\calA$, $\calB$, and $\calC$, amounting to $2^{O(z\log z)}\cdot \tilde O\left((k+d^z)m\eps^{-1}\cdot \poly(d\log\Delta)\right) + W(d,\Delta,k,\eps^{-1},\delta^{-1})$.

~\Cref{alg:streaming coreset} can be implemented in one pass. 
Specifically, We initially run all three algorithms $\calA$, $\calB$ and $\calC$ on the input stream. Once the stream ends, we obtain $G$ from $\calA$, and proceed to run $\calB$ and $\calC$ on a stream composed of deletions of points in $G$. This is equivalent to running the algorithms on a stream representing $X\setminus G$, hence implementing Line~\ref{alg line:calB}. By applying the same approach, we can implement Line~\ref{alg line:calC}.

\subparagraph{Error Analysis} Let us first assume that the guess $\GOPT$ satisfies that $\OPT_z^{(m)}(X)/(z+1)\le \GOPT\le \OPT_z^{(m)}(X)$, and we will remove this assumption later.
\Cref{alg:streaming coreset} fails to return a coreset if and only if $\calB$ or $\calC$ fails, or if $\calB$ returns $\perp$. 
Applying union bound, both algorithms $\calB$ and $\calC$ succeed simultaneously with probability $1-2\delta$.
Furthermore, by~\Cref{lem:bounded size}, since $\GOPT\ge \OPT_z^{(m)}(X)/(z+1)$, there exists a set $F$ with $|F|\le 2^{O(z\log(z+1))}\cdot md^z\eps^{-1}$ such that $|\phi(X\setminus F)|\le O(kd\log d)$. In this case, $G$ returned by $\calA$ also satisfies that $|\phi(X\setminus G)|\le O(kd\log d)$ with probability $1-\delta$. 
Given this condition and the success of $\calB$, we can conclude, according to~\Cref{lem:find light parts}, that $\calB$ does not return $\perp$.
Consequently, the probability of that~\Cref{alg:streaming coreset} succeeds in returning a coreset is at least $1-3\delta$.

Conditioned on the success of~\Cref{alg:streaming coreset}, we have, as guaranteed by~\Cref{lem:find light parts}, that the set $X_{\calD}:=X\setminus (G\cup X_{\calS})$ satisfies the property that, for every $P\in \calP_\phi(X)$ with $P\cap X_{\calD}\neq\emptyset$, $|P\cap X_{\calD}|\ge (1+\eps^{-1})m$. 
Hence by~\Cref{thm:main1}, $S_{\calD}$ is also an $\left(O(\eps),O(\eps)\cdot \GOPT\right)$-coreset of $X_{\calD}$ for \kzmC, implying that $S:=G\cup X_{\calS}\cup S$ is an $\left(O(\eps),O(\eps)\cdot \GOPT\right)$-coreset of the original dataset $X$ for \kzmC due to the composability of the coreset (\Cref{fact:composability}).
Since the guess $\GOPT$ further satisfies that $\GOPT\le \OPT_z^{(m)}(X)$, $S$ becomes an $O(\eps)$-coreset for \kzmC. It suffices to scale both $\eps$ and $\delta$ by a constant factor.

\subparagraph{Removing Assumption of Knowing $\GOPT$}
Finally, we remove the assumption of knowing $\GOPT$ which satisfies that $\OPT_z^{(m)}(X)/(z+1)\le \GOPT\le \OPT_z^{(m)}(X)$ in advance.
We assume without loss of generality that $X$ contains more than $k + m + 1$ distinct points from $[\Delta]^d$, since otherwise one could use a sparse recovery structure to fully recover all the data points.
Then, we have $1\le \OPT^{(m)}_z(X)\le n\cdot (\sqrt d \Delta)^z$.
We run in parallel $\tau:=\lfloor \log_{z+1}(n\cdot (\sqrt d \Delta)^z)\rfloor + 1$ instances $\calD_0,\dots,\calD_{\tau-1}$ of~\Cref{alg:streaming coreset} (resulting in a space increase by a factor of only $\tau \le O(zd\log\Delta)$), where the $\calD_i$ tries $\GOPT_i:=(z+1)^i$. 
We set the failure probability of each instance to be $\delta/\tau$.
Then, with probability $1-\delta$, all instances succeeds, i.e., return a coreset or $\perp$.
We select the output of the $\calD_j$, where $j\ge 0$ is the smallest index such that $\calD_j$ does not return $\perp$, as the final output. 
Condition on that all instances do not fail, 
let $i^* := \lfloor \log_{z+1}(\OPT_z^{(m)}(X)) \rfloor$, we know that the instance $D_{i^*}$ must return a coreset.
Hence, we have $i\le i^*$, which implies that $\GOPT_i\le \OPT_z^{(m)}(X)$, and thus, $\calD_i$ returns an $\eps$-coreset.

      \subsection{$W_2\cdot\poly(d\log(\delta^{-1}\Delta))$-Space: Streaming Implementation of~\Cref{thm:reduction2}}
\label{sec:streaming alg}
We present the streaming implementation of \Cref{thm:reduction2} and achieve the $W_2\cdot \poly(d\log\Delta)$ space complexity.
As an important step in \Cref{thm:reduction2}, one needs to first run \Cref{alg:size-preserving} to turn a vanilla coreset to a coreset with a weaker size-preserving property.
Hence, we also start with the streaming implementation of this \Cref{alg:size-preserving}.
Then similar to \Cref{thm:reduction2}, we calibrate its weight to ensure it is truly size-preserving,
and we can conclude that this resultant coreset works for \kzmC by \Cref{cor:second}.
We would reuse the gadgets of streaming algorithm developed in \Cref{sec:streaming coreset} without mentioning how they are implemented again.

\paragraph{Streaming Implementation of~\Cref{alg:size-preserving}}
Recall that the input for~\Cref{alg:size-preserving} consists of a dataset $X$, which is presented as a dynamic point stream in this case, along with two parameters $\mu\geq 0$ and $k'\geq 1$. 
We would then discuss how each line of \Cref{alg:size-preserving} is implemented.

\subparagraph{Line~\ref{alg line:sparse partition}} 

In the streaming setting, we cannot expect to explicitly compute and store a sparse partition of the dataset. Therefore, instead, we compute a consistent hashing $\phi$ using~\Cref{lem:consistent hashing} with a diameter bound $\mu$ at the very beginning, which takes $\poly(d\log\Delta)$ space to store the description of $\phi$.
According to~\Cref{lem:consistent hashing} and~\Cref{def:sparse partition}, we know that the partition $\calP_\phi(X) = \{\phi^{-1}(y)\cap X: y\in \phi(X)\}$ is a $(\mu,\Lambda,\Gamma)$-sparse partition with $\Lambda = O(d)$ and $\Gamma = O(d\log d)$. 
Moreover, recall from~\Cref{lem:consistent hashing} that for any give point $x\in \R^d$, the value of $\phi(x)$ can be computed using space $\poly(d)$. 
Hence, the image $\phi([\Delta]^d)$ has a size of at most $\Delta^d \cdot 2^{\poly(d)}\le \Delta^{\poly(d)}$, since we can encode the input point and the computation process of $\phi(x)$ by a binary string of length $O(d\log\Delta) + \poly(d)$. 
Hence, we assume that the image of $\phi$ is $[\Delta^{\poly(d)}]$ instead of $\R^d$.

\subparagraph{Line~\ref{alg line:mapping} and Line~\ref{alg line:vanilla coreset}}
The main difficulty arises in implementing Line~\ref{alg line:mapping} and Line~\ref{alg line:vanilla coreset}, where we map the input points into an $O(z\eps^{-1}\cdot \diam(X)\cdot |X|^{1/z})$-separated $|\calP_\phi(X)|$-duplicated space $M^\dup$ of the discrete Euclidean space and construct a vanilla coreset on $M^\dup$.

To accomplish this in the streaming setting, we prove in~\Cref{lem:duplication for Euclidean case} that it suffices to map each point $x$ to some $(x,\phi(x)\cdot w)\in \R^{d+1}$ where $w = O(z\eps^{-1}\cdot \diam(X)\cdot |X|^{1/z})$, and then construct a vanilla coreset on $\R^{d+1}$.
Notice that $\diam(X)\le d\Delta$ and $|X|\le \Delta^d$, we pick $w':= z\eps^{-1}\cdot \Delta^{c\cdot d}$ for sufficiently large constant $c\ge 1$. 
Then, our implementation simply converts each insertion/deletion of a point $x$ to the insertion/deletion of the point $(x,\phi(x)\cdot w')$, and feeds these resulting insertions/deletions to a streaming algorithm for constructing an $O(\eps)$-coreset for \tzC{k'}.
Since $\phi(x)\cdot w' \le z\eps^{-1}\cdot \Delta^{\poly(d)}$, we have $(x,\phi(x)\cdot w')\in [z\eps^{-1}\cdot \Delta^{\poly(d)}]^{d+1}$. Hence, the implementation of Line~\ref{alg line:mapping} and Line~\ref{alg line:vanilla coreset} uses space $W\left(d+1, z\eps^{-1}\Delta^{\poly(d)}, k', O(\eps^{-1}),O(\delta^{-1})\right)$.

\subparagraph{Line~\ref{alg line:preimage}} 
Line~\ref{alg line:preimage} can be directly implemented after the stream ends. Hence, we complete the implementation of~\Cref{alg:size-preserving} with a space complexity of $W(d+1,z\eps^{-1}\Delta^{\poly(d)},k',O(\eps^{-1}), O(\delta^{-1})) + \poly(d\log\Delta)$. We denote by $\calA$ the streaming version of~\Cref{alg:size-preserving}.

\paragraph{Streaming Implementation of~\Cref{thm:reduction2}} 

Similar to the steps in \Cref{sec:streaming coreset}, we assume that we have a guess $\GOPT = \Theta\left(\OPT_z^{(m)}(X)\right)$. Set $\mu:=2^{-O(\log(z+1))}\cdot \eps\cdot \left(\frac{\GOPT}{m} \right)^{1/z}$, and let $\lambda:= O(\frac{\eps\mu}{z\Gamma}) = 2^{-O(\log(z+1))}\cdot \eps^{2}\cdot \left(\frac{\GOPT}{m} \right)^{1/z}\cdot d^{-1}$. 
Let $\phi$ and $\phi'$ be two consistent hashing with diameter bounds $\mu$ and $\lambda$, respectively.
According to~\Cref{lem:bounded size}, there exists a set $F\subseteq X$ with $|F|\le 2^{O(z\log z)}\cdot O(m(d/\eps)^{2z})$ such that $|\phi'(X\setminus F)|\le k\Lambda$. 
Therefore, we first apply the algorithm of~\Cref{lem:isolated} to extract $G\subseteq X$ such that $|\phi'(X\setminus G)|\le 2k\Lambda$ using space $2^{O(z\log z)}\cdot\tilde O(m(d/\eps)^{2z})\cdot \poly(d\log(\delta^{-1}\Delta))$. 
Then we have $X\setminus G$ admits a $\lambda$-bounded partition of size at most $2k\Lambda$.
Moreover, for every $y\in \phi'(X)$, since $\diam(\phi'^{-1}(y))\le \lambda\le \mu/\Gamma$, we have that $|\phi(\phi'^{-1}(y))|\le \Lambda$.
This implies that $|\phi(X\setminus G)|\le |\phi'(X\setminus G)|\cdot \Lambda \le 2k\Lambda^2$.

Set $k' = (k + 2k\Lambda^2 + 2k\Lambda)\Lambda = k\poly(d)$. 
We run $\calA$, the streaming version of~\Cref{alg:size-preserving}, on input consisting of a stream representing $X\setminus G$, $\mu$, and $k'$.
The implementation of running $\calA$ on a stream representing $X\setminus G$ is the same as that in the first algorithm of~\Cref{thm:dynamic coresets}. Specifically, we first run $\calA$ and an algorithm identifying $G$ in parallel. After the stream ends, we continue to run $\calA$ on a stream consisting of deletions of points in $G$.
By doing so, we obtain a weighted set $S$ such that $S$ is an $\eps$-coreset of $X\setminus G$ for \tzC{k+2k\Lambda^2} and is nearly size-preserving with respect to $\calP_\phi(X\setminus G)$ with $|\calP_{\phi}(X\setminus G)|\le 2k\Lambda^2$. This step uses space $W(d+1,z\eps^{-1}\Delta^{\poly(d)},k\poly(d),O(\eps^{-1}),O(\delta^{-1})) + \poly(d\log(\delta^{-1}\Delta))$.

In addition, we run a sparse recovery algorithm (\Cref{lem:sparse recovery}) in parallel with $\calA$ on a stream representing $\phi(X\setminus G)$, which uses space $\tilde O(k\cdot \poly(d\log(\delta^{-1}\Delta)))$ to return the frequencies of each $y \in \phi(X\setminus G)$ if $|\phi(X\setminus G)| \leq 2k\Lambda^2$ and fails with a probability of at most $\delta$.

Then, we can calibrate the weight of $S$ so that $S$ is exactly size-preserving with respect to $\calP_\phi(X\setminus G)$. By~\Cref{cor:second}, we have that $S$ is an $O(\eps)$-coreset of $X\setminus G$ for \kzmC. 
Finally, by composability of coresets, $S\cap G$ is the desired $O(\eps)$-coreset of $X$ for \kzmC. It remains to scale $\eps$ by a constant. The overall space is 
\begin{equation*}
    2^{O(z\log z)}\cdot\tilde O(m(d/\eps)^{2z})\cdot \poly(d\log(\delta^{-1}\Delta))+ W\left(d+1, z\eps^{-1}\Delta^{\poly(d)}, k\poly(d), O(\eps^{-1}), O(\delta^{-1})\right).
\end{equation*}
We complete the proof by removing the assumption of knowing $\GOPT$ using a similar approach as in the first algorithm of~\Cref{thm:dynamic coresets}, which needs to replace $\delta$ with $\delta/O(zd\log\Delta)$ and results in a space increase by a $O(zd\log\Delta)$ factor.

\bibliographystyle{alphaurl}
\bibliography{ref}

\begin{appendices}
     \appendixpage
     \section{Facts about Coresets}
\label{sec:composability}
\ComposabilityCoresets*

\begin{proof}
    For any $k$-point center set $C\subseteq V$ and real number $0\le h\le m$, we prove the direction $\cost_z^{(h)}(S_X\cup S_Y,C)\le (1+\eps)\cdot \cost_z^{(h)}(X\cup Y,C) + \eta_1 + \eta_2$; while for the reverse direction, $\cost_z^{(h)}(X\cup Y,C)\le (1+\eps)\cdot \cost_z^{(h)}(S_X\cup S_Y,C) + \eta_1 + \eta_2$, the proof is almost the same.
Let $h_1,h_2\ge 0$ such that $h_1+h_2 = h$ and $\cost_z^{(h_1)}(S_X,C) + \cost_z^{(h_2)}(S_Y,C) = \cost_z^{(h)}(S_X\cup S_Y, C)$. Then we have 
    \begin{align*}
        \cost_z^{(h)}(X\cup Y,C)\quad\le&\quad \cost_z^{(h_1)}(X,C) + \cost_z^{(h_2)}(Y,C)\\
        \le&\quad (1+\eps)\cdot \left(\cost_z^{(h_1)}(S_X,C) + \cost_z^{(h_2)}(S_Y,C) \right) + \eta_1+\eta_2\\
        =&\quad (1+\eps)\cdot \cost_z^{(h)}(S_X\cup S_Y,C) + \eta_1+\eta_2,
    \end{align*}
    where the first step is due to the optimality of $\cost_z^{(h)}$ and the second follows from the coreset guarantee of $S_X$ and $S_Y$.
\end{proof}

\ASufficientCondition*
\begin{proof}
    For every $0\le h\le m$, by definition of robust clustering, we have 
    \begin{equation}
        \cost_z^{(h)}(X,C) = \min_{a_1,\dots,a_t\ge 0:\sum_{i=1}^t a_i = h} \sum_{i=1}^t\cost_z^{(a_i)}(X_i,C),
    \end{equation}
    and the similar holds for $\cost_z^{(h)}(S,C)$. 
    We prove the direction $\cost_z^{(h)}(S,C)\le (1+\eps)\cdot \cost_z^{(h)}(X,C)$, and the proof for the other direction follows similarly.
    
    Let $l_1,\dots l_t\ge 0$ with $\sum_{i=1}^t l_i=h$ be real numbers such that 
    $\cost_z^{(h)}(X,C) = \sum_{i=1}^t \cost_z^{(l_i)}(X_i,C)$. Then we have 
    \begin{equation}
        \cost_z^{(h)}(S,C) \le \sum_{i=1}^t \cost_z^{(l_i)}(S_i,C)\le (1+\eps)\cdot\sum_{i=1}^t\cost_z^{(l_i)}(X_i,C) = (1+\eps)\cdot \cost_z^{(h)}(X,C).
    \end{equation}
    The claim follows.
\end{proof}

      \section{Separated Duplication of Various Metric Families}
\label{sec:application}
In this section, we discuss how \Cref{thm:reduction2} can be applied in various metric spaces. The main challenge lies in demonstrating that there exists a $w$-separated $h$-duplication $M^\dup_{h,w}$ of $M$ has a ``dimension'' that is almost equal to that of $M$ for any $w$ and $h$. We accomplish this on a case-by-case basis, examining metric spaces for which the construction of coresets is well-studied, including Euclidean space, doubling metric spaces, and graph metrics.

Without loss of generality, in the following presentation, we assume that both the coreset size $N$ and the runtime $T$ are monotonic with respect to $n,k,\eps^{-1}$.

\subsection{$\ell_p$ Metrics}

We begin by discussing the $\ell_p$ metric space for $p\ge 1$, which includes Euclidean space by setting $p=2$. The $\ell_p$ metric space $M=(\R^d, \ell_p)$ is defined on $\R^d$, with distance function computed using $\ell_p$ norm. 
Coresets for clustering have been extensively studied in $\ell_p$ metric spaces, especially in Euclidean space. The state-of-the-art coreset construction has already achieved a size bound of $\text{poly}(k\eps^{-1})$, which is independent of the input size and dimension (see, e.g.,~\cite{Cohen-addad2021New}). 

To apply \Cref{thm:reduction2} to $\ell_p$ metric spaces $M=(\R^d, \ell_p)$, we demonstrate that for any $h$ and $w$, there exists a $w$-separated $h$-duplication of $M$ that can be embedded, with no distortion, into the metric space $(\R^{d+1}, \ell_p)$, as stated in~\Cref{lem:duplication for Euclidean case}.

\begin{lemma}
    \label{lem:duplication for Euclidean case}
    Let $M=(\mathbb{R}^d,\ell_p)$ be an $\ell_p$ metric space, for integer $d\geq 1$ and real number $p\geq 1$.
For every integer $h\ge 1$ and real number $w\ge 0$, there exists a $w$-separated $h$-duplication $M^\dup$ of $M$ that can be embedded into $(\R^{d+1},\ell_p)$ with no distortion. Moreover, for any given point from $M^\dup$, its image under the embedding can be computed in $O(d)$ time.
\end{lemma}
\begin{proof}
    Given $h$ and $w$, we define a distance function $\dist_{h,w}:(\R^d\times [h])\times (\R^d\times [h])\to \R_{\ge 0}$ as follows: for every $x,y\in \R^d$ and every $i,j\in [h]$, $\dist_{h,w}((x,i),(y,j)) = \left(\|x-y\|_p^p + (|i-j|\cdot w)^p \right)^{1/p}$. 
Since $\dist_{h,w}((x,i),(y,j)) = \left(\|x-y\|_p^p + (|i-j|\cdot w)^p \right)^{1/p}\ge \max\{\|x-y\|_p, |i-j|\cdot w\}$, it follows directly that the metric space $M^\dup = (\R^d\times [h], \dist_{h,w})$ is a valid $w$-separated $h$-duplication of $M$. 
Moreover, the mapping $f:\R^d\times [h]\to \R^{d+1}$ that maps a point $(x,i)\in \R^d\times [h]$ to $(x,i\cdot w)\in \R^{d+1}$ is a feasible embedding from $M^\dup$ into $(\R^{d+1}, \ell_p)$, whose value can be efficiently computed given $(x,i)$, $h$, and $w$.
\end{proof}

According to~\Cref{lem:duplication for Euclidean case}, 
if an algorithm $\calA$ can construct vanilla coreset for $(\R^{d+1},\ell_p)$, then there exists a family $\calM^\dup = \left\{M^\dup_{h,w} \right\}_{h\ge 1,w\ge 0}$ of separated duplication of $M$ and an algorithm $\calB$ such that $\calB$ can construct vanilla coreset for every $M^\dup_{h,w}\in\calM^\dup$ by embedding the input dataset into $(\R^{d+1}, \ell_p)$, running algorithm $\calA$, and then returning the pre-image of the constructed coreset.
Since \Cref{lem:duplication for Euclidean case} also guarantees that the embedding can be computed in $O(d)$ time, the additional runtime is only $O(nd)$.

Therefore, we obtain the following corollary, where for simplicity, we assume that the size bound $N$ only depends on $d,k$ and $\eps^{-1}$.

\begin{corollary}[Reduction on $\ell_p$ metric space]
    \label{cor:Euclidean space}
    Let $p\ge 1$ be a real number. 
    Assume there exists an algorithm that given $0<\eps<1$, integers $d, k,z\ge 1$ and an $n$-point dataset $X\subseteq \R^{d}$ as input, runs in time $T(d,n,k,\eps^{-1})$ to construct an $\eps$-coreset of $X$  for \kzC on $(\mathbb{R}^d, \ell_p)$ of size $N(d,k,\eps^{-1})$.

    Then, there exists an algorithm that, given $0<\eps<1$, integers $d,k,z,m\ge 1$, an $n$-point dataset $X\subseteq \R^d$ and a $(2^{O(z)},O(1),O(1))$-approximation solution $C^*$ to \kzmC on $X$ as input, runs in time
    \begin{equation*}
        \tilde O(nkd) + \poly(kdm\eps^{-1}) + 2\cdot T(d+1,n,O(k\log^2(km\eps^{-1})),O(\eps^{-1}))
    \end{equation*}
    to compute an $\eps$-coreset of $X$ for \kzC of size 
    \begin{equation*}
        2^{O(z\log z)}\cdot O\left(m\eps^{-2z}\log^z(km\eps^{-1})\right) + 2\cdot N\left(d+1,O(k\log^2(km\eps^{-1})), O(\eps^{-1})\right).
    \end{equation*}
\end{corollary}

\subsection{Doubling Metrics}
We then consider the doubling metric space, which is an important generalization of $\ell_p$ metric spaces. Here, an important concept is the \emph{doubling dimension}~\cite{DBLP:conf/focs/GuptaKL03}, defined as the least integer $t\ge 0$, such that every ball can be covered by at most $2^t$ balls of half the radius. A metric space with a bounded doubling dimension is called a doubling metric, and we denote by $\ddim(M)$ the doubling dimension of the doubling metric $M$.
Coresets in doubling metrics have been studied in previous works~\cite{DBLP:conf/focs/HuangJLW18,Cohen-addad2021New}, which achieve a size that depends only on the doubling dimension, $k$, and $\eps^{-1}$.

For a doubling metric $M=(V,\dist)$, we have the following lemma, demonstrating that for every $h$ and $w$, there is a $w$-separated $h$-duplication of $M$ that has a doubling dimension $O(\ddim(M))$.

\begin{lemma}
    \label{lem:duplication doubling}
    For a doubling metric $M=(V,\dist)$, integer $h\ge 1$ and real number $w\ge 0$, there exists a $w$-separated $h$-duplication $M^\dup=(V\times [h],\dist')$ of $M$ such that $\ddim(M^\dup) \le 2\ddim(M) + 2$. Moreover, we can evaluate the distance $\dist'(x,y)$ in constant time for any $x,y\in V\times [h]$.
\end{lemma}
\begin{proof}
    Consider the metric space $M^\dup=(V\times [h],\dist')$ where the distance function satisfies that for every $x,y\in V$ and every $i,j\in [h]$, $\dist'\big((x,i),(y,j) \big) = \dist(x,y) + |i-j|\cdot w$.
    Clearly, $M^\dup$ is a metric space,
    as well as a $w$-separated $h$-duplication of $M$, and the distance $\dist'$ can be evaluated in constant time. We then bound the doubling dimension of $M^\dup$.

    To this end, consider any ball $B$ from $M^\dup$ with radius $r>0$.
    We decompose $B$ into $B_0:=\{x\in V:\exists i\in [h], (x,i)\in B\}$ and $B_1:=\{i\in [h]: \exists x\in V, (x,i)\in B\}$. 

    For $B_0$, since the metric $M=(V,\dist)$ has a doubling dimension of $\ddim(M)$, we have that $B_0$ can be covered by at most $2^{2\ddim(M)}$ balls of radius $r/4$. Specifically, let $C_0\subseteq V$ be such a cover, i.e., for every $x\in B_0$, $\dist(x,C_0)\le r/4$.
    
    For $B_1$, consider the line metric $M^\mathrm{line}=([h],\dist'')$ such that for every $i,j\in [h]$, $\dist''(i,j)=|i-j|\cdot w$. Clearly, this line metric $M^\mathrm{line}$ has a doubling dimension of $1$. Hence, $B_1$ can be covered by at most $4$ balls of radius $r/4$. Let $C_1$ be such a cover.

    Let $C:=\{(x,i):x\in C_0,i\in C_1\}\subseteq V\times [h]$. For every $(x,i)\in B$, let $y:=\arg\min_{p\in C_0}\dist(x,p)$ and let $j:=\arg\min_{t\in C_1}|t-j|\cdot w$. By definition, we have $\dist(x,y)\le r/4$ and $|i-j|\cdot w\le r/4$, implying that $\dist'\big((x,i),(y,j) \big) = \dist(x,y) + |i-j|\cdot w\le r/2$. Therefore, $\dist'\big((x,i), C \big)\le \dist'\big((x,i),(y,j) \big)\le r/2$. Since $|C|=|C_0|\cdot |C_1| = 2^{2\ddim + 2}$, we conclude that $M^\dup$ has a doubling dimension of at most $2\ddim + 2$.
\end{proof}

\begin{corollary}
    \label{cor:doubline}
    Assume that there exists an algorithm such that, for every $0<\eps<1$, integers $d,k,z\ge 1$, metric space $M$ with $\ddim(M)\le d$ and $n$-point dataset from $M$, it computes in time $T(d,n,k,\eps^{-1})$ an $\eps$-coreset of size $N(d,k,\eps^{-1})$ for \kzC on $M$.
    
    Then, there is an algorithm that, given $0<\eps<1$, integers $d,k,z,m\ge 1$, a metric space $M$ with $\ddim(M)\le d$, an $n$-point dataset from $M$ and a $(2^{O(z)},O(1),O(1))$-approximation solution $C^*$ to \kzmC on $X$, computes in time 
    \begin{equation*}
        \tilde O(nk) + \poly(km\eps^{-1}) + 2\cdot T(O(d),n,O(k\log^2(km\eps^{-1})),O(\eps^{-1}))
    \end{equation*}
    an $\eps$-coreset for \kzmC on $M$ of size 
    \begin{equation*}
        2^{O(z\log z)}\cdot O\left(m\eps^{-2z}\log^z(km\eps^{-1})\right) + 2\cdot N\left(O(d), O(k\log^2(km\eps^{-1})), O(\eps^{-1})\right).
    \end{equation*}
\end{corollary}

\begin{remark}
    \label{remark:finite}
    A special case of a doubling metric is the general metric $M=(V,\dist)$ with finite ambient size, for which we have a well-known fact that $\ddim(M)\le O(\log |V|)$. Although this fact enables us to use doubling dimension to indicate the relationship between $M$ and its separated duplication, we still hope to establish the relationship based on the ambient size. To achieve this, observe that in \Cref{thm:reduction2}, we actually only apply the assumed algorithm to $M^\dup_{h,w}$ with $h$ less than the data size. 
    Therefore, it suffices to restrict the family $\calM^\dup$ to include only those $M^\dup_{h,w}$ with $1 \leq h \leq |V|$, which has an ambient size at most $|V|\cdot h\leq |V|^2$.
\end{remark}

\subsection{Graph Metrics}

Given an edge-weighted graph $G=(V,E)$, we consider the metric space $M=(V,\dist)$, where $\dist$ is the shortest-path distance on the graph $G$. Coresets for graph datasets have also gained researchers' interest in recent years \cite{DBLP:conf/icml/BakerBHJK020, BJKW21, Cohen-addad2021New,Cohen-AddadD0SS25}, where they relate the coreset size to some complexity measures of the graph, such as treewidth and the size of the excluded minor. Here, we hope to establish reductions from robust coresets to vanilla coresets on graph metrics.

Given a graph $G=(V,E)$, an integer $h\geq 1$, and a real number $w\geq 0$, the following lemma provides a construction of a new graph $G'=(V',E')$ such that the shortest-path metric of $G'$ is a $w$-separated $h$-duplication of that of $G$. Furthermore, $G'$ has similar complexity to $G$ in terms of both treewidth and the size of the excluded minor.

\begin{lemma}
    \label{lem:duplication graph}
    For any graph $G=(V,E)$, integer $h\geq 1$ and real number $w\geq 0$, there exists a graph $G'=(V\times [h], E')$ such that the shortest-path metric of $G'$ is a $w$-separated $h$-duplication of the shortest-path metric of $G$. Moreover, it holds that
    \begin{enumerate}
        \item $\tw(G')\leq \tw(G)$, where $\tw(G)$ denotes the treewidth of $G$; and
        \item if $G$ excludes a fixed minor $H$, then $G'$ excludes the same minor $H$.
    \end{enumerate}
\end{lemma}
\begin{proof}
    We construct the graph $G'=(V\times [h],E')$ by connecting the $n$ copies of $G$ in a ``line'' (see~\Cref{fig:SeparatedDuplication} for an illustration). Specifically, we pick an arbitrary vertex $v\in V$. Then, we define the edges of $G'$ to be $E':=\left\{\left((x,i),(y,i) \right):(x,y)\in E, i\in [h]\right\}\cup\left\{\left((v,i),(v,i+1) \right): i\in [h-1] \right\}$. For every $x,y\in V$, $i\in [h]$, we assign a weight to edge $((x,i),(y,i))$ that is equal to the weight of $(x,y)\in E$ in $G$. For every $i\in[h-1]$, we assign a weight $w$ to the edge $((v,i),(v,i+1))$.

    Let $\dist_G$ and $\dist_{G'}$denote the shortest-path distance function of $G$ and $G'$, respectively.
    It is easy to verify that for every $x,y\in V$ and $i\in [h]$, the length of the shortest-path between $(x,i)$ and $(y,i)$ in $G'$ is the same as that between $x$ and $y$ in $G$, i.e., $\forall x,y\in V,\forall i\in[h]$, $\dist_{G'}((x,i),(y,i)) = \dist(x,y)$. 
Then, consider any $i,j\in [h]$ with $i\neq j$, we have 
    \begin{align*}
        \dist_{G'}((x,i),(y,j)) \quad=&\quad \dist_{G'}((x,i),(v,i))+\dist_{G'}((v,i),(v,j)) + \dist_{G'}((y,j),(v,j))\\
        \ge&\quad \dist_{G}(x,v)+w + \dist_{G}(y,v)\\
        \ge&\quad w+\dist_{G}(x,y),
    \end{align*}
    where the last inequality is due to $\dist_G(x,y) = \min_{z\in V}\left(\dist_G(x,z) + \dist_G(y,z)\right)$. Hence, we certify that $(V\times [h],\dist_{G'})$ is a feasible $w$-separated $h$-duplication of $(V,\dist_G)$.

    For the size of the excluded minor, it is clear by the construction that, if $G$ excludes a fixed minor $H$, then $G'$ also excludes this fixed minor $H$.

    For treewidth, assume that $\tw(G) = t$.
    For every $i\in [h]$, let $G_i$ denote the $i$-th copy of $G$ in $G'$, namely, the subgraph induced by the vertex set $\{(x,i):x\in V\}$. Let $\mathcal{T}_i$ with node set $\calV_i\subseteq 2^V$ (we call each node in $\calV_i$ a bag) be a tree decomposition of $G_i$ such that the maximum bag size $\max_{U\in V_i}|U|$ is equal to $t + 1$ (see, e.g., Definition 2.1 of ~\cite{DBLP:conf/icml/BakerBHJK020} for the definition).

    Then, we can construct a tree decomposition for $G'$ by connecting these $\calT_i$'s. Specifically, for every $i$, we pick an arbitrary bag $U_i\in\calV_i$ such that $v\in U_i$. Then, for every $i\in [h-1]$, we introduce a new bag $\{(v,i),(v,i+1)\}$ and connect it to $U_i$ and $U_{i+1}$. It is easy to verify that this operation results in a valid 
    tree decomposition of $G'$ whose maximum bag size is still $t$. Hence, we conclude that $\tw(G') \leq t$.
\end{proof}

\begin{figure}[t]
	\centering
	\includegraphics[width=0.95\textwidth]{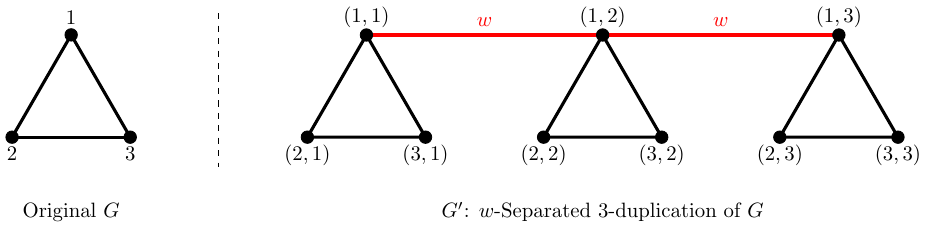}	
	\caption{
        Illustration of the construction of separated duplication of a graph with the weights of edges omitted. On the left, we show an original graph $G$, using a triangle graph as an example. On the right, we demonstrate our construction for the $w$-separated $3$-duplication of $G$, where the three triangle graphs connected by black edges represent the s of $G$, and red edges are weighted by $w$ to satisfy the separation requirement.    
	}
	\label{fig:SeparatedDuplication}
	
\end{figure}

According to~\Cref{lem:duplication graph},~\Cref{thm:reduction2} is applicable to graphs with bounded treewidth and minor-free graphs. Here, we state the reduction result only for graphs with bounded treewidth. 
The result for minor-free graphs is nearly identical, with the only difference lying in the dependence of $N$ and $T$.

\begin{corollary}
    \label{cor:treewidth}
    Assume that there exists an algorithm that, given $0<\eps<1$, integers $t,k,z\ge 1$, a graph $G=(V,E)$ with $\tw(G)\le t$ and an $n$-point dataset $X\subseteq V$ as input, runs in time $T(t,n,k,\eps^{-1})$ to compute an $\eps$-coreset of size $N(t,k,\eps^{-1})$ for \kzC on the shortest-path metric space of $G$.
    
    Then, there is an algorithm that, given $0<\eps<1$, integers $t,k,z\ge 1$, a graph $G=(V,E)$ with $\tw(G)\le t$, an $n$-point dataset $X\subseteq V$ and a $(2^{O(z)},O(1),O(1))$-approximation solution $C^*$ to \kzmC on $X$ as input, runs in time \begin{equation*}
        \tilde O(nk) + \poly(km\eps^{-1}) + 2\cdot T(t,n,O(k\log^2(km\eps^{-1})),O(\eps^{-1}))
    \end{equation*}
    to construct an $\eps$-coreset of size \begin{equation*}
        2^{O(z\log z)}\cdot O\left(m\eps^{-2z}\log^z(km\eps^{-1})\right) + 2\cdot N\left(t, O(k\log^2(km\eps^{-1})), O(\eps^{-1})\right)
    \end{equation*}
    for \kzmC on the shortest-path metric of $G$.
\end{corollary}      
\section{Missing Proofs of~\Cref{sec:streaming coreset}}
\subsection{Proof of~\Cref{lem:isolated}}
\label{apd:isolated}

\ExtractIsolatedPoint*

Firstly, let us consider an \emph{offline} two-level sampling procedure:
we first choose $y\in \phi(X)$ uniformly at random (u.a.r.), then choose u.a.r. a point $x\in \phi^{-1}(y)\cap X$. Then, $G$ is constructed by repeating the two-level sampling to sample points from $X$ \emph{without replacement}.
\begin{lemma}
    \label{lem:offline isolated}
    Let $G$ contain $O(T\log(T\delta^{-1}))$ points, where the points are sampled by the two-level uniform sampling on $X$ without replacement. Then for any subset $F\subseteq X$ with $|F|\le T$, it holds that $\Pr\left[|\phi(X\setminus G)|\le 2\cdot |\phi(X\setminus F)|\right]\ge 1-\delta$.
\end{lemma}
\begin{proof}[Proof Sketch]
    
Fix an arbitrary $F\subseteq X$ with $|F|\le T$, let $\beta:= |\phi(X\setminus F) |$.
    The high-level idea is that, if $|\phi(X)|\le 2\beta$ already holds, it suffices to return an arbitrary set. Otherwise, we color the points in buckets $\phi(X\setminus F)$ red, i.e., a red point $x$ satisfies that $\phi(x)\in \phi(X\setminus F)$; and color the other points blue.
    Let $R$ denote the red points, and $B$ denote the blue points. By definition, we have $\phi(R)=\phi(X\setminus F)$, and it is easy to verify that $B\subseteq F$, and thus $|B|\le |F|\le T$. Moreover, $|\phi(B)| = |\phi(X)| - |\phi(R)|\ge \beta = |\phi(R)|$. This implies that if we sample a point $x$ via the two-level sampling, $x$ is blue with probability at least $1/2$. Hence, by repeatedly sampling without replacement $O(\log(T\delta^{-1}))$ times via the two-level sampling procedure, with probability at least $1-\delta/T$, we will obtain at least one point from $B$. Since $|B|\le T$, let $G$ denote the set of $O(T \log(T\delta^{-1}))$ points sampled via two-level sampling without replacement, we have $|\phi(B\setminus G)|\le \beta$ with probability at least $1-\delta$. Hence, $|\phi(X\setminus G)| \le |\phi(A) | + |\phi(B\setminus G) |\le 2\beta$.
\end{proof}

    \paragraph{Streaming Implementation}
    We utilize an implementation proposed by~\cite{arxiv.2204.02095}.
    \begin{lemma}[Two-level $\ell_0$-sampler,~{\cite[Lemma 3.3]{arxiv.2204.02095}}]
        \label{lem:two-level sampler}
        There is a randomized algorithm, that given as input a matrix $M\in \R^{M\times N}$, with $M\le N$ and integer entries bounded by $\poly(N)$, that is present a stream of additive entry-wise updates, returns an index-pair $(i,j)$ of $M$, where $i$ is chosen u.a.r. from the non-zero rows, and then $j$ is chosen u.a.r. from the non-zero columns in that row $i$. The algorithm uses space $\poly(\log (\delta^{-1}N))$ and fails with probability at most $\delta$.
    \end{lemma}
    Recall that for any point $x\in \R^d$, the value of $\phi(x)$ can be computed in space $\poly(d)$.
    Hence, the image $\phi([\Delta]^d)$ has a size of at most $\Delta^d \cdot 2^{\poly(d)}\le \Delta^{\poly(d)}$, since we can encode the input point and the computation process of $\phi(x)$ by a binary string of length $O(d\log\Delta) + \poly(d)$. 
    Therefore, we assume without loss of generality that the image of $\phi$ is $[\Delta^{\poly(d)}]$ instead of $\R^d$.
    
    We convert each update to the dataset given by the stream into an update to a frequency matrix $M$, where rows correspond to all images of $\phi([\Delta]^d)$, and columns correspond to all points in $[\Delta]^d$. This leads to $M\le N\le \Delta^{\poly(d)}$.
To achieve sampling $\sigma:=O(T\log(T\delta^{-1}))$ points without replacement, 
    we maintain $\sigma$ independent two-level $\ell_0$-samplers $l_1,\dots,l_\sigma$ during the stream simultaneously using~\Cref{lem:two-level sampler} with each failing with probability $\delta/\sigma$. At the end of stream, we begin to sample.
For the $i$-th sampling, we use $l_i$ to obtain an index pair $(\phi(x_i),x_i)$ (in case $l_i$ fails, no action is taken), and then we update $l_{i+1},\dots,l_{\sigma}$ by decreasing the frequency of $(\phi(x_i),x_i)$ by $1$, which means we remove the point $x_i$ from the dataset.
By union bound,
the probability that all samplers succeed is at least $1-\delta$, and hence, by~\Cref{lem:offline isolated}, the sampled $G:=\{x_1,\dots,x_{\sigma}\}$ satisfies $|\phi(X\setminus G)|\le 2\cdot |\phi(X\setminus F)|$ for a fixed $F$ with probability at least $1-2\delta$. We finish the proof by rescaling $\delta$.

\subsection{Proof of~\Cref{lem:find light parts}}
\label{sec:find light parts}
\IndentifySparseSubsets*

Similarly, we assume that the image of $\phi$ is $[\Delta^{\poly(d)}]$ instead of $\R^d$.
We first present an offline algorithm in~\Cref{alg:light parts} and then we discuss how to implement it in dynamic streams. Assume the algorithm succeeds in returning a collection $\calS$ of subsets, as can be seen in Line~\ref{alg line:find i}, we have that, for every $y\in \phi(X)$, there exists $i\in [w]$ such that $B^{(i)}_{h_i(y)} = \phi^{-1}(y)\cap X$. This implies that we recover the partition $\calP_\phi(X)=\{\phi^{-1}(y)\cap X:y\in \phi(X) \}$ of $X$ induced by $\phi$ exactly, and return $\calS:=\{P\in \calP_\phi(X): |P|\le  M\}$, hence achieving the guarantee of~\Cref{lem:find light parts}.

\begin{algorithm}
    \caption{Identify light parts (offline)}
    \label{alg:light parts}
    \begin{algorithmic}[1]
        \State let $w\gets \Theta(\log (\delta^{-1}N))$
        \State let $h_1,\dots,h_w:[\Delta^{\poly(d)}]\to [2N]$ be independent 2-universal hash functions. \label{alg line: hash}
        \For{$i\in [w], j\in [2N]$} \label{alg line: for}
            \State let $B^{(i)}_j\gets \{x\in X: h_i(\phi(x))=j\}$ \label{alg line:B^i_j}
        \EndFor \label{alg line: endfor}
        \State return $\perp$ if $|\phi(X)| > N$
        \State $\calS\gets \emptyset$ \label{alg line:remaining begins}
        \For{$y\in \phi(X)$} \label{alg line:phi(X)}
            \State find $i\in[w]$ s.t. $\forall y'\in \phi(X),y'\neq y$, it holds that $h_i(y')\neq h_i(y)$ \label{alg line:find i} 
            \State algorithm fails if no such $i$ exists
            \State if $|B^{(i)}_{h_i(y)}|\le M$, let $\calS\gets \calS\cup \{B^{(i)}_{h_i(y)}\}$ \Comment{$B^{(i)}_{h_i(y)} = \phi^{-1}(y)$} \label{alg line:|B^i_j|<= beta}
        \EndFor
        \State return $X_{\calS}$ \label{alg line:remaining ends}
    \end{algorithmic}
\end{algorithm}

\paragraph{Streaming Implementation}
Before the stream starts, the algorithm builds the 2-universal hash functions $h_1,\dots, h_w$, which can be implemented using space $\poly(\log (\Delta^{\poly(d)})) = \poly(d\log\Delta)$. 
Once the hash functions  have been constructed, it becomes straightforward to update and maintain $B^{(i)}_j$ for $i\in[w],j\in[2 N]$ and $\phi(X)$ during the stream if there is no space limit. 
When aiming to manage them with a limited space, we first notice that we actually need a subroutine that exactly maintains $B^{(i)}_j$ (or $\phi(X)$) only when its size does not exceed $ M$ (or $ N$, respectively).
In cases where the size exceeds the threshold, it suffices to instead return a symbol, such as $\perp$, to indicate this.
Therefore, we can similarly consider the task of maintaining $B_j^{(i)}$ as a sparse recovery problem, where the goal is to exactly recover a vector $v\in [0,1]^X$ with $v_x=\mathbf{1}[x\in B_j^{(i)}]$\footnote{$\mathbf{1}[\calE]$ is an indicator variable which equals $1$ if the event $\calE$ happens, and equals $0$ otherwise.} when the support of $v$ is no greater than $ M$.
\begin{lemma}[Sparse recovery,~\cite{Cormode06Combinatorial}]
    \label{lem:sparse recovery}
    There exists a streaming algorithm that, given $0<\delta<1$, integers $I,J,K\ge 1$, and a frequency vector $v\in [-J,J]^I$ presented as a dynamic stream, where we denote its support by $\supp(v):=\{i\in[I], v_i\neq 0\}$, uses space $O(K\cdot \poly(\log(\delta^{-1}IJ)))$. If $|\supp(v)|\le K$, then with probability $1-\delta$, the algorithm returns all the elements in the support $\supp(v)$ and their frequencies. Otherwise, it returns $\perp$.
\end{lemma}
As a result, we can run $2 N w$ instances of the sparse recovery algorithm of~\Cref{lem:sparse recovery} in parallel, with each instance corresponding to one of the sets $B_i^{(j)}$.
Similarly, we can employ a separate instance for the purpose of maintaining $\phi(X)$. The total space is $2 N w\cdot  M\poly(d\log(\delta^{-1}\Delta)) +  N\poly(d\log(\delta^{-1}\Delta)) = \tilde O( N M \poly(d\log(\delta^{-1}\Delta)))$. The remaining steps from Line~\ref{alg line:remaining begins} to Line~\ref{alg line:remaining ends} are easy to implement after the stream ends.
\paragraph{Success Probability} The algorithm fails if and only if one of the instances of sparse recovery fails, or it fails to find an $i\in [w]$ satisfying the condition of Line~\ref{alg line:find i} for some $y\in \phi(X)$. 
For the former, the probability that one instance of sparse recovery fails can be reduced to $\delta/(2 N w + 1)$, at the cost of increasing the space by factor of $O(\log ( N w)) = O(\log  N)$. Therefore, applying union bound, the probability that one of the $(2 N w + 1)$ instances fails is at most $\delta$.

As for the latter, we first consider a fixed $y\in \phi(X)$, and $i\in [w]$, the probability of the collision of $h_i(y)=h_i(y')$ for some $y'\in \phi(X)\setminus \{y\}$ is 
\begin{equation*}
    \Pr[h_i(y)=h_i(y')]\le \frac{1}{2N}.
\end{equation*}
By union bound over all $y'\in \phi(X)\setminus\{y\}$, we have 
\begin{equation*}
    \Pr[\exists y'\in \phi(X)\setminus\{y\}, h_i(y)=h_i(y')]\le \frac{|\phi(X)|-1}{2N}\le \frac{1}{2}
\end{equation*}
Recall that $h_1,\dots,h_w$ are independent hash function, then the probability of that such event happens for every $i\in[w]$ simultaneously is at most $(1/2)^w\le \delta/N$. As a result, we can find a desired $i\in[w]$ for $y\in \phi(X)$ that satisfies the condition of Line~\ref{alg line:find i} with probability at least $1-\delta/N$. Applying union bound again, we have that with probability at least $1-\delta$, we can find a desired $i\in[w]$ for every $y\in\phi(X)$ simultaneously. Overall, the success probability of the streaming version of~\Cref{alg:light parts} is at least $1-2\delta$. It remains to scale $\delta$ by a constant. \section{Streaming Lower Bounds Based on INDEX}
\begin{claim}
    \label{claim:lb_m}
    For every integer $m\ge 10$, any algorithm that, with constant probability, computes an $m$-approximation $g>0$ to the optimal objective for \ProblemName{$(1,m)$-Median} of a dataset $X\subseteq [2m^{10}]^2$ presented as an insertion-only point stream must use space $\Omega(m)$.
\end{claim}

We prove~\Cref{claim:lb_m} based on the INDEX problem, as stated below.

\begin{definition}[INDEX problem]
    Alice is given a vector $x\in \{0,1\}^n$, and Bob is given an index $i\in[n]$. Alice can send Bob exactly one message $M$, and Bob needs to use his input $i$ and this message $M$ to compute $x_i\in \{0,1\}$.
\end{definition}
The well-known fact is that achieving constant probability success in the INDEX problem requires a communication complexity of at least $|M|=\Omega(n)$ (see e.g.,~\cite{DBLP:books/daglib/0011756, DBLP:journals/cc/KremerNR99, DBLP:journals/toc/JayramKS08}).

\begin{proof}[Proof of~\Cref{claim:lb_m}]
    We reduce the INDEX problem to \ProblemName{$(1,m)$-Median} problem. Assume that there exists a streaming algorithm $\calA$ that, with constant probability, reports a $m$-approximation to the optimal objective for \ProblemName{$(1,m)$-Median} of a dataset from discrete Euclidean space $[2m^{10}]^2$ presented as an insertion-only stream using space $o(m)$.
Given a vector $x\in\{0,1 \}^{m+1}$, Alice constructs an insertion-only point stream as follows: for every $i\in [m+1]$, if $x_i = 0$, she adds a point $y_i:=(m^5\cdot i,1)$ to the stream; otherwise, she adds a point $y_i:=(m^5\cdot i,2m^2 + 2)$. The resulting dataset is $Y:=\{y_1,\dots, y_{m+1}\}\subseteq [2m^{10}]^2$. Then, Alice runs algorithm $\calA$ on the stream representing $Y$ and sends the internal state of $\calA$ to Bob. 
    
    Suppose Bob is given an index $i\in[m+1]$.
    Bob resumes algorithm $\calA$ and continues to run $\calA$ on a stream consisting of only one insertion, the point $y_i':=(m^5\cdot i,2)$. Hence, algorithm $\calA$ is, in fact, run on a dataset $Y':=\{y_1,\dots, y_{m+1}, y_i'\}$. We can check that the optimal objective of \ProblemName{$(1,m)$-Median} on $Y'$ is $\dist(y_i,y_i')$, which is
    \begin{equation*}
    \dist(y_i,y_i') = \begin{cases}
    1 & \text{if } x_i = 0\\
    2m^2 & \text{if } x_i = 1
    \end{cases}.
    \end{equation*}
    Therefore, when querying $\calA$, with constant probability, $\calA$ will return a value $g>0$ such that $g\le m$ if $x_i=0$ and $g\ge 2m$ otherwise. Such a separation allows Bob to determine the value of $x_i$. This provides a communication protocol for the INDEX problem with a communication complexity equal to the space complexity of $\calA$, which is $o(m)$. This contradicts the $\Omega(m)$ lower bound for the INDEX problem, concluding the proof.
\end{proof}
 \end{appendices}
\end{document}